\documentclass[onecolumn,nofootinbib,thightenlines,notitlepage,tightenlines,longbibliography,superscriptaddress,11pt]{revtex4-1} 
\newcommand{\pagenumbaa}{1}
\usepackage{graphicx}
\usepackage{amssymb}
\usepackage{amstext}
\usepackage{enumerate}

\usepackage[usenames,dvipsnames]{xcolor}
\usepackage{tikz}
\usetikzlibrary{chains}
\usetikzlibrary{fit}
\usepackage{pgflibraryarrows}		
\usepackage{pgflibrarysnakes}		

\usepackage[caption=false]{subfig}

\usepackage{lmodern}
\usepackage{epsfig}
\usetikzlibrary{shapes.symbols,patterns} 
\usepackage{pgfplots}

\usepackage{pgfplotstable}
\usepackage{dsfont}
\usepackage{amsthm}
\usepackage{amsmath}
\usepackage{amsfonts}
\usepackage{amssymb,amstext}
\usepackage{bbm} 
\usepackage{dsfont}

\usepackage[colorlinks=true,urlcolor=blue, hyperindex, breaklinks=true] {hyperref}
\usepackage{todonotes}

\usepackage{footnote}

\theoremstyle{plain}
\newtheorem{mythm}{Theorem}[section]
\newtheorem{myprop}[mythm]{Proposition}
\newtheorem{mycor}[mythm]{Corollary}

\theoremstyle{definition}
\newtheorem{mydef}[mythm]{Definition}
\newtheorem{myex}[mythm]{Example}
\newtheorem{myremark}[mythm]{Remark}

\newcommand{\prlsection}[1]{{\bf{#1}.} }

\newcommand{\ket}[1]{| #1  \rangle}
\newcommand{\bra}[1]{\langle #1 |}
\newcommand{\norm}[1]{\left\lVert#1\right\rVert}

\newcommand{\braket}[2]{|#1\rangle \langle #2 |}
\newcommand{\ketbra}[2]{\langle#1 | #2 \rangle}

\def\Hb{\ensuremath{H}_b}
\newcommand{\W}{\mathsf{W}} 
\newcommand{\V}{\mathsf{V}} 
\newcommand{\Rc}{\mathsf{R}} 

\def\setC{\mathsf{c}}

\newcommand{\Hh}[1]{H\!\left({#1}\right)} 
\newcommand{\Hc}[2]{H\!\left({#1}\!\left|{#2}\right.\right)} 
\newcommand{\I}[2]{I\!\left({#1};{#2}\right)} 
\newcommand{\Ic}[3]{I\!\left({#1};{#2}\!\left|{#3} \right. \right)} 

\newcommand{\Prob}[1]{\,{\mathbb P} \!\left[#1\right]} %

\newcommand{\Bernoulli}[1]{\textnormal{Bernoulli}\left(#1\right)}  

\newcommand{\F}[1]{F({#1})} 

\newcommand{\Tr}[1]{{\rm tr}\!\left[{#1}\right]} 
\newcommand{\Trp}[2]{{\rm tr}_{{#1}}\!\left[\,{#2}\right]} 
\DeclareMathOperator{\Trs}{tr}
\DeclareMathOperator{\enc}{enc}
\DeclareMathOperator{\dec}{dec}

\DeclareMathOperator{\st}{s.t.}

\makeatletter \renewcommand{\todo}[2][]{\tikzexternaldisable\@todo[#1]{#2}\tikzexternalenable} \makeatother

\newcommand{\BSC}{\textnormal{BSC}}
\newcommand{\BEC}{\textnormal{BEC}}

\newcommand{\brackett}[1]{ \left \lbrace \text{#1}\right \rbrace }

\allowdisplaybreaks

\newcommand{\markovDavid}{\small{\mbox{$-\hspace{-1.3mm} \circ \hspace{-1.3mm}-$}}}

\usepackage{makecell}

\long\def\symbolfootnote[#1]#2{\begingroup\def\thefootnote{\fnsymbol{footnote}}\footnote[#1]{#2}\endgroup}

\begin{document}

\title{Alignment of Polarized Sets}

 \author{Joseph M.\ Renes}
 \email[]{$\brackett{renes,\,suttedav}$@phys.ethz.ch}
 \affiliation{Institute for Theoretical Physics, ETH Zurich, Switzerland}
  
  \author{David Sutter}
 \email[]{$\brackett{renes,\,suttedav}$@phys.ethz.ch}
 \affiliation{Institute for Theoretical Physics, ETH Zurich, Switzerland}

  \author{S.\ Hamed Hassani }
 \email[]{hamed@inf.ethz.ch}
 \affiliation{Department of Computer Science, ETH Zurich, Switzerland}


\begin{abstract}
Ar{\i}kan's polar coding technique is based on the idea of synthesizing $n$ channels from the $n$ instances of the physical channel by a simple linear encoding transformation. Each synthesized channel corresponds to a particular input to the encoder. For large $n$, the synthesized channels become either essentially noiseless or almost perfectly noisy, but in total carry as much information as the original $n$ channels. Capacity can therefore be achieved by transmitting messages over the essentially noiseless synthesized channels. 

Unfortunately, the set of inputs corresponding to reliable synthesized channels is poorly understood, in particular how the set depends on the underlying physical channel. In this work, we present two analytic conditions sufficient to determine if the reliable inputs corresponding to different discrete memoryless channels are \emph{aligned} or not, i.e.\ if one set is contained in the other. 
Understanding the alignment of the polarized sets is important as it is directly related to universality properties of the induced polar codes, which are essential in particular for network coding problems. We demonstrate the performance of our conditions on a few examples for wiretap and broadcast channels. 
Finally we show that these conditions imply that the simple quantum polar coding scheme of Renes \emph{et al.} [Phys.\ Rev.\ Lett.\ 109, 050504 (2012)] requires entanglement assistance for general channels, but also show such assistance to be unnecessary in many cases of interest. 
\end{abstract}

 \maketitle

 \setcounter{page}{\pagenumbaa}  
 \thispagestyle{plain}

\section{Introduction}
In Ar{\i}kan's celebrated \emph{polarization phenomenon}~\cite{arikan09}, applying a specific linear transformation called the \emph{polar transform} to $n$ instances of a binary-input output-symmetric discrete memoryless channel (DMC) $\W$ induces $n$ synthesized channels which become either ideal or useless channels as $n$ grows large. More precisely, when assigned with an index, the $n$ induced synthesized channels can be classified into two categories, defining two index sets: the set $\mathcal{D}(\W)$ of indices corresponding to good channels and the set $\mathcal{R}(\W)$ of indices that belong to bad channels. 
Polarization is the property that the sizes of these sets satisfy $\lim_{n\to \infty} \tfrac{1}{n}|\mathcal{D}(\W)| = I(\W)$ and $\lim_{n\to \infty} \tfrac{1}{n}|\mathcal{R}(\W)| = 1-I(\W)$, and this ensures that polar codes are capacity achieving \cite{arikan09}.

 However, the structure of $\mathcal{D}(\W)$ and $\mathcal{R}(\W)$ is poorly understood. In particular, the dependency on $\W$ is difficult to analyze in general. 
 For $\V$ a binary-input output-symmetric DMC different from $\W$, it is unclear if $\mathcal{D}(\W)$ and $\mathcal{D}(\V)$ are \emph{aligned} or not, i.e.\ whether $\mathcal{D}(\W) \subseteq \mathcal{D}(\V)$ or $\mathcal{D}(\W) \supseteq \mathcal{D}(\V)$.
 An exception is the case when $\V$ is assumed to be a \emph{degraded} version of $\W$ (cf.\ Definition~\ref{def:MC_LN_DEG}) which implies that $\mathcal{D}(\V) \subseteq \mathcal{D}(\W)$ \cite{korada_phd}. 
The methods introduced in \cite{hassani09} can be used to detect nonalignment of $\mathcal{D}(\W)$ and $\mathcal{D}(\V)$, but not their alignment.  
 
Understanding the structure (and the relation) of the polarized sets $\mathcal{D}(\W)$ and $\mathcal{D}(\V)$ is important in several respects. First, this is directly linked to the universality of polar codes, if one fixed code can be used for reliable communication over each member of a given class of channels $\mathcal{W}$.  Universal codes are important in different coding scenarios, for instance when the statistics of the actual channel are not known precisely. Second, several different channels are simultaneously involved in network coding tasks such as wiretap or broadcast channels, and alignment is helpful in designing efficient polar coding schemes. Third, knowledge of the structure and relation of polarized sets can be helpful in other aspects of polar coding, e.g.\ in the  construction of polar codes (see \cite[Chap.~5]{hassani_phd}).

Polar coding with successive cancellation (SC) decoding is not universal in general \cite{hassani09}. However, universality holds for certain classes of channels with a specific ordering, such as less noisy comparable channels (cf.\ Definition~\ref{def:MC_LN_DEG}) as explained in Proposition~\ref{prop:LN}. There has been recent progress in slightly modifying standard polar codes such that they become universal, however at the cost of larger blocklengths \cite{hassani13,sasogluLele13}.
Therefore it is of interest to have a computationally efficient way to determine if for a given class of channels $\mathcal{W}$ standard polar codes using SC decoding are universal on $\mathcal{W}$ or not.

Recently, the polarization phenomenon has been used to construct efficient codes, quantum polar codes, for transmitting quantum information. These codes inherit several desirable features of (classical) polar codes. In particular, quantum polar codes achieve high rates while allowing for an efficient encoding and decoding \cite{renes12,wilde_polar_2011}. An important open question regards the necessity of \emph{preshared entanglement}: Specifically, whether the coding scheme requires the sender and receiver to share a nonzero amount of maximally entangled states before the protocol begins. 

\vspace{2mm}
\prlsection{Contributions}
In this article, we introduce a condition for alignment (Theorem~\ref{thm:UB}) and a condition for nonalignment (Theorem~\ref{thm:LB}) of two arbitrary binary input symmetric channels. Applied to several examples of interest, we show that these conditions are sometimes close in the sense that it can be conclusively determined if there is an alignment of the polarized sets or not. 

Since aligned polarized sets imply that the corresponding polar codes are universal with SC decoding, our conditions can be used to determine if for a given set of DMCs polar codes are universal or not. We also show how alignment leads to simple polar coding schemes for a range of non-degradable wiretap and broadcast channels. 

In addition, we show that the two conditions can be used to determine whether quantum polar codes require entanglement assistance or not. We provide examples of quantum channels where no preshared entanglement is needed (e.g., a low-noise BB84 channel) and examples where entanglement assistance provably is needed (e.g., a high-noise depolarizing channel).

\vspace{2mm}

\prlsection{Structure} Section~\ref{sec:prelim} introduces basic concepts of polar codes and provides some background on wiretap and broadcast channel coding. In Section~\ref{sec:alignmentBounds} we present and prove the main results which are two conditions that can be used to analyze the alignment of polarized sets for arbitrary DMCs. Section~\ref{sec:applications} discusses a few applications of the two conditions. In particular we cover a BSC/BEC pair, BSC-BEC wiretap channels and a BSC-BEC broadcast channel. Section~\ref{sec:entanglement} shows how ideas developed in the previous sections can be used to answer the question if quantum polar codes need entanglement assistance or not. We conclude in Section~\ref{sec:conclusion} with a summary and potential subjects of further research. 

\vspace{2mm}

\prlsection{Notation}
Let $[k]:=\left \lbrace 1,\ldots,k \right \rbrace$ for $k\in \mathbb{Z}^+$. For $x \in \mathbb{Z}_2^k$ and $\mathcal{I}\subseteq [k] $ we have $x[\mathcal{I}]=[x_i:i\in \mathcal{I}]$, $x^i=[x_1,\ldots,x_i]$ and $x_j^i=[x_j,\ldots,x_i]$ for $j\leq i$. For two sets $\mathcal{A},\mathcal{B}\subseteq [n]$  we write $\mathcal{A}$ $\scriptsize{\overset{\boldsymbol{\cdot}}{\subseteq}}$ $\mathcal{B}$ meaning that $\mathcal{A}$ is essentially contained in $\mathcal{B}$ or more precisely $|\mathcal{A}\backslash \mathcal{B}|=o(n)$. The complement of as set $\mathcal{A} \subseteq [n]$ is denoted by $\bar{\mathcal{A}}:=[n]\backslash \mathcal{A}$. All logarithms in this article are with respect to the basis $2$. For $\alpha \in [0,1]$, $\Hb(\alpha):=-\alpha \log \alpha - (1-\alpha) \log(1-\alpha) $ denotes the binary entropy function. We denote the Bhattacharyya parameter of a binary-input discrete memoryless channel $\W:\{0,1\} \to \mathcal{Y}$ by $Z(\W):= \sum_{y \in \mathcal{Y}} \sqrt{\W(y|0) \W(y|1)} \in [0,1]$. For some binary string $b \in \{0,1\}^k$ we denote its binary complement by $\bar b$. The logical \emph{and} is denoted by $\wedge$ and the logical \emph{or} by $\vee$. The binary symmetric channel with transition probability $\alpha \in [0,\tfrac{1}{2}]$ is 
abbreviated by $\BSC(\alpha)$ and the binary erasure channel with erasure probability $\beta \in [0,1]$ is denoted by $\BEC(\beta)$. The space of all Hermitian operators in a finite dimensional Hilbert space $\mathcal{H}$ is denoted by $\mathrm{H}$. We denote the set of density operators on a Hilbert space $\mathcal{H}$ by $\mathcal{D}(\mathcal{H}):= \{\rho \in \mathrm{H}: \rho \geq 0,  \Tr{\rho} = 1\}$. For a density operator $\rho \in \mathcal{D}(\mathcal{H})$ we define its von Neumann entropy by $H(\rho):=-\Tr{\rho \log \rho}$. The space of trace class operators acting on some Hilbert space $\mathcal{H}$ is denoted by $\mathcal{S}(\mathcal{H})$.
The Pauli matrices are denoted by $\sigma_X, \sigma_Y$ and $\sigma_Z$. For a matrix $A \in \mathbb{C}^{m \times n}$ the trace norm is defined as $\norm{A}_{\Trs}:=\mathrm{tr}[\sqrt{A^{\dagger}A}]$. For two maps $\Phi: A \to B$ and $\Theta:B\to C$ the map $\Theta \circ \Phi: A\to C $ denotes the concatenation of $\Phi$ with $\Theta$.
\section{Preliminaries} \label{sec:prelim}
Given a binary-input output-symmetric DMC $\W: \{ 0,1 \} \to \mathcal{Y}$, following \cite{arikan09}  we define a \emph{channel splitting} map $(\W,\W) \to (\W_0,\W_1)$ where the synthesized channels $\W_0: \{0,1 \}\to \mathcal{Y}^2$ and $\W_1:\{0,1\}\to \{0,1\}\times \mathcal{Y}^2$ are given by 
\begin{alignat}{2}
&\W_0(y_1,y_2|u_1) &&=\sum_{u_2 \in\{0,1\}} \frac{1}{2} \W(y_1|u_1 \oplus u_2) \W(y_2|u_2) \quad \textnormal{and} \\
&\W_1(y_1,y_2,u_1|u_2) &&= \frac{1}{2} \W(y_1|u_1 \oplus u_2) \W(y_2|u_2),
\end{alignat}
where $u_1,u_2$ are (for symmetric channels) assumed to be i.i.d.\ Bernoulli$(\tfrac{1}{2})$ distributed. The channel splitting map outputs two synthesized channels where $\W_0$ is more noisy and $\W_1$ more reliable than the original channel $\W$. By applying the transform $k = \log n$ times we get $n$ synthesized channels such that in the limit $n\to \infty$ essentially all synthesized channels are either almost noiseless or very noisy \cite{arikan09}. A recursive application of the rate splitting can be visualized in a \emph{polarization tree} that defines the notation of the synthesized channels (cf.\ Figure~\ref{fig:tree}).
\begin{figure}[!htb]
\centering
\def \x{0.7}
\def \y{0.8}
\def \la{0.25}
\def \l{4}

\begin{tikzpicture}[scale=1,auto, node distance=1cm,>=latex']
	

\node at (0,0) {$\W$};
\node at (-2*\x,-\y) {$\W_0$};
\node at (2*\x,-\y) {$\W_1$};
\node at (-3*\x,-2*\y) {$\W_{00}$};
\node at (-\x,-2*\y) {$\W_{01}$};
\node at (1*\x,-2*\y) {$\W_{10}$};
\node at (3*\x,-2*\y) {$\W_{11}$};
\node[rotate=70] at (-3*\x-\la,-2*\y-2*\la) {$\ldots$};
\node[rotate=110] at (3*\x+\la,-2*\y-2*\la) {$\ldots$};
\node[rotate=90] at (-\x,-2*\y-2*\la) {$\ldots$};
\node[rotate=90] at (\x,-2*\y-2*\la) {$\ldots$};

\draw[] (-\la,-0.5*\la) -- (-2*\x+\la,-\y+0.75*\la);
\draw[] (\la,-0.5*\la) -- (2*\x-\la,-\y+0.75*\la);

\draw[] (-2*\x-\la,-\y-0.5*\la) -- (-3*\x+0.5*\la,-2*\y+0.75*\la);
\draw[] (-2*\x+\la,-\y-0.5*\la) -- (-1*\x-0.5*\la,-2*\y+0.75*\la);

\draw[] (2*\x-\la,-\y-0.5*\la) -- (1*\x+0.5*\la,-2*\y+0.75*\la);
\draw[] (2*\x+\la,-\y-0.5*\la) -- (3*\x-0.5*\la,-2*\y+0.75*\la);

\node[] at (\l,0) {level $0$};
\node[] at (\l,-\y) {level $1$};
\node[] at (\l,-2*\y) {level $2$};
\end{tikzpicture}
\caption{Polarization tree up to level $2$.}
\label{fig:tree}
\end{figure}
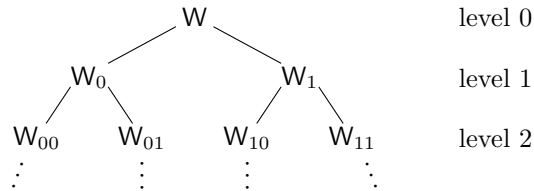 

Let $X^n$ be a vector with i.i.d.\ Bernoulli($p$) distributed entries for $p\in[0,1]$ and  $n=2^k$ with $k\in \mathbb{N}$. Then, define $U^n = G_n X^n$, where $G_n:=\left(\begin{smallmatrix}
1 & 1\\
0 & 1
\end{smallmatrix} \right)^{\otimes  \log n}$ denotes the polarization (or polar) transform. Furthermore, let $Y^n = \W^n X^n$, where $\W^n$ denotes $n$ independent uses of a DMC $\W:\mathcal{X}\to \mathcal{Y}$ and let $Z^n=\V^n X^n$, where $\V:\mathcal{X} \to \mathcal{Z}$ denotes another DMC. For any $\varepsilon \in (0,1)$ we consider the four sets
\begin{subequations}  \label{eq:polarizedSets}
\begin{align}
\mathcal{D}_{\varepsilon}^n(\W)&:= \left \lbrace i \in[n]: Z(\W_{b(i-1)})\leq \varepsilon \right \rbrace \label{eq:DsetW}\\
\mathcal{R}_{\varepsilon}^n(\W)&:= \left \lbrace i \in[n]: Z(\W_{b(i-1)})\geq 1-\varepsilon \right \rbrace \label{eq:RsetW}\\
\mathcal{D}_{\varepsilon}^n(\V)&:= \left \lbrace i \in[n]: Z(\V_{b(i-1)})\leq \varepsilon \right \rbrace\label{eq:DsetV}\\
\mathcal{R}_{\varepsilon}^n(\V)&:= \left \lbrace i \in[n]: Z(\V_{b(i-1)})\geq 1-\varepsilon \right \rbrace,\label{eq:RsetV}
\end{align}
\end{subequations}
where $b(i)$ for $i \in [n]$ denotes the binary representation of the integer $i$ with $\log n$ bits. The sets $\mathcal{D}_{\varepsilon}^n(\W)$ and $\mathcal{D}_{\varepsilon}^n(\V)$ define a polar code for $\W$ respectively $\V$ that is reliable using SC decoding. 
Within this article the parameter $\varepsilon \in (0,1)$ can be arbitrary. As discussed in \cite{arikan09} the error probability of the polar codes for $\W$ and $\V$ will decay faster for small $\varepsilon$. Therefore this parameter should be chosen as small as possible. As a result, for most applications it is convenient to assume that $\varepsilon=O(2^{-n^{\nu}})$ for some $\nu < \tfrac{1}{2}$. We note that in general for an arbitrary DMC $\W$ and $\varepsilon \in (0,\frac{1}{2})$ we have $\overline{\mathcal{D}_{\varepsilon}^n(\W)}=\mathcal{R}_{1-\varepsilon}^n(\W)\supsetneq \mathcal{R}_{\varepsilon}^n(\W)$.

Recall that we call two sets, e.g., $\mathcal{D}_{\varepsilon}^n(\W)$ and $\mathcal{D}_{\varepsilon}^n(\V)$ being aligned if $\mathcal{D}_{\varepsilon}^n(\W) \subseteq \mathcal{D}_{\varepsilon}^n(\V)$ or $\mathcal{D}_{\varepsilon}^n(\W) \supseteq \mathcal{D}_{\varepsilon}^n(\V)$. We say that these two sets are \emph{essentially aligned} if $\mathcal{D}_{\varepsilon}^n(\W)$ $\scriptsize{\overset{\boldsymbol{\cdot}}{\subseteq}}$ $\mathcal{D}_{\varepsilon}^n(\V)$ or $\mathcal{D}_{\varepsilon}^n(\W)$ $\scriptsize{\overset{\boldsymbol{\cdot}}{\supseteq}}$ $\mathcal{D}_{\varepsilon}^n(\V)$. 
 We next define three standard orderings between two DMCs for which some results about the alignment of the sets \eqref{eq:polarizedSets} are available.


\begin{mydef}[more capable, less noisy, degradable] \label{def:MC_LN_DEG}
Let $\W:\mathcal{X} \to \mathcal{Y}$ and $\V:\mathcal{X}\to \mathcal{Z}$ be two DMCs then
\begin{itemize}
\item $\W$ is \emph{more capable} than $\V$ if $\I{X}{Y}\geq \I{X}{Z}$ for all distributions $P_X $.
\item $\W$ is \emph{less noisy} than $\V$ if $\I{U}{Y}\geq \I{U}{Z}$ for all distributions $P_{U,X}$ where $U$ has finite support and $U\markovDavid X \markovDavid (Y,Z)$ form a Markov chain.
\item $\V$ is said to be a (stochastically) \emph{degraded} version of $\W$ if there exists a channel $\mathsf{T}:\mathcal{Y}\to \mathcal{Z}$ such that $\V(z|x)=\sum_{y \in \mathcal{Y}} \W(y|x)\mathsf{T}(z|y)$ for all $x \in \mathcal{X}$, $z\in \mathcal{Z}$.
\end{itemize}
\end{mydef}
Note that the relation between these three classes of channels is well understood. Every channel that is degradable is also less noisy and every channel that is less noisy is also more capable. The converse is not true, i.e., there exist channels that are more capable but not less noisy and channels that are less noisy but not degradable \cite{korner75}.

\begin{myprop}[Alignment for less noisy channels {\cite[Thm.\ 10]{sutter13_2} and \cite[App.\ A]{sasogluLele13}}] \label{prop:LN}
If $\W$ is less noisy than $\V$, then for every $\varepsilon \in (0,1)$, $\mathcal{D}_{\varepsilon}^n(\W)$ $\scriptsize{\overset{\boldsymbol{\cdot}}{\supseteq}}$ $\mathcal{D}_{\varepsilon}^n(\V)$ and $\mathcal{R}_{\varepsilon}^n(\W)$ $\scriptsize{\overset{\boldsymbol{\cdot}}{\subseteq}}$ $\mathcal{R}_{\varepsilon}^n(\V)$.\footnote{We note that for finite values of $n$ it makes a difference if the polarized sets are defined with respect to the Bhattacharyya parameter (as done within this work) or the entropy (as done in \cite{sutter13_2}). As a result we get in \cite{sutter13_2} a proper alignment whereas with the definition used in this article we obtain an essential alignment result.}
\end{myprop}

If $\W$ is more capable than $\V$, in general $|\mathcal{D}_{\varepsilon}^n(\V) \backslash \mathcal{D}_{\varepsilon}^n(\W)| = \Omega(n)$, i.e., the two sets $\mathcal{D}_{\varepsilon}^n(\W)$ and $\mathcal{D}_{\varepsilon}^n(\V)$ are not aligned -- not even essentially \cite{hassani09}. However, when considering a particular input distribution an alignment result for two more capable channels has been proven recently.

\begin{myprop}[{\cite[Cor.\ 9]{sutter13_2}}] \label{prop:MC}
Let $\W$ be more capable than $\V$ and consider an input distribution $P_X$ that it maximizes $\I{X}{Y}-\I{X}{Z}$. Then for $\epsilon =O(2^{-n^{\nu}})$ with $\nu <\tfrac{1}{2}$ we have $\mathcal{D}_{\varepsilon}^n(\V)$ $\scriptsize{\overset{\boldsymbol{\cdot}}{\subseteq}}$ $\mathcal{D}_{\varepsilon}^n(\W)$.
\end{myprop}

\subsection{Wiretap channels}
In a wiretap channel coding scenario, Alice would like to transmit a message $M^k \in \mathcal{M}^k$ privately to Bob. The messages can be distributed according to some arbitrary distribution $P_{M^k}$. To do so, she first encodes the message by computing $X^n=\enc(M^k)$ for some encoding function $\enc:\mathcal{M}^k \to \mathcal{X}^n$ and then sends $X^n$ over $n$ copies of the discrete memoryless wiretap channel. A wiretap channel consists of a channel that transmits the sent message to Bob, i.e., $Y^n = \W^n X^n$, where $\W: \mathcal{X} \to \mathcal{Y}$ denotes the channel between Alice and Bob. At the same time the sent message is transmitted over a (possibly) different channel $\V:\mathcal{X}\to \mathcal{Z}$ to the eavesdropper, i.e., $Z^n = \V^n X^n$.
Bob next decodes the received message to obtain a guess for Alice's message $\hat M^k = \dec(Y^n)$ for some decoding function $\dec:\mathcal{Y}^n \to \mathcal{M}^k$. The private channel coding scheme should be reliable, i.e.,\ satisfy the reliability condition
\begin{equation}
 \lim \limits_{k\to \infty}\Prob{M^k\ne \hat M^k}=0 \label{eq:reliability}
\end{equation}
and secure, i.e.,\ satisfy the (strong) secrecy condition
\begin{equation}
 \lim_{k\rightarrow \infty}\norm{P_{M^k,Z^n,C}-P_{M^k}\times P_{Z^n,C}}_1=0.
\label{eq:secrecy}
\end{equation}
The variable $C$ denotes additional information made public by the protocol. 
In the limit $k\to \infty$ the secrecy condition \eqref{eq:secrecy} is  equivalent to the historically older (strong) secrecy condition $\lim_{k \to \infty} \I{M^k}{Z^n,C}=0$. 
The highest rate fulfilling \eqref{eq:reliability} and \eqref{eq:secrecy} is called the \emph{secrecy capacity}. 
Csisz\'ar and K\"orner showed \cite[Corollary 2]{csiszar78} that there exists a single-letter formula for the secrecy capacity.\footnote{Csisz\'ar and K\"orner considered a weaker security criterion that was shown later to be insufficient. Maurer and Wolf showed that the single-letter formula remains valid considering the (strong) secrecy condition \eqref{eq:secrecy} \cite{maurer00}.}
\begin{mythm}[Secrecy capacity \cite{csiszar78}] \label{thm:Cs}
For an arbitrary discrete memoryless wiretap channel as introduced above the secrecy capacity is given by
\begin{equation} \label{eq:secrecyCapacity}
C_s(\W,\V)=\left \lbrace \begin{array}{rl}
\max \limits_{P_{U,X}} & \Hc{U}{Z}-\Hc{U}{Y}\\
\mathrm{s.t.} & U \markovDavid X \markovDavid (Y,Z), \\
&|\mathcal{U}|\leq |\mathcal{X}|. 
\end{array} \right.
\end{equation}
\end{mythm}
This expression can be simplified using additional assumptions about the wiretap channel.
\begin{mycor}[Secrecy capacity for more capable wiretap channels \cite{korner77}] \label{cor:CsMoreCapable}
If $\W$ is more capable than $\V$,
\begin{equation}
C_s(\W,\V)= \max_{P_X} \Hc{X}{Z}-\Hc{X}{Y}. \label{eq:capacityCapable}
\end{equation}
\end{mycor}

In \cite{vardy11}, Mahdavifar and Vardy showed how to use polar codes to efficiently achieve the secrecy capacity for degradable wiretap channels. Their secrecy criterion was a weaker form of the one given in \eqref{eq:secrecy}. In \cite{sasoglu13}, it has been shown how to use polar codes to achieve the secrecy capacity for degradable wiretap channels with respect to the strong secrecy condition \eqref{eq:secrecy}. In \cite{sutter13_3}, two of us reported a concatenated protocol based on polar codes that is strongly secure, efficient and achieves the secrecy capacity, whose code construction however might be difficult. Recently, it has been shown how to achieve the secrecy capacity of a wiretap channel with polar codes using the chaining technique introduced in \cite{hassani13} to ensure an alignment of the polarized sets in case where the wiretap channel is not less noisy \cite{uluklus14,barg14}.

\subsection{Broadcast channels}
The discrete memoryless broadcast channel with $k$ broadcast receivers consists of a discrete input alphabet $\mathcal{X}$, discrete output alphabets $\mathcal{Y}_i$ for $i \in [k]$, and a conditional distribution $P_{Y_1,Y_2,\ldots,Y_m | X}(y_1,y_2,\ldots,y_m|x)$ where $x \in \mathcal{X}$ and $y_i \in \mathcal{Y}_i$.
In this article we consider the broadcast channel problem that consists of a single source transmitting two independent messages to two receivers through a single discrete, memoryless, broadcast channel. The private-message capacity region is known if the channel structure is deterministic, degraded, less-noisy, or more-capable \cite{elgamal12}. For the general case the (private-message) capacity region is unknown however there exist different inner and outer bounds. One possible inner bound bound that will be important in this article is the one that is achieved with superposition coding.
\begin{mythm}[Superposition coding inner bound \cite{cover72}] \label{thm:superposition}
The union of rate pairs $(R_1,R_2)$ satisfying
\begin{subequations} 
\begin{align}
R_1&< \Ic{X}{Y_1}{U}\\
R_2 &< \I{U}{Y_2}\\
R_1 + R_2 &< \I{X}{Y_1}
\end{align}
\end{subequations} 
over all $(U,X)$ such that $U \markovDavid X \markovDavid (Y_1,Y_2)$ form a Markov chain is achievable.
\end{mythm}
Note that for degradable discrete memoryless broadcast channels the superposition coding inner bound coincides with the capacity region \cite{gallager74}. Recently in \cite{goela13}, it has been shown how to use polar codes to achieve the capacity region for degradable discrete memoryless broadcast channels. The assumption that the broadcast channel is degradable is used to ensure that the polar codes are universal. In \cite{mondelli14}, it has been shown how to achieve the superposition region and more generally the Marton's inner region \cite{marton79}\footnote{Marton's inner region is in general a better bound than the superposition coding inner bound. A nice overview can be found in \cite{elgamal12}.} with polar codes by using the chaining method  to obtain a universal code at the cost of a larger blocklength.

\section{Alignment of Polarized Sets} \label{sec:alignmentBounds}
In this section we will state and prove our main results (Theorems~\ref{thm:LB} and \ref{thm:UB}), which are two sufficient conditions for the sets $\mathcal{D}_{\varepsilon}^n(\W)$ and $\mathcal{D}_{\varepsilon}^n(\V)$ being aligned or being not aligned (not even essentially). The conditions can be applied to arbitrary DMCs $\W$ and $\V$. The first criterion, that is derived in Section~\ref{sec:LB} and can be used to conclude that $\mathcal{D}_{\varepsilon}^n(\W)$ and $\mathcal{D}_{\varepsilon}^n(\V)$ are not aligned, is based on a simple counting argument using the polarization phenomenon. The second criterion derived in Section~\ref{sec:UB} that implies that two polarized sets $\mathcal{D}_{\varepsilon}^n(\W)$ and $\mathcal{D}_{\varepsilon}^n(\V)$ are aligned, is more elaborate and uses a particular property of the polarization transformation together with an uncertainty relation from quantum mechanics for which the (classical) channel has to be embedded into a quantum-mechanical channel as explained in Section~\ref{sec:complementaryChannel}.  

For this reason we have to introduce some basic quantum information theoretic concepts and notations. For a general overview, see~\cite{nielsen_quantum_2000}. 
A binary-input classical-quantum (cq) channel $\W: \{0,1\} \ni x \mapsto \rho_x \in \mathcal{D}(\mathcal{H})$ prepares a quantum state $\rho_x$ at the output, depending on a classical input bit $x$. The analog of the Bhattacharyya parameter for classical channels is the \emph{fidelity} of a cq channel that is defined as $\F{\W}:= \norm{\sqrt{\rho_0} \sqrt{\rho_1}}_{\Trs}$.
The symmetric Holevo information is defined as $I(\W):=H(\tfrac{1}{2}(\rho_0 + \rho_1)) - \tfrac{1}{2}(\Hh{\rho_0}+\Hh{\rho_1})$.
It is straightforward to verify that in case $\W$ is a classical binary-input discrete memoryless channel $\F{\W} =Z(\W)$ and that the symmetric Holevo information coincides with the symmetric mutual information. The polarization process for cq channels works similarly as for classical DMCs \cite{wilde_polar_2011}. We can define a channel splitting map $(\W,\W) \to (\W_0,\W_1)$, where the synthesized channels $\W_0:\{0,1\} \to \mathcal{D}(\mathcal{H} \otimes \mathcal{H})$ and $\W_1: \{0,1\} \to \{0,1\} \otimes  \mathcal{D}(\mathcal{H} \otimes \mathcal{H})$ are properly defined in \cite{wilde_polar_2011}.

\begin{myprop} \label{prop:Hamed}
For two binary-input cq channels $\W$ and $\V$ such that $\F{\W}+\F{\V} \leq 1$ we have $\F{\W_{0}} + \F{\V_1} \leq 1$ and $\F{\W_{1}} + \F{\V_0} \leq 1$.
\end{myprop}

\begin{proof}
Recall that according to \cite[Prop.\ 9]{wilde_polar_2011} for every binary-input cq channel $\W$, $\F{\W_0} \leq 2 \F{\W} - \F{\W}^2$ and $\F{\W_1}=\F{\W}^2$. Using these two relations gives
\begin{subequations} 
\begin{align}
\F{\W_0} + \F{\V_1} &\leq 2 \F{\W} - \F{\W}^2 + \F{\V}^2 \\
& \leq 2 \F{\W} - \F{\W}^2 + (1-\F{\W})^2 \label{eq:opti}\\
& \leq 1,
\end{align}
\end{subequations} 
where inequality \eqref{eq:opti} uses the assumption $\F{\W}+\F{\V} \leq 1$. The proof of the second statement of the proposition follows by swapping $\W$ and $\V$.
\end{proof}
Applying Proposition~\ref{prop:Hamed} recursively to the polarization tree given in Figure~\ref{fig:tree} proves the following corollary.
\begin{mycor} \label{cor:Hamed}
Consider two binary-input cq channels $\W$ and $\V$ such that $\F{\W}+\F{\V} \leq 1$. Then $\F{\W_b} + \F{\V_{\bar b}} \leq 1$ for all $b \in \{0,1\}^{\log n}$.\footnote{Recall that for some binary string $b \in \{0,1\}^k$, we denote its complement by $\bar b$.}
\end{mycor}

\begin{myremark} \label{rmk:doubleBadUB}
For two binary-input discrete memoryless channels $\W$ and $\V$ such that $1-I(\W)+ I(\V) \geq 1$, we have $|\mathcal{R}_{\varepsilon}^n(\W) \cap \mathcal{D}_{\varepsilon}^n(\V)| = \Omega(n)$. 
\end{myremark}
Remark~\ref{rmk:doubleBadUB} follows by the polarization phenomenon \cite{arikan09,arikan10} which ensures that $n(1-I(\W))= | \mathcal{R}_{\varepsilon}^n(\W) | + o(n)$ and $n I(\V)= | \mathcal{D}_{\varepsilon}^n(\V) | + o(n)$. By replacing $\W$ and $\V$ the same argument shows that  $I(\W)+1- I(\V) \geq 1$ implies $|\mathcal{D}_{\varepsilon}^n(\W) \cap \mathcal{R}_{\varepsilon}^n(\V)|  = \Omega(n)$. 

\subsection{Sufficient conditions for nonalignment} \label{sec:LB}
Let $\W$ and $\V$ be two binary-input discrete memoryless channels. Remark~\ref{rmk:doubleBadUB} can be used to derive sufficient conditions for  $|\mathcal{R}_{\varepsilon}^n(\W) \cap \mathcal{D}_{\varepsilon}^n(\V)|= \Omega(n)$ and $|\mathcal{D}_{\varepsilon}^n(\W) \cap \mathcal{R}_{\varepsilon}^n(\V)| = \Omega(n)$. In the following we will state the conditions for  $|\mathcal{R}_{\varepsilon}^n(\W) \cap \mathcal{D}_{\varepsilon}^n(\V)| =\Omega(n)$ as the conditions for  $|\mathcal{D}_{\varepsilon}^n(\W) \cap \mathcal{R}_{\varepsilon}^n(\V)| = \Omega(n)$ follow by the same argument by swapping $\W$ and $\V$. We can derive conditions on every level of the polarization tree. 

\begin{mythm}[Level $k$ condition for no alignment] \label{thm:LB}
Let $k \in \mathbb{N}_0$ and $\varepsilon \in (0,1)$. If $1-I(\W_b) + I(\V_{ b}) \geq 1$ for some $b \in \{0,1\}^k$, then $|\mathcal{R}_{\varepsilon}^n(\W) \cap \mathcal{D}_{\varepsilon}^n(\V)|  = \Omega(n)$.
\end{mythm}
\begin{proof}
The level 0 statement follows directly from Remark~\ref{rmk:doubleBadUB}. Remark~\ref{rmk:doubleBadUB} can be applied at every step of the polarization tree which proves the assertion.
\end{proof}
By definition of the counterpart of a channel given in Section~\ref{sec:complementaryChannel} we have $I(\V)+I(\V^\setC) =1$ for every binary-input discrete memoryless channel $\V$. Thus the level 0 condition for a lower bound on $|\mathcal{R}_{\varepsilon}^n(\W) \cap \mathcal{D}_{\varepsilon}^n(\V)|$, i.e., $1-I(\W) + I(\V) \geq 1$ can be written equivalently as $I(\W)+I(\V^\setC) \leq 1$ which then resembles the level 0 condition given in Theorem~\ref{thm:UB}.

\begin{myremark}[Criterion for nonalignment cannot get worse for higher levels] \label{rmk:LBworse}
Using the identiy $I(\W_0) + I(\W_1) = 2 I(\W)$ we obtain $2 \left(1-I(\W)+I(\V) \right) = 1-I(\W_0) + I(\V_0) + 1-I(\W_1)-I(\V_1)$, which shows that if the conditions that imply no alignment (cf.\ Theorem~\ref{thm:LB}) at level $k$ are satisfied they are also satisfied for all levels $\ell \leq k$.\footnote{The opposite is not true. Oftentimes the criterion for no alignment becomes strictly better by considering higher levels.}
\end{myremark}

\subsection{Counterpart of a channel } \label{sec:complementaryChannel}
In order to prove the sufficient conditions for alignment of the polarized sets given in Theorem~\ref{thm:UB}, we need the concept of a \emph{quantum counterpart} of a DMC. The quantum counterpart is useful because its information tranmission capabilities are directly related to those of the original channel by uncertainty relations. Such counterpart channels were defined generally in \cite[Sec.\ IIA]{renes14} and we give a slightly different presentation here.   

Suppose we are given a binary-input DMC $\W: \{0,1\} \to \mathcal{Y}$ characterized by the transition probabilities $P_{Y|X}(y|x)$ for $x\in \{0,1 \}$ and $y\in \mathcal{Y}$. To the input and output alphabets we may associate orthonormal bases of finite-dimensional vector spaces, which we regard as the state spaces of quantum systems. Let the input alphabet correspond to the basis $\ket{x}^A$ of system $A$ and the output alphabet correspond to the basis $\ket{y}^B$ of system $B$. By defining the quantum states $\varphi_x=\sum_{y\in \mathcal{Y}} P_{Y|X}(y|x) \ket y \bra y^B$, it is always possible to embed $\W$ into a quantum channel as
\begin{equation*}
\W: \ket x\bra x^A \mapsto \varphi_x^B. 
\end{equation*}
Indeed, there are many quantum channels with this action, as we have not specified the mapping for quantum states not diagonal in the basis $\{\ket x\}$. Since we are modelling a classical channel, the output at $B$ should always be a convex combination of the states $\varphi_x^B$, a condition we will take care to enforce in the construction below. 

Once in the quantum setting, we may consider the description of $\W$ in terms of the \emph{Stinespring dilation} (see \cite[Chap.\ 8]{nielsen_quantum_2000}). Let $C$ and $D$ be additional quantum systems isomorphic to $B$ and define the states $\ket{\varphi_x}^{BC}=\sum_{y\in \mathcal{Y}} \sqrt{P_{Y|X}(y|x)} \ket y^B\ket y^C$, which satisfy  $\varphi_x^{B}=\Trp{C}{\ket{\varphi_x}\bra{\varphi_x}^{BC}}$. Then, a Stinespring dilation of $\W$ is the partial isometry $U_\W^{A\rightarrow BCD}$ from $A$ to $B\otimes C \otimes D$ such that 
\begin{align}
U_\W^{A\rightarrow BCD}\ket{x}^A=\ket{\varphi_x}^{BC}  \ket{x}^D.
\end{align}
The action of the channel can be expressed in terms of the dilation as mapping any quantum state $\rho$ to ${\rm tr}_{CD}[U_\W^{A\rightarrow BCD}\rho^A(U_\W^{A\rightarrow BCD})^\dagger]$. The presence of the additional $\ket{x}^D$ ensures that the output states at $B$ are convex combinations of the $\varphi_x$, as required. 

Using $U_{\W}^{A\rightarrow BCD}$ we can define the quantum counterpart to $\W$ as 
\begin{align}
\W^{\setC} : \{0,1 \} \ni x \mapsto \sigma_x^{CD}:=\Trp{B}{U_\W^{A\rightarrow BCD} \ket {\tilde x}\bra {\tilde x}^A\, (U_\W^{A\rightarrow BCD})^\dagger} \in \mathcal{D}(\mathcal{H}) 
\end{align}
for $ \ket{\tilde x}=\tfrac{1}{\sqrt{2}} \sum_{z \in \{0,1 \}} (-1)^{xz} \ket{z}$. These are the same output states defined in \cite[Eq.\ 6]{renes14}. 
The isometry is not unique, but all possible isometries are related by isometries involving the additional systems $C$ and $D$ only, and therefore these isometries do not change the distinguishability of the outputs of the counterpart channel. Up to this freedom, the counterpart channel is essentially unique. 
An equivalent means of defining the counterpart is via the \emph{channel state}. Define the quantum state
\begin{subequations}
\begin{align}
\ket{\psi_\W}&=\tfrac1{\sqrt2}\sum_{z\in\{0,1\}}\ket{z}^A\ket{\varphi_{z}}^{BC}\ket{z}^D\\
&=\tfrac1{\sqrt2}\sum_{x\in\{0,1\}}\ket{\tilde x}^A\ket{\sigma_x}^{BCD},
\end{align}
\end{subequations}
and denote the associated density operator by simply $\psi_\W^{ABCD}$. 
In the second expression we have used $\ket{\sigma_x}^{BCD}=\tfrac1{\sqrt 2}\sum_z (-1)^{xz}\ket{\varphi_z}^{BC}\ket{z}^D$ for the purification $\ket{\sigma_x}^{BCD}$ of $\sigma_x^{CD}$. 
Then the outputs of $\W$ are just $\varphi_z^B=2{\rm tr}_{ACD}[\ket{z}\bra z^A {\psi_{W}}^{ABCD}]$, while the outputs of the counterpart $\W^\setC$ are $\sigma_x^{CD}=2{\rm tr}_{AB}[\ket{\tilde x}\bra{\tilde x}^A {\psi_{W}}^{ABCD}]$.

Although defined completely independently, the counterpart and channel synthesis operations in fact have a particular relation to each other. This relation is the basis of the quantum polar coding technique of \cite{renes12,renes14}. For $n$ systems, consider the channel state
\begin{subequations}
\begin{align}
\ket{\xi_\W}&=\tfrac1{\sqrt{2^n}}\sum_{z^n\in\{0,1\}^n}\ket{z^n}^A\ket{\varphi_{G_n z^n}}^{BC}\ket{G_n z^n}^D\\
&=\tfrac1{\sqrt{2^n}}\sum_{x^n\in\{0,1\}^n}\ket{\tilde x^n}^A\ket{\sigma_{G_n^T x^n}}^{BCD}.
\end{align}
\end{subequations}
The action of $\W_b$ is $z_j\to \frac1{2^{n-1}}\sum_{\bar {z_i}}\ket{z_1^{j-1}}\bra{z_1^{j-1}}^{A_1^{j-1}}\otimes \varphi_{G_nz^n}^B$ for the $j\in [n]$ such that the binary expansion of $j+1$ is $b$, where the summation runs over all $z_k\in\{0,1\}$ for $k\neq j$ \cite{renes14}. Observe that the output is obtained from $\xi_\W$ by projecting the $j$th system of $A$ onto $\ket{z_j}$, tracing out $A_{j+1}^nCD$ but keeping the first $j-1$ systems of $A$. In~\cite{renes12,renes14} it is shown that the polar transform is transposed for the counterpart, which has the effect of reversing the ordering of inputs. That is, the same position $j$ corresponds to $(\W^{\setC})_{\bar b}$, and the discussion subsequent to Equation 25 of \cite{renes14} shows that its action is $x_j\to \frac 1{2^n}\sum_{\bar x_j}\ket{\tilde x_{j+1}^{n}}\bra{\tilde x_{j+1}^{n}}^{A_{j+1}^n}\otimes U^D_{\rm enc}\sigma_{G_n^Tx^n}^{CD}(U^D_{\rm enc})^\dagger$,
where $U_{\rm enc}$ is the polar transform as a unitary operation: $U_{\rm enc}\ket{z^n}=\ket{G_nz^n}$.  Up to this unitary, which is irrelevant for the counterpart channel, this output is obtained from $\xi_\W$ by projecting system $A_j$ onto $\ket{\tilde x_j}$, measuring the subsequent $n-j$ systems of $A$ in the $\ket{\tilde x}$ basis and tracing out $A_1^{j-1}B$. 

On the other hand, the counterpart of $\W_b$ involves the mapping 
\begin{align}
\ket{\tilde x_j}&\to \tfrac1{\sqrt{2^n}}\sum_{z^n\in\{0,1\}^n}(-1)^{x z_j}\ket{z_1^{j-1}}^{A_1^{j-1}}\ket{z_1^{j-1}}^{D_1^{j-1}}\ket{z_j}^{D_j}\ket{z_{j+1}^n}^{D_{j+1}^n}\ket{\varphi_{G_nz^n}}^{BC}
\end{align}
where systems $A_1^{j-1}B$ are the outputs of the original channel and $CD$ are the outputs of the counterpart. 
The output of the counterpart can be obtained from $\xi_\W$ by again projecting $A_j$ onto $\ket{\tilde x_j}$, tracing out $A_{1}^{j-1}B$, but now leaving the remaining $A$ systems untouched rather than measuring them. This shows that $(\W^\setC)_{\bar b}$ is a degraded version of $(\W_b)^\setC$, since we can measure the systems $A_{j+1}^n$ of the latter to obtain the former.

A useful uncertainty relation constrains the fidelities of the two channels:  
\begin{myprop}\label{prop:UR}
Let $\W$ be a binary-input discrete memoryless channel and $\W^\setC$ be its counterpart as defined above. Then for every $b\in \{0,1\}^{\log n}$ we have $\F{\W_b}+\F{(\W^\setC)_{\bar b}} \geq 1$.
\end{myprop}
\begin{proof}
The proof is based on the fact that the fidelity between the outputs of the counterpart channel is actually equal to the trace distance or variational distance $\delta(\W):=\tfrac12\|\varphi_0-\varphi_1\|_1$ between the outputs of the original channel. Known relations between the trace distance and fidelity then yield the uncertainty relation. 

Let us first establish the claim for $b=\emptyset$, i.e.\ the channel and its counterpart. 
Uhlmann's theorem~\cite[Thm.\ 9.4]{nielsen_quantum_2000} gives a convenient means to compute the fidelity:
\begin{align}
F(\W^\setC) &= \max_V |\bra{\tilde 0}(U_\W^{A\rightarrow BCD})^\dagger V^B U_\W^{A\rightarrow BCD}\ket{\tilde 1}^A|,
\end{align}
where the maximization is over all unitaries on the $B$ system. Computing this quantity, we find
\begin{subequations}
\begin{align}
F(\W^\setC) &=\max_V \Big|\big(\tfrac1{\sqrt2}\sum_z \bra{\varphi_z}^{BC}\bra z^D\big)V^B\big(\tfrac1{\sqrt2}\sum_{z'}(-1)^{z'} \ket{\varphi_{z'}}^{BC}\ket {z'}^D\big)\Big|\\
&=\max_V \tfrac12\Big|\sum_z(-1)^z \bra{\varphi_z}V^B\ket{\varphi_z}^{BC}\Big|\\
&=\max_V \tfrac12\Big|\sum_z (-1)^z{\rm tr}[V^B\varphi_z^B] \Big|\\
&=\max_V\tfrac12 \Big| {\rm tr}[V^B(\varphi_0-\varphi_1)]\Big|\\
&=\delta(\W).
\end{align}
\end{subequations}
The bound $F(\W)+\delta(\W)\geq 1$~\cite[Eq.\ 9.110]{nielsen_quantum_2000} gives the uncertainty relation $F(\W)+F(\W^\setC)\geq 1$. 

For the case of synthesized channels, it suffices to use the fact that $(\W^\setC)_{\bar b}$ is a degraded version of $(\W_b)^\setC$, and use the monotonicity of fidelity under quantum operations.
\end{proof}

In the following we explain in detail how to derive the counterpart for three classical DMCs. This will be useful in Section~\ref{sec:applications}. 
\begin{myex}[Counterpart of $\BEC(\beta)$] \label{ex:complementBEC} 
Consider $\W=\BEC(\beta)$ for $\beta \in [0,1]$. The associated isometry $U_\W^{A\rightarrow BC}$ has the action 
\begin{align}
U_\W^{A\rightarrow BCD}\ket{x}^A=\sqrt{1-\beta}\ket{x}^B\ket{?}^C\ket x^D+\sqrt\beta\ket ?^B\ket{x}^C\ket x^D.
\end{align}
Applied to $\ket{\tilde x}$ this gives
\begin{align}
U_\W^{A\rightarrow BCD}\ket{\tilde x}^A=\tfrac1{\sqrt2}\sum_{z \in \{0,1 \}} (-1)^{xz}\left(\sqrt{1-\beta}\ket{z}^B\ket{?}^C\ket z^D+\sqrt\beta\ket ?^B\ket{z}^C\ket z^D\right).
\end{align}
We may simplify the outputs without changing their distinguishability by applying a unitary operator $V^{CD}$ on systems $CD$, described by the action $\ket?\ket{z}\to \ket?\ket z$ and $\ket z\ket z\to \ket z\ket0$. This results in 
\begin{align}
V^{CD}U_\W^{A\rightarrow BCD}\ket{\tilde x}^A=\tfrac1{\sqrt2}\sum_{z \in \{0,1\}} (-1)^{xz}\left(\sqrt{1-\beta}\ket{z}^B\ket{?}^C\ket z^D+\sqrt\beta\ket ?^B\ket{z}^C\ket 0^D\right).
\end{align}
Tracing out $B$ gives the output of the counterpart
\begin{align}
\sigma_x^{CD}=(1-\beta)\ket?\bra?^C\otimes \tfrac12\mathbbm 1^D+\beta \ket{\tilde x}\bra{\tilde x}^C\otimes \ket 0\bra 0^D.
\end{align}
We may also remove system $D$, since $\sigma_x^{CD}$ can be recreated from $\sigma_x^C$: Just create $\tfrac 12\mathbbm 1^D$ if system $C$ is in the state $\ket?$, else create $\ket 0\bra 0^D$.

The map $x\mapsto \sigma_x^C=(1-\beta)\ket ?\bra?^C+\beta\ket{\tilde x}\bra{\tilde x}^C$ is just a BEC with erasure probability $1-\beta$, and thus the counterpart of $\BEC(\beta)$ is simply $\BEC(1-\beta)$. 
\end{myex}


\begin{myex}[Counterpart of $\BSC(\alpha)$] \label{ex:counterpartBSC}
Let $\W=\BSC(\alpha)$ with $\alpha \in [0,\tfrac{1}{2}]$ and to simplify notation let $p_0 := \alpha$ and $p_1 := 1-\alpha$. The action of $U_\W$ is 
\begin{align}
U_\W^{A\rightarrow BCD} \ket{x}^A = \sum_{u \in \{0,1 \} }\sqrt{p_u}\ket{x+u}^B \ket{u}^C\ket{x}^D,
\end{align}
and applied to $\ket{\tilde x}$ gives
\begin{subequations}
\begin{align}
U_\W^{A\rightarrow BCD} \ket{\tilde x}^A &= \tfrac1{\sqrt2}\sum_{u,z \in \{0,1 \} }(-1)^{xz}\sqrt{p_u}\ket{z+u}^B \ket{u}^C\ket{z}^D\\
&=\tfrac1{\sqrt2}\sum_{u,y \in \{0,1 \} }(-1)^{x(y-u)}\sqrt{p_u}\ket{y}^B \ket{u}^C\ket{y-u}^D.
\end{align}
\end{subequations}
Applying the unitary operation $V^{CD}$ specified by $\ket u\ket x\mapsto \ket {u}\ket {u+x}$ does not change the distinguishability of the output states, but simplifies the channel action to 
\begin{subequations}
\begin{align}
V^{CD}U_\W^{A\rightarrow BCD} \ket{\tilde x}^A &=\tfrac1{\sqrt2}\sum_{u,y \in \{0,1 \}}(-1)^{x(y-u)}\sqrt{p_u}\ket{y}^B \ket{u}^C\ket{y}^D\\
&=\ket{\theta_x}^C\otimes \tfrac1{\sqrt2}\sum_{y \in \{0,1 \}}(-1)^{xy}\ket{y}^B\ket{y}^D,
\end{align}
\end{subequations}
where $\ket{\theta_x}=\sum_{u\in\{0,1\}}\sqrt{p_u}(-1)^{xu}\ket{u}$ with $x\in \{0,1\}$. Just as in the BEC example, the $D$ system does not contribute to the distinguishability of $\sigma_x^{CD}$ since now $\sigma_x^{CD}=\ket{\theta_x}\bra{\theta_x}^C\otimes \tfrac12\mathbbm{1}^D$. 
It is straightforward to verify that $Z(\W)=2 \sqrt{\alpha(1-\alpha)}$ and $\F{\W^\setC}=\ketbra{\theta_0}{\theta_1}=1-2\alpha$.
\end{myex}

\begin{myex}[Counterpart of $\BEC(\beta) \circ \BSC(\alpha)$] \label{ex:counterpartBSCBEC}
Consider $\W=\BEC(\beta) \circ \BSC(\alpha)$ for $(\alpha,\beta) \in [0,\tfrac{1}{2}] \times [0,1]$, which is a DMC that consists of a sequence of a $\BSC(\alpha)$ and a $\BEC(\beta)$.
Combining the isometries of Example~\ref{ex:complementBEC} and Example~\ref{ex:counterpartBSC}, in this case we have for $x\in \{0,1\}$
\begin{subequations}
\begin{align}
U_\W^{A\rightarrow BCD}\ket{x}^A &= \!\! \sum_{u \in \{0,1 \} }\sqrt{p_u} \ket{u}^{C_1}\ket x^{D_1}\!\!\left(\sqrt{1-\beta}\ket{x+u}^B\ket?^{C_2}+\sqrt\beta\ket?^B\ket{x+u}^{C_2}\right)\ket{x+u}^{D_2},\\
&\simeq\!\! \sum_{u \in \{0,1 \} }\sqrt{p_u} \ket{u}^{C_1}\ket {u+x}^{D_1}\left(\sqrt{1-\beta}\ket{x+u}^B\ket?^{C_2}+\sqrt\beta\ket?^B\ket{x+u}^{C_2}\right).
\end{align}
\end{subequations}
The second expression is unitarily equivalent to the first, since we can generate $\ket u^{C_1}\ket x^{D_1}\ket {u+x}^{D_2}$ from $\ket u^{C_1}\ket{x+u}^{D_1}$. 
Applied to $\ket{\tilde x}^A$ we have
\begin{subequations}
\begin{align}
U_\W\ket{\tilde x}^A  &\simeq \tfrac1{\sqrt 2}\!\!\! \sum_{u,z \in \{0,1 \} }\!\!(-1)^{xz}\sqrt{p_u} \ket{u}^{C_1}\ket {u+z}^{D_1}\left(\sqrt{1-\beta}\ket{z+u}^B\ket?^{C_2}+\sqrt\beta\ket?^B\ket{z+u}^{C_2}\right)\\
&=\tfrac1{\sqrt 2}\!\!\! \sum_{u,y \in \{0,1 \} }\!\!(-1)^{x(y-u)}\sqrt{p_u} \ket{u}^{C_1}\ket {y}^{D_1}\left(\sqrt{1-\beta}\ket{y}^B\ket?^{C_2}+\sqrt\beta\ket?^B\ket{y}^{C_2}\right)\\
&=\ket{\theta_x}^{C_1}\otimes \tfrac1{\sqrt 2}\!\! \sum_{y \in \{0,1 \} }\!\!(-1)^{xy}\ket {y}^{D_1}\left(\sqrt{1-\beta}\ket{y}^B\ket?^{C_2}+\sqrt\beta\ket?^B\ket{y}^{C_2}\right)\\
&\simeq \ket{\theta_x}^{C_1}\otimes \tfrac1{\sqrt 2}\!\! \sum_{y \in \{0,1 \} }\!\!(-1)^{xy}\left(\sqrt{1-\beta}\ket{y}^B\ket?^{C_2}\ket {y}^{D_1}+\sqrt\beta\ket?^B\ket{y}^{C_2}\ket0^{D_1}\right).
\end{align}
\end{subequations}
In the last step we have applied the same unitary on $C_2$ and $D_1$ as in the BEC example. 
Tracing out $B$ gives the states
\begin{align}
\sigma_x^{CD}=\ket{\theta_x}\bra{\theta_x}^{C_1}\otimes \left((1-\beta)\ket?\bra?^{C_2}\otimes \tfrac12\mathbbm{1}^{D_1}+\beta \ket{\tilde x}\bra{\tilde x}^{C_2}\otimes \ket0\bra0^{D_1}\right).
\end{align}
Again, the $D_1$ system is irrelevant. The fidelity of the two output states is then easily seen to equal $(1-\beta)(1-2\alpha)=F(\W^\setC)$.  


\end{myex}

\subsection{Sufficient conditions for alignment} \label{sec:UB}
Given two binary-input discrete memoryless channels $\W$ and $\V$ we can use Corollary~\ref{cor:Hamed} and Proposition~\ref{prop:UR} to derive sufficient conditions for $\mathcal{R}_{\varepsilon}^n(\W) \subseteq \mathcal{R}_{\varepsilon}^n(\V)$ or similarly $\mathcal{R}_{\varepsilon}^n(\W) \supseteq \mathcal{R}_{\varepsilon}^n(\V)$  by swapping the role of $\W$ and $\V$. We can derive such conditions on every level of the polarization tree. With $\V^\setC$ we denote the counterpart of channel $\V$ as defined in Section~\ref{sec:complementaryChannel}. 
\begin{mythm}[Level $k$ condition for alignment] \label{thm:UB}
Let $k \in \mathbb{N}_0$ and $0<\varepsilon<1$. If $ \F{\W_b} + \F{(\V^\setC)_{\bar b}} \leq 1$ for all $b \in \{0,1\}^k$, then $\mathcal{R}_{\varepsilon}^n(\W) \subseteq \mathcal{R}_{\varepsilon}^n(\V)$.
\end{mythm}
\begin{proof}
Consider $n\geq k$ and suppose $d \in \{0,1\}^n$ is such that the synthesized channel $\W_d$ is noisy, i.e. $\F{\W_d} \geq 1 - \varepsilon$. According to Corrolary~\ref{cor:Hamed} together with the assumption of the theorem this implies that $\F{(\V^{\setC})_{\bar d}}\leq \varepsilon$. Proposition~\ref{prop:UR} then ensures that $\F{\V_d} \geq 1 -  \varepsilon$. This implies that $\mathcal{R}_{\varepsilon}^n(\W) \subseteq \mathcal{R}_{\varepsilon}^n(\V)$. 
\end{proof}
Consider the first level where we have two channel pairs $(\W_0, \W_1)$ and $((\V^{\setC})_0,(\V^{\setC})_1)$. Note that in general the two channels $(\V^{\setC})_0$ and $(\V^{\setC})_1$ are not counterpart channels of $\V_1$ and $\V_0$ using the definition given in Section~\ref{sec:complementaryChannel} (i.e., in general $(\V^{\setC})_0 \ne (\V_1)^{\setC}$ and $(\V^{\setC})_1 \ne (\V_0)^{\setC}$). One could work with channels $((\V_0)^{\setC},(\V_1)^{\setC})$ instead of $((\V^{\setC})_0,(\V^{\setC})_1)$ which could lead to better criterion for alignment. However, at the drawback that the criterion would be more difficult to compute. For that reason this approach is not pursued in this article. 

\begin{myremark}[Criterion for alignment cannot get worse for higher levels] \label{rmk:UBworse}
Suppose the sufficient conditions at level 1 in Theorem~\ref{thm:UB} are satisfied. Then using the inequality $\F{\W_0} \leq  2 \F{\W}-\F{\W}^2$ and the identity $\F{\W_1}=\F{\W}^2$ \cite[Prop.\ 17]{renes14_2} and $0\leq \F{\W}\leq 1$, we obtain 
\begin{align}
\F{\W}+\F{\V^\setC} \leq \sqrt{\F{\W_1}} + 1- \sqrt{1-\F{\V_0^{\setC}}} \leq 1,
\end{align}
where the last inequality uses $\F{\W_1}+\F{(\V^{\setC})_0}\leq 1$ which is given by assumption. This argument can be applied to each level and thus shows that if the assumptions in Theorem~\ref{thm:UB} at level $k$ are satisfied they are also satisfied for all levels $\ell \leq k$.
\end{myremark}

\begin{myremark}[No improvement after level 0 for $\BEC$s]  \label{rmk:noImprovementBEC}
In case $\W$ or $\V$ is a $\BEC$, the sufficient conditions in Theorem~\ref{thm:UB} cannot be improved by going to higher levels than level 0. Let $\W$ be a $\BEC(\alpha)$. The level 0 condition requires that $\alpha \geq \F{\V^{\setC}}$. One condition of the first level is $Z(\W_0)+\F{(\V^\setC)_1} \leq 1$. Since $\W$ is a $\BEC$ we know that $Z(\W_0) = 2 Z(\W) - Z(\W)^2 = 1- \beta^2$. Moreover $\F{(\V^{\setC})_1} = \F{\V^{\setC}}^2$ and thus as $\beta \in [0,1]$ the condition from level 1 coincides with the one from level 0. This argument carries over to higher levels. Note that in case $\V$ is a $\BEC$ the same justification can be applied as the counterpart channel of a $\BEC$ is a $\BEC$ again (see Example~\ref{ex:complementBEC}).
\end{myremark}


\subsection{Channels with non-uniform input distribution}
The sufficient conditions that imply $|\mathcal{R}_{\varepsilon}^n(\W) \cap \mathcal{D}_{\varepsilon}^n(\V)|  = \Omega(n)$ or $|\mathcal{D}_{\varepsilon}^n(\W) \cap \mathcal{R}_{\varepsilon}^n(\V)|  = \Omega(n)$ as introduced in Section~\ref{sec:LB} are valid for binary-input discrete memoryless channels $\W$ and $\V$ with an arbitrary input distribution. The sufficient conditions for $\mathcal{R}_{\varepsilon}^n(\W) \subseteq \mathcal{R}_{\varepsilon}^n(\V)$ and $\mathcal{R}_{\varepsilon}^n(\W) \subseteq \mathcal{R}_{\varepsilon}^n(\V)$ derived in Section~\ref{sec:UB} depend on the input distribution to the channels $\W$ and $\V$ but it can be shown that they remain valid for non-uniform input distributions. The idea is to consider a generalized fidelity measure that is defined for a binary-input cq channel that is described via the mapping $\{0,1\} \ni x \mapsto \rho_x \in \mathcal{D}(\mathcal{H})$ as $Z(X|B):=2 \sqrt{p(1-p)} \F{\rho_0,\rho_1}$ where $p$ denotes the probability that we observe at the output the state $\rho_0$. It has been shown that $Z(X|B)$ polarizes in the same way as $\F{\rho_0,\rho_1}$ \cite[Prop.\ 17]{renes14_2} which proves that Theorem~\ref{thm:UB} remains valid also for channels with a non-uniform input distribution.

\section{Applications} \label{sec:applications}
In this section we demonstrate the performance of the statements derived in Theorems~\ref{thm:LB} and \ref{thm:UB} on several well-known scenarios. A special emphasis will be put on BSC and BEC channels as they oftentimes show extreme behavior.
\begin{myremark}[{\cite[Ex.\ 5.4, p.~121]{elgamal12}}] \label{rmk:BSC_BEC}
Let $\W:\mathcal{X} \to \mathcal{Y}$ be a $\BSC(\alpha)$ and $\V:\mathcal{X}\to \mathcal{Z}$ be a $\BEC(\beta)$. Then the following holds:
\begin{enumerate}[(i)]
\item For $0\leq \beta \leq 2 \alpha$, $\W$ is a degraded version of $\V$. 
\item For $2\alpha < \beta \leq 4\alpha(1-\alpha)$, $\V$ is less noisy than $\W$, but  $\W$ is not a degraded version of $\V$.
\item For $4 \alpha(1-\alpha) < \beta \leq \Hb(\alpha)$, $\V$ is more capable than $\W$, but not less noisy.
\item For $\Hb(\alpha) < \beta \leq 1$, $\V$ and $\W$ are not more capable comparable.
\end{enumerate}
\end{myremark}

\subsection{BSC/BEC pair with a uniform input distribution} 
Let $\W:\mathcal{X} \to \mathcal{Y}$ be a $\BSC(\alpha)$ for $\alpha \in [0,\tfrac{1}{2}]$ and $\V:\mathcal{X} \to \mathcal{Z}$ be a $\BEC(\beta)$ for $\beta \in [0,1]$. Consider a uniform input distribution, i.e., $X \sim \Bernoulli{\tfrac{1}{2}}$.
According to Remark~\ref{rmk:BSC_BEC} and Proposition~\ref{prop:LN} we know that for $\beta \leq 4\alpha(1-\alpha)$ the channel $\V$ is less noisy than $\W$ and hence $\mathcal{D}_{\varepsilon}^n(\W)$ $\scriptsize{\overset{\boldsymbol{\cdot}}{\subseteq}}$ $\mathcal{D}_{\varepsilon}^n(\V)$ and $\mathcal{R}_{\varepsilon}^n(\W)$ $\scriptsize{\overset{\boldsymbol{\cdot}}{\supseteq}}$ $\mathcal{R}_{\varepsilon}^n(\V)$. To determine a region where $\mathcal{R}_{\varepsilon}^n(\W) \subseteq \mathcal{R}_{\varepsilon}^n(\V)$ we can use the technique derived in Section~\ref{sec:UB} which ensures that this is the case if $Z(\W)-Z(\V^{\setC}) \leq 1$.\footnote{This is the condition at level $0$.} Recalling that $\V^{\setC} = \BEC(1-\beta)$ (see Example~\ref{ex:complementBEC}) then gives $\beta \geq 2 \sqrt{\alpha(1-\alpha)}$. As discussed in Remark~\ref{rmk:noImprovementBEC} this criterion cannot be improved by considering higher levels as the channel $\V$ is a $\BEC$. Using the technique explained in Section~\ref{sec:LB} (cf.\ Theorem~\ref{thm:LB}) we can determine regions where $|\mathcal{R}_{\varepsilon}^n(\W)\cap \mathcal{D}_{\varepsilon}^n(\V)| = \Omega(n)$ or $|\mathcal{D}_{\varepsilon}^n(\W)\cap \mathcal{R}_{\varepsilon}^n(\V)|  = \Omega(n)$. Figure~\ref{fig:BSCBEC} summarizes the results about the alignment properties of the polarized sets $\mathcal{R}_{\varepsilon}^n(\W)$, $\mathcal{R}_{\varepsilon}^n(\V)$, $\mathcal{D}_{\varepsilon}^n(\W)$, and $\mathcal{D}_{\varepsilon}^n(\V)$ for all pairs $(\alpha,\beta) \in [0,\tfrac{1}{2}]\times[0,1]$.

\begin{figure}[!htb]
  \begin{tikzpicture}
	\begin{axis}[
		height=9cm,
		width=16cm,
		xlabel=$\alpha$,
		ylabel=$\beta$,
		xmin=0,
		xmax=0.5,
		ymax=1.0,
		ymin=0,
		legend style={at={(0.678,0.38)},anchor=north,legend cell align=left} 
	]

\addplot [draw=none,smooth,fill=blue!05,forget plot]coordinates {
(0.,0.) (0.01,0.0396) (0.02,0.0784) (0.03,0.1164) (0.04,0.1536) (0.05,0.19) (0.06,0.2256) (0.07,0.2604) (0.08,0.2944) (0.09,0.3276) (0.1,0.36) (0.11,0.3916) (0.12,0.4224) (0.13,0.4524) (0.14,0.4816) (0.15,0.51) (0.16,0.5376) (0.17,0.5644) (0.18,0.5904) (0.19,0.6156) (0.2,0.64) (0.21,0.6636) (0.22,0.6864) (0.23,0.7084) (0.24,0.7296) (0.25,0.75) (0.26,0.7696) (0.27,0.7884) (0.28,0.8064) (0.29,0.8236) (0.3,0.84) (0.31,0.8556) (0.32,0.8704) (0.33,0.8844) (0.34,0.8976) (0.35,0.91) (0.36,0.9216) (0.37,0.9324) (0.38,0.9424) (0.39,0.9516) (0.4,0.96) (0.41,0.9676) (0.42,0.9744) (0.43,0.9804) (0.44,0.9856) (0.45,0.99) (0.46,0.9936) (0.47,0.9964) (0.48,0.9984) (0.49,0.9996) (0.5,1.) (0.5,0) (0,0)	
	};	
\addplot [draw=none,smooth,fill=YellowOrange!7,forget plot]coordinates {
(0.,0) (0.01,0.0585613) (0.02,0.103412) (0.03,0.144193) (0.04,0.182559) (0.05,0.219203) (0.06,0.254488) (0.07,0.288624) (0.08,0.321744) (0.09,0.353931) (0.1,0.385241) (0.11,0.41571) (0.12,0.445364) (0.13,0.474219) (0.14,0.502285) (0.15,0.529571) (0.16,0.55608) (0.17,0.581816) (0.18,0.60678) (0.19,0.630975) (0.2,0.6544) (0.21,0.677057) (0.22,0.698946) (0.23,0.720066) (0.24,0.740418) (0.25,0.760002) (0.26,0.778818) (0.27,0.796865) (0.28,0.814145) (0.29,0.830657) (0.3,0.846401) (0.31,0.861377) (0.32,0.875585) (0.33,0.889025) (0.34,0.901697) (0.35,0.913601) (0.36,0.924736) (0.37,0.935105) (0.38,0.944704) (0.39,0.953538) (0.4,0.961579) (0.41,0.968917) (0.42,0.975145) (0.43,0.982) (0.44,0.989) (0.45,0.990453) (0.46,0.99865) (0.47,0.9993) (0.48,0.99995) (0.49,1) (0.5,1)
(0.5,1.) (0.49,0.999712) (0.48,0.998849) (0.47,0.99742) (0.46,0.995435) (0.45,0.992907)(0.44,0.989849)(0.43,0.986276) (0.42,0.982198) (0.41,0.977628) (0.4,0.972573) (0.39,0.967042) (0.38,0.961039) (0.37,0.954568) (0.36,0.947632) (0.35,0.940229) (0.34,0.93236) (0.33,0.92402) (0.32,0.915206) (0.31,0.90591) (0.3,0.896125) (0.29,0.885842) (0.28,0.875049) (0.27,0.863734) (0.26,0.851881) (0.25,0.839473) (0.24,0.826491) (0.23,0.812913) (0.22,0.798714) (0.21,0.783865) (0.2,0.768335) (0.19,0.752087)  (0.18,0.73508) (0.17,0.717267) (0.16,0.698591) (0.15,0.67899)  (0.14,0.658389) (0.13,0.636699)  (0.12,0.613815) (0.11,0.58961) (0.1,0.563926) (0.09,0.536565) (0.08,0.507276) (0.07,0.475723)  (0.06,0.44145) (0.05,0.403802) (0.04,0.361777) (0.03,0.313694) (0.02,0.25629) (0.01,0.181165)  (0.,0) 		
	};	
\addplot [draw=none,smooth,fill=ForestGreen!05,forget plot]coordinates {	
(0.,0.) (0.01,0.198997) (0.02,0.28) (0.03,0.341174) (0.04,0.391918) (0.05,0.43589) (0.06,0.474974) (0.07,0.510294) (0.08,0.542586) (0.09,0.572364) (0.1,0.6) (0.11,0.62578) (0.12,0.649923) (0.13,0.672607) (0.14,0.693974) (0.15,0.714143) (0.16,0.733212) (0.17,0.751266) (0.18,0.768375) (0.19,0.784602) (0.2,0.8) (0.21,0.814616) (0.22,0.828493) (0.23,0.841665) (0.24,0.854166) (0.25,0.866025) (0.26,0.877268) (0.27,0.887919) (0.28,0.897998) (0.29,0.907524) (0.3,0.916515) (0.31,0.924986) (0.32,0.932952) (0.33,0.940425) (0.34,0.947418) (0.35,0.953939) (0.36,0.96) (0.37,0.965609) (0.38,0.970773) (0.39,0.9755) (0.4,0.979796) (0.41,0.983667) (0.42,0.987117) (0.43,0.990152) (0.44,0.992774) (0.45,0.994987) (0.46,0.996795) (0.47,0.998198) (0.48,0.9992) (0.49,0.9998) (0.5,1.) (0,1) (0,0)		
};		
		\addplot[draw=none,smooth,fill=red!05,forget plot] coordinates {
(0.,0) (0.01,0.0807931) (0.02,0.141441) (0.03,0.194392) (0.04,0.242292) (0.05,0.286397) (0.06,0.327445) (0.07,0.365924) (0.08,0.402179) (0.09,0.43647) (0.1,0.468996) (0.11,0.499916) (0.12,0.529361) (0.13,0.557438) (0.14,0.584239) (0.15,0.60984) (0.16,0.63431) (0.17,0.657705) (0.18,0.680077) (0.19,0.701471) (0.2,0.721928) (0.21,0.741483) (0.22,0.760168) (0.23,0.778011) (0.24,0.79504) (0.25,0.811278) (0.26,0.826746) (0.27,0.841465) (0.28,0.855451) (0.29,0.868721) (0.3,0.881291) (0.31,0.893173) (0.32,0.904381) (0.33,0.914926) (0.34,0.924819) (0.35,0.934068) (0.36,0.942683) (0.37,0.950672) (0.38,0.958042) (0.39,0.9648) (0.4,0.970951) (0.41,0.9765) (0.42,0.981454) (0.43,0.985815) (0.44,0.989588) 		
(0.5,1) 
(0.5,1.) (0.49,0.999712) (0.48,0.998849) (0.47,0.99742) (0.46,0.995435) (0.45,0.992907)(0.44,0.989849)(0.43,0.986276) (0.42,0.982198) (0.41,0.977628) (0.4,0.972573) (0.39,0.967042) (0.38,0.961039) (0.37,0.954568) (0.36,0.947632) (0.35,0.940229) (0.34,0.93236) (0.33,0.92402) (0.32,0.915206) (0.31,0.90591) (0.3,0.896125) (0.29,0.885842) (0.28,0.875049) (0.27,0.863734) (0.26,0.851881) (0.25,0.839473) (0.24,0.826491) (0.23,0.812913) (0.22,0.798714) (0.21,0.783865) (0.2,0.768335) (0.19,0.752087)  (0.18,0.73508) (0.17,0.717267) (0.16,0.698591) (0.15,0.67899)  (0.14,0.658389) (0.13,0.636699)  (0.12,0.613815) (0.11,0.58961) (0.1,0.563926) (0.09,0.536565) (0.08,0.507276) (0.07,0.475723)  (0.06,0.44145) (0.05,0.403802) (0.04,0.361777) (0.03,0.313694) (0.02,0.25629) (0.01,0.181165)  (0.,0) 			              	 											
	};


		\addplot[ForestGreen,thick,smooth] coordinates {
(0.,0.) (0.01,0.198997) (0.02,0.28) (0.03,0.341174) (0.04,0.391918) (0.05,0.43589) (0.06,0.474974) (0.07,0.510294) (0.08,0.542586) (0.09,0.572364) (0.1,0.6) (0.11,0.62578) (0.12,0.649923) (0.13,0.672607) (0.14,0.693974) (0.15,0.714143) (0.16,0.733212) (0.17,0.751266) (0.18,0.768375) (0.19,0.784602) (0.2,0.8) (0.21,0.814616) (0.22,0.828493) (0.23,0.841665) (0.24,0.854166) (0.25,0.866025) (0.26,0.877268) (0.27,0.887919) (0.28,0.897998) (0.29,0.907524) (0.3,0.916515) (0.31,0.924986) (0.32,0.932952) (0.33,0.940425) (0.34,0.947418) (0.35,0.953939) (0.36,0.96) (0.37,0.965609) (0.38,0.970773) (0.39,0.9755) (0.4,0.979796) (0.41,0.983667) (0.42,0.987117) (0.43,0.990152) (0.44,0.992774) (0.45,0.994987) (0.46,0.996795) (0.47,0.998198) (0.48,0.9992) (0.49,0.9998) (0.5,1.)									
	};
	\addlegendentry{$\mathcal{R}_{\varepsilon}^n(\W)\subseteq \mathcal{R}_{\varepsilon}^n(\V)$ (level 0)}

		\addplot[red,thick] coordinates {
		(0.,0) (0.01,0.181165) (0.02,0.25629) (0.03,0.313694) (0.04,0.361777) (0.05,0.403802) (0.06,0.44145) (0.07,0.475723) (0.08,0.507276) (0.09,0.536565) (0.1,0.563926) (0.11,0.58961) (0.12,0.613815) (0.13,0.636699) (0.14,0.658389) (0.15,0.67899) (0.16,0.698591) (0.17,0.717267) (0.18,0.73508) (0.19,0.752087) (0.2,0.768335) (0.21,0.783865) (0.22,0.798714) (0.23,0.812913) (0.24,0.826491) (0.25,0.839473) (0.26,0.851881) (0.27,0.863734) (0.28,0.875049) (0.29,0.885842) (0.3,0.896125) (0.31,0.90591) (0.32,0.915206) (0.33,0.92402) (0.34,0.93236) (0.35,0.940229) (0.36,0.947632) (0.37,0.954568) (0.38,0.961039) (0.39,0.967042) (0.4,0.972573) (0.41,0.977628) (0.42,0.982198) (0.43,0.986276) (0.44,0.989849) (0.45,0.992907) (0.46,0.995435) (0.47,0.99742) (0.48,0.998849) (0.49,0.999712) (0.5,1.)													
	};
	\addlegendentry{$|\mathcal{R}_{\varepsilon}^n(\W) \cap \mathcal{D}_{\varepsilon}^n(\V)| = \Omega(n)$ (level $4$)}		

	\addplot[black,thick,smooth] coordinates {
(0.,0) (0.01,0.0807931) (0.02,0.141441) (0.03,0.194392) (0.04,0.242292) (0.05,0.286397) (0.06,0.327445) (0.07,0.365924) (0.08,0.402179) (0.09,0.43647) (0.1,0.468996) (0.11,0.499916) (0.12,0.529361) (0.13,0.557438) (0.14,0.584239) (0.15,0.60984) (0.16,0.63431) (0.17,0.657705) (0.18,0.680077) (0.19,0.701471) (0.2,0.721928) (0.21,0.741483) (0.22,0.760168) (0.23,0.778011) (0.24,0.79504) (0.25,0.811278) (0.26,0.826746) (0.27,0.841465) (0.28,0.855451) (0.29,0.868721) (0.3,0.881291) (0.31,0.893173) (0.32,0.904381) (0.33,0.914926) (0.34,0.924819) (0.35,0.934068) (0.36,0.942683) (0.37,0.950672) (0.38,0.958042) (0.39,0.9648) (0.4,0.970951) (0.41,0.9765) (0.42,0.981454) (0.43,0.985815) (0.44,0.989588) (0.45,0.992774) (0.46,0.995378) (0.47,0.997402) (0.48,0.998846) (0.49,0.999711) (0.5,1.)											
	};
	\addlegendentry{$\I{X}{Y}=\I{X}{Z}$}


		\addplot[YellowOrange,thick,smooth] coordinates {
(0.,0) (0.01,0.0585613) (0.02,0.103412) (0.03,0.144193) (0.04,0.182559) (0.05,0.219203) (0.06,0.254488) (0.07,0.288624) (0.08,0.321744) (0.09,0.353931) (0.1,0.385241) (0.11,0.41571) (0.12,0.445364) (0.13,0.474219) (0.14,0.502285) (0.15,0.529571) (0.16,0.55608) (0.17,0.581816) (0.18,0.60678) (0.19,0.630975) (0.2,0.6544) (0.21,0.677057) (0.22,0.698946) (0.23,0.720066) (0.24,0.740418) (0.25,0.760002) (0.26,0.778818) (0.27,0.796865) (0.28,0.814145) (0.29,0.830657) (0.3,0.846401) (0.31,0.861377) (0.32,0.875585) (0.33,0.889025) (0.34,0.901697) (0.35,0.913601) (0.36,0.924736) (0.37,0.935105) (0.38,0.944704) (0.39,0.953538) (0.4,0.961579) (0.41,0.968917) (0.42,0.975145) (0.43,0.982) (0.44,0.989) (0.45,0.990453) (0.46,0.99865) (0.47,0.9993) (0.48,0.99995) (0.49,1) (0.5,1)
	};
		\addlegendentry{$|\mathcal{D}_{\varepsilon}^n(\W) \cap \mathcal{R}_{\varepsilon}^n(\V)| = \Omega(n)$ (level 3) }

		\addplot[blue,thick,smooth] coordinates {
(0.,0.) (0.01,0.0396) (0.02,0.0784) (0.03,0.1164) (0.04,0.1536) (0.05,0.19) (0.06,0.2256) (0.07,0.2604) (0.08,0.2944) (0.09,0.3276) (0.1,0.36) (0.11,0.3916) (0.12,0.4224) (0.13,0.4524) (0.14,0.4816) (0.15,0.51) (0.16,0.5376) (0.17,0.5644) (0.18,0.5904) (0.19,0.6156) (0.2,0.64) (0.21,0.6636) (0.22,0.6864) (0.23,0.7084) (0.24,0.7296) (0.25,0.75) (0.26,0.7696) (0.27,0.7884) (0.28,0.8064) (0.29,0.8236) (0.3,0.84) (0.31,0.8556) (0.32,0.8704) (0.33,0.8844) (0.34,0.8976) (0.35,0.91) (0.36,0.9216) (0.37,0.9324) (0.38,0.9424) (0.39,0.9516) (0.4,0.96) (0.41,0.9676) (0.42,0.9744) (0.43,0.9804) (0.44,0.9856) (0.45,0.99) (0.46,0.9936) (0.47,0.9964) (0.48,0.9984) (0.49,0.9996) (0.5,1.)																
	};
		\addlegendentry{$\mathcal{D}_{\varepsilon}^n(\W)$ $\scriptsize{\overset{\boldsymbol{\cdot}}{\subseteq}}$ $\mathcal{D}_{\varepsilon}^n(\V)$ and $\mathcal{R}_{\varepsilon}^n(\W)$ $\scriptsize{\overset{\boldsymbol{\cdot}}{\supseteq}}$ $\mathcal{R}_{\varepsilon}^n(\V)$ }	

	\end{axis}  
\node[blue] at (10,4.5) {$\mathcal{D}_{\varepsilon}^n(\W)$ $\scriptsize{\overset{\boldsymbol{\cdot}}{\subseteq}}$ $\mathcal{D}_{\varepsilon}^n(\V)$};
\node[blue] at (10,3.8) {$\mathcal{R}_{\varepsilon}^n(\W)$ $\scriptsize{\overset{\boldsymbol{\cdot}}{\supseteq}}$ $\mathcal{R}_{\varepsilon}^n(\V)$};
\node[ForestGreen] at (2,6) {$\mathcal{R}_{\varepsilon}^n(\W)\subseteq \mathcal{R}_{\varepsilon}^n(\V)$};
\node[rotate=36,red] at (2.9,3.75) {{\footnotesize $|\mathcal{R}_{\varepsilon}^n(\W) \cap \mathcal{D}_{\varepsilon}^n(\V)| = \Omega(n)$}};
\node[rotate=38,YellowOrange] at (3.2,3.35) {{\footnotesize $|\mathcal{D}_{\varepsilon}^n(\W) \cap \mathcal{R}_{\varepsilon}^n(\V)| = \Omega(n)$}};

\end{tikzpicture}

\caption{Alignment of the polarized sets defined in \eqref{eq:polarizedSets} for $\W=\BSC(\alpha)$, $\V=\BEC(\beta)$ with $\alpha \in [0,\tfrac{1}{2}]$ and $\beta \in [0,1]$ and a uniform input distribution. The black line shows the region where the two channels have the same capacity, $\beta = \Hb(\alpha)$. In the blue region, $\V$ is less noisy than $\W$ and hence  Proposition~\ref{prop:LN} ensures $\mathcal{D}_{\varepsilon}^n(\W)$ $\scriptsize{\overset{\boldsymbol{\cdot}}{\subseteq}}$ $\mathcal{D}_{\varepsilon}^n(\V)$ and $\mathcal{R}_{\varepsilon}^n(\W)$ $\scriptsize{\overset{\boldsymbol{\cdot}}{\supseteq}}$ $\mathcal{R}_{\varepsilon}^n(\V)$. The remaining colored regions are determined using the conditions given in Theorems~\ref{thm:LB} and \ref{thm:UB} evaluated for different levels.}
\label{fig:BSCBEC}
\end{figure}
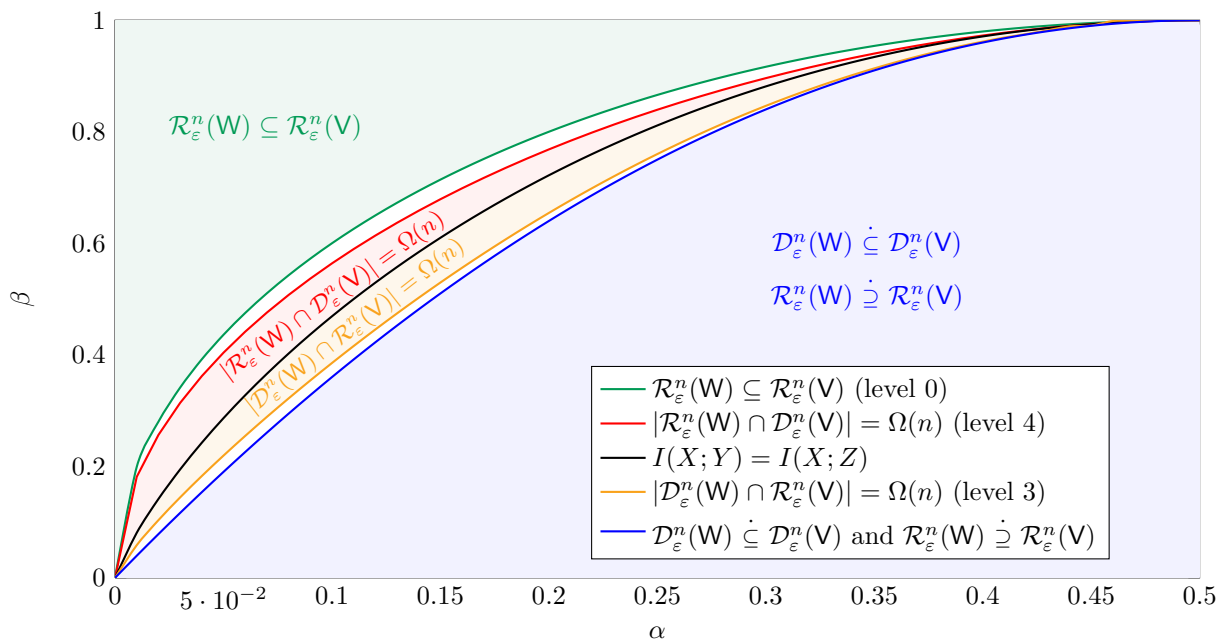

\subsection{BSC-BEC wiretap channel}\label{sec:BscBecWTC}
Consider a discrete memoryless wiretap channel where the channel from Alice to Bob, $\W:\mathcal{X}\to \mathcal{Y}$ is a $\BSC(\alpha)$ with $\alpha \in [0,\tfrac{1}{2}]$ and the channel from Alice to Eve $\V:\mathcal{X} \to \mathcal{Z}$ is a $\BEC(\beta)$ with $\beta \in [0,1]$. As discussed in Remark~\ref{rmk:BSC_BEC}, $\W$ is not more capable than $\V$ for any $\alpha \in [0,\tfrac{1}{2}]$ and $\beta \in [0,1]$. For $0\leq \beta \leq 4 \alpha (1-\alpha)$, $\V$ is less noisy than $\W$ (see Remark~\ref{rmk:BSC_BEC}) which implies that the secrecy capacity of the wiretap channel is zero \cite[Cor.\ 3]{csiszar78}. Therefore, we consider the setup where $4 \alpha (1-\alpha) <\beta \leq 1$. In this setup the secrecy capacity is positive as $\V$ is not less noisy than $\W$. It has been shown \cite[Sec.\ 5]{ulukus13} that for this model the secrecy capacity given in Theorem~\ref{thm:Cs} simplifies to
\begin{equation}
C_s(\W,\V) = \max \limits_{\gamma \in [0,\tfrac{1}{2}]} \I{U}{Y} - \I{U}{Z} 
\end{equation}
where 
\begin{equation}
P_{X|U} = \BSC(\gamma), \,\textnormal{ and } \, \Prob{U=0} = \Prob{U=1} = \frac{1}{2}.
\end{equation}
For $\gamma^\star_{\alpha,\beta}=\arg \max_{\gamma \in [0,1/2]}  \I{U}{Y} - \I{U}{Z}$, we define a wiretap channel that includes the optimal preprocessing\footnote{By preprocessing we mean the distribution $P_{X|U}$ which is a $\BSC(\gamma)$ for this case.} as $\bar \W_{\gamma^\star_{\alpha,\beta}} := \W \circ \BSC(\gamma^\star_{\alpha,\beta})$ being the channel from Alice to Bob and $\bar \V_{\gamma^\star_{\alpha,\beta}}:=\V \circ  \BSC(\gamma^\star_{\alpha,\beta}) $ the channel to the eavesdropper. We also know as stated above that the optimal input distribution for the wiretap channel $(\bar \W_{\gamma^\star_{\alpha,\beta}}, \bar \V_{\gamma^\star_{\alpha,\beta}})$ is the uniform.

Using the notation defined in \eqref{eq:polarizedSets} and following the idea introduced in \cite{vardy11}, we can derive a coding scheme based on polar codes that achieves the secrecy capacity of the $(\W,\V)$ wiretap channel by inserting message bits to the indices specified by the set $\mathcal{M}_{\varepsilon}^n(\bar \W_{\gamma^\star_{\alpha,\beta}},\bar \V_{\gamma^\star_{\alpha,\beta}}):=\mathcal{D}_{\varepsilon}^n(\bar \W_{\gamma^\star_{\alpha,\beta}}) \cap \mathcal{R}_{\varepsilon}^n(\bar \V_{\gamma^\star_{\alpha,\beta}})$. The indices given by $\mathcal{A}_{\varepsilon}^n(\bar \W_{\gamma^\star_{\alpha,\beta}},\bar \V_{\gamma^\star_{\alpha,\beta}}):=\mathcal{D}_{\varepsilon}^n(\bar \W_{\gamma^\star_{\alpha,\beta}}) \cap \overline{\mathcal{R}_{\varepsilon}^n(\bar \V_{\gamma^\star_{\alpha,\beta}})}$ are filled with random bits and the indices of $\mathcal{F}_{\varepsilon}^n (\bar \W_{\gamma^\star_{\alpha,\beta}},\bar \V_{\gamma^\star_{\alpha,\beta}}):=\overline{\mathcal{D}_{\varepsilon}^n(\bar \W_{\gamma^\star_{\alpha,\beta}})} \cap \mathcal{R}_{\varepsilon}^n(\bar \V_{\gamma^\star_{\alpha,\beta}})$ are frozen to $0$. Finally, the indices specified by the set $\mathcal{K}_{\varepsilon}^n(\bar \W_{\gamma^\star_{\alpha,\beta}},\bar \V_{\gamma^\star_{\alpha,\beta}}):=\overline{\mathcal{D}_{\varepsilon}^n(\bar \W_{\gamma^\star_{\alpha,\beta}})} \cap \overline{\mathcal{R}_{\varepsilon}^n(\bar \V_{\gamma^\star_{\alpha,\beta}})}$ have to be filled with a secret key that is shared between Alice and Bob. As it may be difficult to provide a (large) secret key shared between Alice and Bob, it is desirable to have an understanding of how large $|\mathcal{K}_{\varepsilon}^n(\bar \W_{\gamma^\star_{\alpha,\beta}},\bar \V_{\gamma^\star_{\alpha,\beta}})|$ is. In particular it is interesting to locate choices of $(\alpha,\beta)$ for which essentially no secret key is necessary and such for which secret key is clearly needed. This can be done with the help of Theorems~\ref{thm:LB} and \ref{thm:UB} derived in Sections~\ref{sec:LB} and \ref{sec:UB}. Recall that by definition of the polarized sets given in \eqref{eq:polarizedSets} we have $\overline{\mathcal{D}_{\varepsilon}^n(\W)}=\mathcal{R}_{1-\varepsilon}^n(\W)\supsetneq \mathcal{R}_{\varepsilon}^n(\W)$ for every DMC $\W$ and $\varepsilon \in (0,\frac{1}{2})$. Furthermore, by the polarization phenomenon \cite{arikan09} the following relation $|\overline{\mathcal{D}_{\varepsilon}^n(\W)} \cap \mathcal{R}_{\varepsilon}^n(\W)| = o(n)$ holds. Therefore an alignment result of the form $\mathcal{R}_{\varepsilon}^n(\W) \subseteq  \mathcal{R}_{\varepsilon}^n(\V)$ implies that $|\overline{\mathcal{D}_{\varepsilon}^n(\W)} \cap \overline{\mathcal{R}_{\varepsilon}^n(\V)}| = o(n)$.

The level $0$ condition of Theorem~\ref{thm:UB} ensures that for any pair $(\alpha,\beta)\in [0,\tfrac{1}{2}] \times [0,1]$ where $\F{\bar \W_{\gamma^\star_{\alpha,\beta}}}+\F{(\bar \V_{\gamma^\star_{\alpha,\beta}})^{\setC}} \leq 1$ we have $\mathcal{R}_{\varepsilon}^n(\bar \W_{\gamma^\star_{\alpha,\beta}}) \subseteq \mathcal{R}_{\varepsilon}^n(\bar \V_{\gamma^\star_{\alpha,\beta}})$ which as explained above implies that $|\mathcal{K}_{\varepsilon}^n(\bar \W_{\gamma^\star_{\alpha,\beta}},\bar \V_{\gamma^\star_{\alpha,\beta}})|=o(n)$. Recall that the counterpart channel $(\bar \V_{\gamma^\star_{\alpha,\beta}})^{\setC}$ has been derived in Example~\ref{ex:counterpartBSC}. We note that the conditions given in Theorem~\ref{thm:LB} evaluated for high levels seems to always lie inside the region where $C_s=0$, i.e., Theorem~\ref{thm:LB} does not provide any useful information.
Figure~\ref{fig:WTC} determines pairs $(\alpha,\beta)$ for which essentially no key-assistance is needed, i.e., $|\mathcal{K}_{\varepsilon}^n(\bar \W_{\gamma^\star_{\alpha,\beta}},\bar \V_{\gamma^\star_{\alpha,\beta}})|=o(n)$.

\subsection{BEC-BSC wiretap channel}\label{sec:BecBscWTC}
Consider a discrete memoryless wiretap channel where the channel from Alice to Bob $\V:\mathcal{X}\to \mathcal{Y}$ is a $\BEC(\beta)$ with $\beta \in [0,1]$ and the channel from Alice to Eve $\W:\mathcal{X} \to \mathcal{Z}$ is a $\BSC(\alpha)$ with $\alpha \in [0,\tfrac{1}{2}]$. As discussed in Remark~\ref{rmk:BSC_BEC} for $\beta \leq \Hb(\alpha)$, $\V$ is more capable than $\W$ and thus by Corollary~\ref{cor:CsMoreCapable} $C_s(\W,\V)= \max_{P_X} \I{X}{Y}-\I{X}{Z}$, i.e., no preprocessing is needed to achieve the secrecy capacity and therefore it is straightforward to build a coding scheme using polar codes that achieves $C_s(\W,\V)$ with $\mathcal{K}_{\varepsilon}^n(\W, \V):=\overline{\mathcal{R}_\varepsilon^n(\W)} \cap\overline{\mathcal{D}_\varepsilon^n(\V)}$ representing the set where key assistance is needed. If $\beta \leq 4 \alpha(1-\alpha)$, $\V$ is less noisy than $\W$ and by Proposition~\ref{prop:LN} this implies $|\mathcal{K}_{\varepsilon}^n(\W, \V)| = o(n)$. For $4 \alpha (1-\alpha) <\beta \leq \Hb(\alpha)$, $V$ is more capable than $\W$ but not less noisy. Proposition~\ref{prop:MC} implies that for the capacity achieving input distribution we have $|\mathcal{K}_{\varepsilon}^n(\W, \V)| = o(n)$. Therefore for $0\leq \beta \leq \Hb(\alpha)$ the key recycling protocol introduced in \cite{sasoglu13} can be used to achieve the secrecy capacity without the need of initial preshared key.

For $\beta \geq \Hb(\alpha)$ it has been shown \cite{ulukus13} that the secrecy capacity is given by
\begin{align} \label{eq:capacityUlu}
C_s(\W,\V) = \left\{
\begin{array}{lll}
			&\max\limits_{r,\gamma} 		& f((1-r)\gamma)-r f(0) - (1-r)f(\gamma) \\
			&\st					& 0\leq r \leq 1\\
			& 						& 0\leq \gamma \leq 1,
	\end{array} \right.
\end{align}
with $[0,1] \ni p \mapsto f(p):=\I{X}{Y}-\I{X}{Z} \in [-1,1]$ for $p:=\Prob{X=0}$. Let $r_{\alpha,\beta}^{\star}$ and $\gamma_{\alpha,\beta}^{\star}$ denote the optimizers of \eqref{eq:capacityUlu}, it has been shown in \cite{ulukus13} that $C_s(\W,\V)=\I{U}{Y}-\I{U}{Z}$ with a preprocessing $U \in \{0,1 \}$ such that $\Prob{U=0|X=0}=0$, $\Prob{U=0|X=1}=\gamma^{\star}_{\alpha,\beta}$ and $\Prob{U=0}= r^{\star}_{\alpha,\beta}$. For $\Rc_{\gamma^{\star}_{\alpha,\beta}}:\mathcal{U}\to \mathcal{X}$ being a channel that describes the preprocessing $U \markovDavid X \markovDavid (Y,Z)$ explained above we define $\bar \W_{\gamma^\star_{\alpha,\beta}}:=\W \circ \Rc_{\gamma^{\star}_{\alpha,\beta}}$ and $\bar \V_{\gamma^\star_{\alpha,\beta}}:=\V \circ \Rc_{\gamma^{\star}_{\alpha,\beta}}$. Considering an input distribution $\Prob{U=0}= r^{\star}_{\alpha,\beta}$, we can derive a coding scheme based on polar codes that achieves the secrecy capacity of the $(\V,\W)$ wiretap channel by inserting message bits to the indices specified by the set $\mathcal{M}_{\varepsilon}^n(\bar \W_{\gamma^\star_{\alpha,\beta}},\bar \V_{\gamma^\star_{\alpha,\beta}}):=\mathcal{R}_{\varepsilon}^n(\bar \W_{\gamma^\star_{\alpha,\beta}}) \cap \mathcal{D}_{\varepsilon}^n(\bar \V_{\gamma^\star_{\alpha,\beta}})$, filling the indices given by $\mathcal{A}_{\varepsilon}^n(\bar \W_{\gamma^\star_{\alpha,\beta}},\bar \V_{\gamma^\star_{\alpha,\beta}}):=\overline{\mathcal{R}_{\varepsilon}^n(\bar \W_{\gamma^\star_{\alpha,\beta}})} \cap \mathcal{D}_{\varepsilon}^n(\bar \V_{\gamma^\star_{\alpha,\beta}})$ with random bits and freeze the indices corresponding to $\mathcal{F}_{\varepsilon,\varepsilon}^n (\bar \W_{\gamma^\star_{\alpha,\beta}},\bar \V_{\gamma^\star_{\alpha,\beta}}):=\mathcal{R}_{\varepsilon}^n(\bar \W_{\gamma^\star_{\alpha,\beta}}) \cap \overline{\mathcal{D}_{\varepsilon}^n(\bar \V_{\gamma^\star_{\alpha,\beta}})}$ to $0$. The indices specified by the set $\mathcal{K}_{\varepsilon}^n(\bar \W_{\gamma^\star_{\alpha,\beta}},\bar \V_{\gamma^\star_{\alpha,\beta}}):=\overline{\mathcal{R}_{\varepsilon}^n(\bar \W_{\gamma^\star_{\alpha,\beta}})} \cap \overline{\mathcal{D}_{\varepsilon}^n(\bar \V_{\gamma^\star_{\alpha,\beta}})}$ have to be filled with secret key that is shared between Alice and Bob. As depicted in Figure~\ref{fig:WTC} the condition to apply Theorem~\ref{thm:LB} evaluated at level $3$ can be used to determine a region where $|\mathcal{K}_{\varepsilon}^n(\bar \W_{\gamma^\star_{\alpha,\beta}},\bar \V_{\gamma^\star_{\alpha,\beta}})|=\Omega(n)$, i.e., preshared key is needed to achieve the secrecy capacity using the protocol given in \cite{sasoglu13}. Recall that Theorem~\ref{thm:LB} detect cases where $|\mathcal{D}_{\varepsilon}^n(\bar \W_{\gamma^\star_{\alpha,\beta}})\cap \mathcal{R}_{\varepsilon}^n(\bar \V_{\gamma^\star_{\alpha,\beta}})|=\Omega(n)$ which implies that $|\overline{\mathcal{R}_{\varepsilon}^n(\bar \W_{\gamma^\star_{\alpha,\beta}})} \cap \overline{\mathcal{D}_{\varepsilon}^n(\bar \V_{\gamma^\star_{\alpha,\beta}})}|=\Omega(n)$ as by the polarization phenomenon \cite{arikan09} we have $|\mathcal{D}_{\varepsilon}^n(\bar \W_{\gamma^\star_{\alpha,\beta}}) \cap \overline{\mathcal{R}_{\varepsilon}^n(\bar \W_{\gamma^\star_{\alpha,\beta}})}|=o(n)$ and $|\mathcal{R}_{\varepsilon}^n(\bar \V_{\gamma^\star_{\alpha,\beta}}) \cap \overline{\mathcal{D}_{\varepsilon}^n(\bar \V_{\gamma^\star_{\alpha,\beta}})}|=o(n)$.

\begin{figure}[!htb]
\hspace{-7mm}
    \subfloat[$\BSC(\alpha)-\BEC(\beta)$ wiretap channel]{  \begin{tikzpicture}
	\begin{axis}[
		height=8cm,
		width=8cm,
		xlabel=$\alpha$,
		ylabel=$\beta$,
		xmin=0,
		xmax=0.5,
		ymax=1.0,
		ymin=0,
		legend style={at={(0.59,0.20)},anchor=north,legend cell align=left,font=\footnotesize} 
	]

\addplot [draw=none,smooth,fill=red!05,forget plot]coordinates {
(0.,0.) (0.01,0.0396) (0.02,0.0784) (0.03,0.1164) (0.04,0.1536) (0.05,0.19) (0.06,0.2256) (0.07,0.2604) (0.08,0.2944) (0.09,0.3276) (0.1,0.36) (0.11,0.3916) (0.12,0.4224) (0.13,0.4524) (0.14,0.4816) (0.15,0.51) (0.16,0.5376) (0.17,0.5644) (0.18,0.5904) (0.19,0.6156) (0.2,0.64) (0.21,0.6636) (0.22,0.6864) (0.23,0.7084) (0.24,0.7296) (0.25,0.75) (0.26,0.7696) (0.27,0.7884) (0.28,0.8064) (0.29,0.8236) (0.3,0.84) (0.31,0.8556) (0.32,0.8704) (0.33,0.8844) (0.34,0.8976) (0.35,0.91) (0.36,0.9216) (0.37,0.9324) (0.38,0.9424) (0.39,0.9516) (0.4,0.96) (0.41,0.9676) (0.42,0.9744) (0.43,0.9804) (0.44,0.9856) (0.45,0.99) (0.46,0.9936) (0.47,0.9964) (0.48,0.9984) (0.49,0.9996) (0.5,1)	(0.5,0) (0,0)				
	};	

%

\addplot [draw=none,smooth,fill=ForestGreen!05,forget plot]coordinates {
(0.,0.) (0.005,0.175) (0.01,0.232) (0.015,0.275) (0.03,0.370) (0.04,0.418) (0.05,0.460) (0.08,0.562) (0.09,0.590)(0.1,0.617) (0.11,0.641) (0.12,0.665) (0.13,0.686) (0.14,0.707) (0.15,0.726) (0.16,0.744) (0.17,0.761) (0.18,0.778) (0.19,0.793)(0.20,0.808)(0.21,0.822)(0.22,0.836)(0.23,0.848)(0.24,0.860)(0.25,0.872)(0.26,0.882)(0.27,0.893)(0.28,0.902)(0.29,0.911)(0.30,0.920)(0.31,0.928)(0.32,0.936)(0.33,0.943)(0.34,0.950)(0.35,0.956)(0.36,0.962)(0.37,0.967)(0.38,0.972)(0.39,0.977)(0.40,0.981)(0.41,0.985)(0.42,0.988)(0.43,0.991)(0.44,0.994)(0.45,0.996)(0.46,0.997)(0.47,0.9983)(0.48,0.9993)(0.49,0.9997)(0.5,1)(0,1) (0,0)
	};		

		\addplot[ForestGreen,thick,smooth] coordinates {
(0.,0.) (0.005,0.175) (0.01,0.232) (0.015,0.275) (0.03,0.370) (0.04,0.418) (0.05,0.460) (0.08,0.562) (0.09,0.590)(0.1,0.617) (0.11,0.641) (0.12,0.665) (0.13,0.686) (0.14,0.707) (0.15,0.726) (0.16,0.744) (0.17,0.761) (0.18,0.778) (0.19,0.793)(0.20,0.808)(0.21,0.822)(0.22,0.836)(0.23,0.848)(0.24,0.860)(0.25,0.872)(0.26,0.882)(0.27,0.893)(0.28,0.902)(0.29,0.911)(0.30,0.920)(0.31,0.928)(0.32,0.936)(0.33,0.943)(0.34,0.950)(0.35,0.956)(0.36,0.962)(0.37,0.967)(0.38,0.972)(0.39,0.977)(0.40,0.981)(0.41,0.985)(0.42,0.988)(0.43,0.991)(0.44,0.994)(0.45,0.996)(0.46,0.997)(0.47,0.9983)(0.48,0.9993)(0.49,0.9997)(0.5,1)					
	};	
	\addlegendentry{$\mathcal{R}_{\varepsilon}^n(\bar \W_{\gamma^\star_{\alpha,\beta}}) \subseteq \mathcal{R}_{\varepsilon}^n(\bar \V_{\gamma^\star_{\alpha,\beta}})$}
		\addplot[red,thick,smooth] coordinates {
(0.,0.) (0.01,0.0396) (0.02,0.0784) (0.03,0.1164) (0.04,0.1536) (0.05,0.19) (0.06,0.2256) (0.07,0.2604) (0.08,0.2944) (0.09,0.3276) (0.1,0.36) (0.11,0.3916) (0.12,0.4224) (0.13,0.4524) (0.14,0.4816) (0.15,0.51) (0.16,0.5376) (0.17,0.5644) (0.18,0.5904) (0.19,0.6156) (0.2,0.64) (0.21,0.6636) (0.22,0.6864) (0.23,0.7084) (0.24,0.7296) (0.25,0.75) (0.26,0.7696) (0.27,0.7884) (0.28,0.8064) (0.29,0.8236) (0.3,0.84) (0.31,0.8556) (0.32,0.8704) (0.33,0.8844) (0.34,0.8976) (0.35,0.91) (0.36,0.9216) (0.37,0.9324) (0.38,0.9424) (0.39,0.9516) (0.4,0.96) (0.41,0.9676) (0.42,0.9744) (0.43,0.9804) (0.44,0.9856) (0.45,0.99) (0.46,0.9936) (0.47,0.9964) (0.48,0.9984) (0.49,0.9996) (0.5,1.)					
	};
		\addlegendentry{$C_s=0$ i.e.\ $\beta\leq4 \alpha (1-\alpha)$}	

	\end{axis}  
\node[red] at (5,4.0) {$C_s = 0$};
\node[ForestGreen] at (1.9,6.05) {{\footnotesize $\mathcal{R}_{\varepsilon}^n(\bar \W_{\gamma^\star_{\alpha,\beta}}) \subseteq \mathcal{R}_{\varepsilon}^n(\bar \V_{\gamma^\star_{\alpha,\beta}})$ }};
\end{tikzpicture} }
    \subfloat[$\BEC(\beta)-\BSC(\alpha)$ wiretap channel]{  \begin{tikzpicture}
	\begin{axis}[
		height=8cm,
		width=8cm,
		xlabel=$\alpha$,
		ylabel=$\beta$,
		xmin=0,
		xmax=0.5,
		ymax=1.0,
		ymin=0,
		legend style={at={(0.58,0.28)},anchor=north,legend cell align=left,font=\scriptsize} 
	]


	\addplot[draw=none,fill=YellowOrange!7,smooth,forget plot] coordinates {
(0.,0) (0.01,0.0807931) (0.02,0.141441) (0.03,0.194392) (0.04,0.242292) (0.05,0.286397) (0.06,0.327445) (0.07,0.365924) (0.08,0.402179) (0.09,0.43647) (0.1,0.468996) (0.11,0.499916) (0.12,0.529361) (0.13,0.557438) (0.14,0.584239) (0.15,0.60984) (0.16,0.63431) (0.17,0.657705) (0.18,0.680077) (0.19,0.701471) (0.2,0.721928) (0.21,0.741483) (0.22,0.760168) (0.23,0.778011) (0.24,0.79504) (0.25,0.811278) (0.26,0.826746) (0.27,0.841465) (0.28,0.855451) (0.29,0.868721) (0.3,0.881291) (0.31,0.893173) (0.32,0.904381) (0.33,0.914926) (0.34,0.924819) (0.35,0.934068) (0.36,0.942683) (0.37,0.950672) (0.38,0.958042) (0.39,0.9648) (0.4,0.970951) (0.41,0.9765) (0.42,0.981454) (0.43,0.985815) (0.44,0.989588) (0.45,0.992774) (0.46,0.995378) (0.47,0.997402) (0.48,0.998846) (0.49,0.999711) (0.5,1.) (0,1) (0,0)	
	};

\addplot [draw=none,smooth,fill=blue!05,forget plot]coordinates {
(0.,0.) (0.01,0.0396) (0.02,0.0784) (0.03,0.1164) (0.04,0.1536) (0.05,0.19) (0.06,0.2256) (0.07,0.2604) (0.08,0.2944) (0.09,0.3276) (0.1,0.36) (0.11,0.3916) (0.12,0.4224) (0.13,0.4524) (0.14,0.4816) (0.15,0.51) (0.16,0.5376) (0.17,0.5644) (0.18,0.5904) (0.19,0.6156) (0.2,0.64) (0.21,0.6636) (0.22,0.6864) (0.23,0.7084) (0.24,0.7296) (0.25,0.75) (0.26,0.7696) (0.27,0.7884) (0.28,0.8064) (0.29,0.8236) (0.3,0.84) (0.31,0.8556) (0.32,0.8704) (0.33,0.8844) (0.34,0.8976) (0.35,0.91) (0.36,0.9216) (0.37,0.9324) (0.38,0.9424) (0.39,0.9516) (0.4,0.96) (0.41,0.9676) (0.42,0.9744) (0.43,0.9804) (0.44,0.9856) (0.45,0.99) (0.46,0.9936) (0.47,0.9964) (0.48,0.9984) (0.49,0.9996) (0.5,1)	(0.5,0) (0,0)				
	};	

	\addplot[draw=none,fill=gray!20,smooth,forget plot] coordinates {
(0.,0) (0.01,0.0807931) (0.02,0.141441) (0.03,0.194392) (0.04,0.242292) (0.05,0.286397) (0.06,0.327445) (0.07,0.365924) (0.08,0.402179) (0.09,0.43647) (0.1,0.468996) (0.11,0.499916) (0.12,0.529361) (0.13,0.557438) (0.14,0.584239) (0.15,0.60984) (0.16,0.63431) (0.17,0.657705) (0.18,0.680077) (0.19,0.701471) (0.2,0.721928) (0.21,0.741483) (0.22,0.760168) (0.23,0.778011) (0.24,0.79504) (0.25,0.811278) (0.26,0.826746) (0.27,0.841465) (0.28,0.855451) (0.29,0.868721) (0.3,0.881291) (0.31,0.893173) (0.32,0.904381) (0.33,0.914926) (0.34,0.924819) (0.35,0.934068) (0.36,0.942683) (0.37,0.950672) (0.38,0.958042) (0.39,0.9648) (0.4,0.970951) (0.41,0.9765) (0.42,0.981454) (0.43,0.985815) (0.44,0.989588) (0.45,0.992774) (0.46,0.995378) (0.47,0.997402) (0.48,0.998846) (0.49,0.999711) (0.5,1.)	
(0.49,0.9996) (0.48,0.9984)(0.47,0.9964)(0.46,0.9936)(0.45,0.99) (0.44,0.9856)(0.43,0.9804)(0.42,0.9744) (0.41,0.9676)	(0.4,0.96) (0.39,0.9516)(0.38,0.9424)(0.37,0.9324) (0.36,0.9216)	(0.35,0.91)	(0.34,0.8976)(0.33,0.8844)(0.32,0.8704)(0.31,0.8556)  (0.3,0.84) (0.29,0.8236) (0.28,0.8064)(0.27,0.7884)(0.26,0.7696)(0.25,0.75)(0.24,0.7296)(0.23,0.7084)(0.22,0.6864) (0.21,0.6636) (0.2,0.64) (0.19,0.6156) 	(0.18,0.5904)	(0.17,0.5644) (0.16,0.5376) (0.15,0.51)(0.14,0.4816)(0.13,0.4524) (0.12,0.4224)(0.11,0.3916)(0.1,0.36)	(0.09,0.3276)	(0.08,0.2944)	(0.07,0.2604)(0.06,0.2256) (0.05,0.19)(0.04,0.1536)(0.03,0.1164)(0.02,0.0784)(0.01,0.0396) (0.,0.) 
	};

	\addplot[YellowOrange,thick,smooth] coordinates {
(0.,0) (0.01,0.0807931) (0.02,0.141441) (0.03,0.194392) (0.04,0.242292) (0.05,0.286397) (0.06,0.327445) (0.07,0.365924) (0.08,0.402179) (0.09,0.43647) (0.1,0.468996) (0.11,0.499916) (0.12,0.529361) (0.13,0.557438) (0.14,0.584239) (0.15,0.60984) (0.16,0.63431) (0.17,0.657705) (0.18,0.680077) (0.19,0.701471) (0.2,0.721928) (0.21,0.741483) (0.22,0.760168) (0.23,0.778011) (0.24,0.79504) (0.25,0.811278) (0.26,0.826746) (0.27,0.841465) (0.28,0.855451) (0.29,0.868721) (0.3,0.881291) (0.31,0.893173) (0.32,0.904381) (0.33,0.914926) (0.34,0.924819) (0.35,0.934068) (0.36,0.942683) (0.37,0.950672) (0.38,0.958042) (0.39,0.9648) (0.4,0.970951) (0.41,0.9765) (0.42,0.981454) (0.43,0.985815) (0.44,0.989588) (0.45,0.992774) (0.46,0.995378) (0.47,0.997402) (0.48,0.998846) (0.49,0.999711) (0.5,1.)											
	};
	\addlegendentry{\!$|\mathcal{D}_{\varepsilon}^n(\bar \W_{\gamma^\star_{\alpha,\beta}}) \!\cap\! \mathcal{R}_{\varepsilon}^n(\bar \V_{\gamma^\star_{\alpha,\beta}})| \!=\! \Omega(n)$\!\!\!\!}


	\addplot[gray,thick,smooth] coordinates {
(0.,0) (0.01,0.0807931) (0.02,0.141441) (0.03,0.194392) (0.04,0.242292) (0.05,0.286397) (0.06,0.327445) (0.07,0.365924) (0.08,0.402179) (0.09,0.43647) (0.1,0.468996) (0.11,0.499916) (0.12,0.529361) (0.13,0.557438) (0.14,0.584239) (0.15,0.60984) (0.16,0.63431) (0.17,0.657705) (0.18,0.680077) (0.19,0.701471) (0.2,0.721928) (0.21,0.741483) (0.22,0.760168) (0.23,0.778011) (0.24,0.79504) (0.25,0.811278) (0.26,0.826746) (0.27,0.841465) (0.28,0.855451) (0.29,0.868721) (0.3,0.881291) (0.31,0.893173) (0.32,0.904381) (0.33,0.914926) (0.34,0.924819) (0.35,0.934068) (0.36,0.942683) (0.37,0.950672) (0.38,0.958042) (0.39,0.9648) (0.4,0.970951) (0.41,0.9765) (0.42,0.981454) (0.43,0.985815) (0.44,0.989588) (0.45,0.992774) (0.46,0.995378) (0.47,0.997402) (0.48,0.998846) (0.49,0.999711) (0.5,1.)											
	};
	\addlegendentry{$\mathcal{D}_{\varepsilon}^n(\W)\!\!$ $\scriptsize{\overset{\boldsymbol{\cdot}}{\subseteq}}$ $\!\!\mathcal{D}_{\varepsilon}^n(\V)$, $\mathcal{R}_{\varepsilon}^n(\W)\!$ $\scriptsize{\overset{\boldsymbol{\cdot}}{\supseteq}}$ $\!\mathcal{R}_{\varepsilon}^n(\V)$\!\!\!}

		\addplot[blue,thick,smooth] coordinates {
(0.,0.) (0.01,0.0396) (0.02,0.0784) (0.03,0.1164) (0.04,0.1536) (0.05,0.19) (0.06,0.2256) (0.07,0.2604) (0.08,0.2944) (0.09,0.3276) (0.1,0.36) (0.11,0.3916) (0.12,0.4224) (0.13,0.4524) (0.14,0.4816) (0.15,0.51) (0.16,0.5376) (0.17,0.5644) (0.18,0.5904) (0.19,0.6156) (0.2,0.64) (0.21,0.6636) (0.22,0.6864) (0.23,0.7084) (0.24,0.7296) (0.25,0.75) (0.26,0.7696) (0.27,0.7884) (0.28,0.8064) (0.29,0.8236) (0.3,0.84) (0.31,0.8556) (0.32,0.8704) (0.33,0.8844) (0.34,0.8976) (0.35,0.91) (0.36,0.9216) (0.37,0.9324) (0.38,0.9424) (0.39,0.9516) (0.4,0.96) (0.41,0.9676) (0.42,0.9744) (0.43,0.9804) (0.44,0.9856) (0.45,0.99) (0.46,0.9936) (0.47,0.9964) (0.48,0.9984) (0.49,0.9996) (0.5,1.)					
	};
		\addlegendentry{\!$\mathcal{D}_{\varepsilon}^n(\W)\!$ $\scriptsize{\overset{\boldsymbol{\cdot}}{\subseteq}}$\! $\mathcal{D}_{\varepsilon}^n(\V)$, $\mathcal{R}_{\varepsilon}^n(\W)\!$ $\scriptsize{\overset{\boldsymbol{\cdot}}{\supseteq}}$ $\! \mathcal{R}_{\varepsilon}^n(\V)$\!\!\!}	

	\end{axis}  
\node[blue] at (5,4.0) {\footnotesize$\mathcal{D}_{\varepsilon}^n(\W)$ $\scriptsize{\overset{\boldsymbol{\cdot}}{\subseteq}}$ $\mathcal{D}_{\varepsilon}^n(\V)$};
\node[blue] at (5,3.2) {\footnotesize $\mathcal{R}_{\varepsilon}^n(\W)$ $\scriptsize{\overset{\boldsymbol{\cdot}}{\supseteq}}$ $\mathcal{R}_{\varepsilon}^n(\V)$};
\node[rotate=54,gray,scale=0.9] at (1.85,3.53) {\footnotesize$\mathcal{D}_{\varepsilon}^n(\W)$ $\scriptsize{\overset{\boldsymbol{\cdot}}{\subseteq}}$ $\mathcal{D}_{\varepsilon}^n(\V)$ };
\node[YellowOrange] at (1.99,5.95) {\footnotesize $|\mathcal{D}_{\varepsilon}^n(\bar \W_{\gamma^\star_{\alpha,\beta}}) \cap \mathcal{R}_{\varepsilon}^n(\bar \V_{\gamma^\star_{\alpha,\beta}})| $};
\node[YellowOrange] at (1.5,5.4) {\footnotesize $= \Omega(n)$};

\end{tikzpicture}}\\
    \caption{Alignment for wiretap channels involving a BEC and a BSC. The left plot considers a $\BSC(\alpha)-\BEC(\beta)$ wiretap channel, the right $\BEC(\beta)-\BSC(\alpha)$. For channels in the green region on the left, the coding scheme described in the text can achieve the secrecy capacity $C_s$ and requires essentially no key-assistance, i.e.\  $|\mathcal{K}_{\varepsilon}^n(\bar \W_{\gamma^\star_{\alpha,\beta}},\bar \V_{\gamma^\star_{\alpha,\beta}})|=o(n)$. The region is determined using Theorem~\ref{thm:UB} for level $0$. As much holds for the blue and gray regions on the right, but now the conclusion follows from the fact that in the blue (gray) region $\V$ is less noisy (more capable) than $\W$. The boundaries of these regions are given by $4\alpha (1-\alpha) < \beta \leq \Hb(\alpha)$ and  $\beta \leq 4 \alpha (1-\alpha)$, respectively. $\V$ is more capable but not less noisy than $\W$, i.e., $4\alpha (1-\alpha) < \beta \leq \Hb(\alpha)$.  For $\beta > \Hb(\alpha)$ the conditions of Theorem~\ref{thm:LB} evaluated at level $3$ can be used to infer that key-assistance is required for channels in the orange region, as $|\overline{\mathcal{R}_{\varepsilon}^n(\bar \W_{\gamma^\star_{\alpha,\beta}})} \cap \overline{\mathcal{D}_{\varepsilon}^n(\bar \V_{\gamma^\star_{\alpha,\beta}})}| = \Omega(n)$.
}
    \label{fig:WTC}
\end{figure}
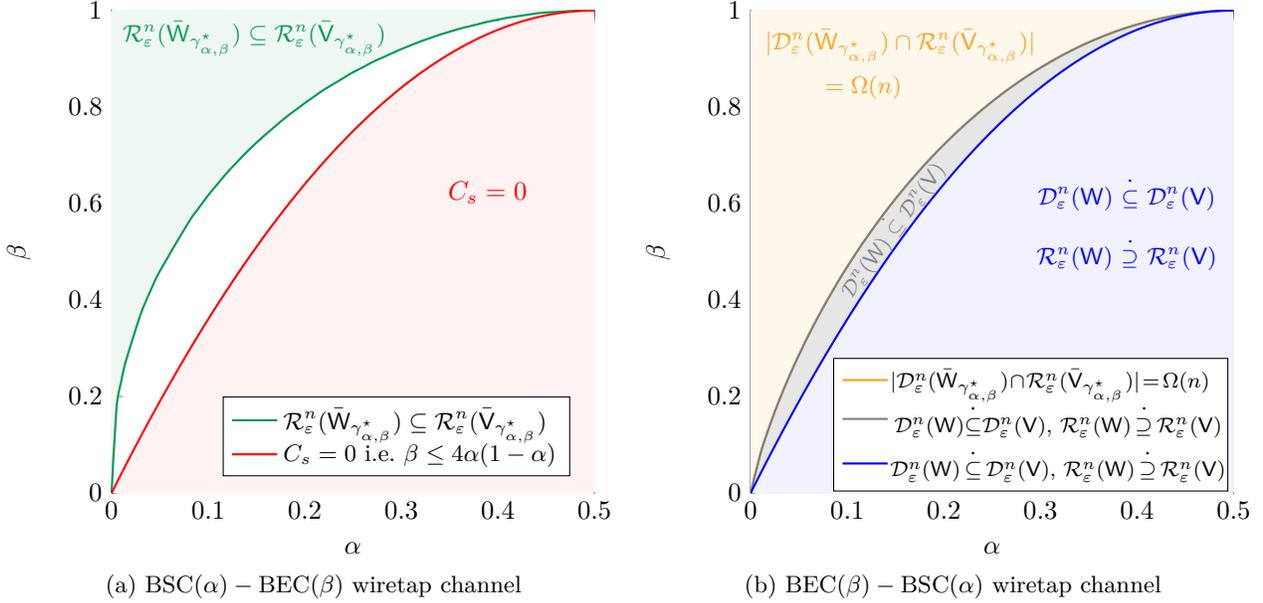

\begin{myremark} 
Constructing a polar coding scheme for wiretap channels with $|\mathcal K_{\varepsilon}^n(\W,\V)|=o(n)$ would nominally require preshared key for inputs associated with this set, in order to ensure strong security of the messages from the eavesdropper. However, \cite{sasoglu13} presents a bootstrapping technique whereby the coding procedure is repeated and key required in the current round is generated in previous rounds, without affecting the security statement. 

A different approach, which avoids bootstrapping altogether, is to construct a coding scheme for the wiretap channel by modifying a polar code for transmitting quantum information. The idea, sketched in \cite{renes12}, is that the resulting polar code for the wiretap channel could be thought of as virtually implementing the quantum code, and thus inherits all its coding and security properties. The results of Section~\ref{sec:entanglement} can be used to infer that absolutely no preshared key is required for the quantum code, and therefore the same holds for the classical wiretap code. A more detailed description of this construction will be presented elsewhere, but let us comment on the difference between the two approaches. The present construction builds a polar code for the wiretap channel purely in the classical domain, only resorting to quantum arguments (the uncertainty relation of Proposition~\ref{prop:UR}) to find channel pairs for which the bootstrapping coding scheme of \cite{sasoglu13} can be applied. In contrast, the argument via quantum coding constructs the code for the classical wiretap channel directly from the quantum code, avoiding the uncertainty relation in this form. 
%
\end{myremark}

\subsection{BSC/BEC broadcast channel}
In this section we consider a broadcast channel consisting of a $\BSC(\alpha)$ and a $\BEC(\beta)$ for $\alpha \in [0,\tfrac{1}{2}]$ and $\beta \in [0,1]$, where a sender wants to transmit two messages to two receivers (i.e., there is no common message) a setup that is explained in detail in \cite{nair10}. It has been shown recently that superposition coding is optimal in all regimes \cite{nair10} and hence the inner bound described by Theorem~\ref{thm:superposition} coincides with the capacity region. Furthermore it has been shown that the optimal preprocessing is a $\BSC(s)$ with $s \in [0,\tfrac{1}{2}]$, i.e., $P_{X|U}= \BSC(\gamma)$ for $\gamma \in[0,\tfrac{1}{2}$] and that the optimal input distribution is uniform, i.e., $\Prob{U=0}=\Prob{U=1}=\tfrac{1}{2}$ \cite{nair10}. To simplify notation, we define $\bar \W_\gamma := \BSC(\alpha) \circ \BSC(\gamma)$ and $\bar \V_{\gamma}:=\BEC(\beta) \circ \BSC(\gamma)$. If the polarized sets for the channels $\bar \W_\gamma$ respectively $\bar \V_{\gamma}$ would be aligned (as it happens e.g., if the two channels $\bar \W_\gamma$ and $\bar \V_\gamma$ are less noisy ordered), the protocol introduced in  \cite{goela13} can be used to achieve the capacity region. 

For $0 \leq \beta \leq 4 \alpha (1-\alpha)$ $\V$ is less noisy than $\W$ and thus for all $\gamma \in [0,\tfrac{1}{2}]$ the channel $\bar \V_{\gamma}$ is less noisy than $\bar \W_{\gamma}$. Proposition~\ref{prop:LN} then implies that for $0 \leq \beta \leq 4 \alpha (1-\alpha)$ the polarized sets are  essentially aligned, i.e., $\mathcal{R}_{\varepsilon}^n(\bar \W_{\gamma})$ $\scriptsize{\overset{\boldsymbol{\cdot}}{\supseteq}}$ $\mathcal{R}_{\varepsilon}^n(\bar \V_{\gamma})$ and $\mathcal{D}_{\varepsilon}^n(\bar \W_{\gamma})$ $\scriptsize{\overset{\boldsymbol{\cdot}}{\subseteq}}$ $\mathcal{D}_{\varepsilon}^n(\bar \V_{\gamma})$. As a result, the protocol explained in \cite{goela13} can be used to achieve the capacity region in  this regime.

For pairs $(\alpha,\beta)$ such that $I(\bar \V_{\gamma})\leq I(\bar \W_{\gamma})$, the conditions given in Theorems~\ref{thm:LB} and \ref{thm:UB} can be used to determine regions where $|\mathcal{R}_{\varepsilon}^n(\bar \W_{\gamma}) \cap \mathcal{D}_{\varepsilon}^n(\bar \V_{\gamma})| =\Omega(n)$ and where $\mathcal{R}_{\varepsilon}^n(\bar \W_{\gamma}) \subseteq \mathcal{R}_{\varepsilon}^n(\bar \V_{\gamma})$. Whenever encountering a setup of $(\alpha,\beta)$ where $\mathcal{R}_{\varepsilon}^n(\bar \W_{\gamma}) \subseteq \mathcal{R}_{\varepsilon}^n(\bar \V_{\gamma})$, i.e., the polarized sets are aligned one can use the protocol introduced in \cite{goela13} to achieve the capacity region.
Figure~\ref{fig:BCoverview} shows an overview about the alignment characteristics for the polarized set of the channels $\bar \W_{\gamma}$ and $\bar \V_{\gamma}$ for different values of $\gamma$.


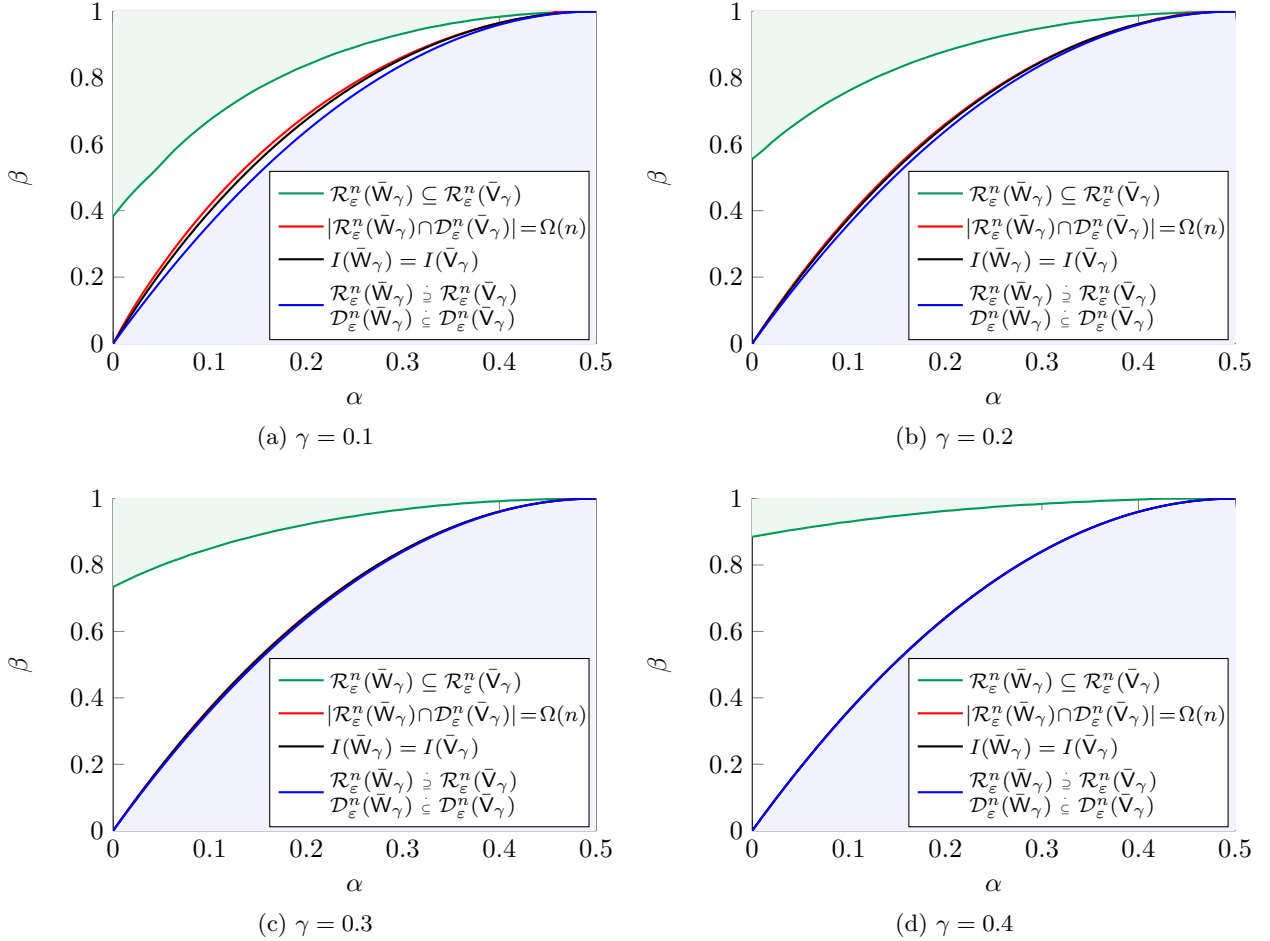
\begin{figure}[!htb]
\hspace{-10mm}
    \subfloat[$\gamma=0.1$]{   \begin{tikzpicture}
	\begin{axis}[
		height=6cm,
		width=8cm,
		xlabel=$\alpha$,
		ylabel=$\beta$,
		xmin=0,
		xmax=0.5,
		ymax=1.0,
		ymin=0,
		legend style={at={(0.655,0.52)},anchor=north,legend cell align=left,font=\scriptsize} 
	]

\addplot [draw=none,smooth,fill=ForestGreen!05,forget plot]coordinates {
(0.,0.383) (0.01,0.420) (0.02,0.453) (0.03,0.483) (0.04,0.511) (0.05,0.541) (0.06,0.572) (0.07,0.600) (0.08,0.626) (0.09,0.650) (0.1,0.673) (0.11,0.694) (0.12,0.714) (0.13,0.733) (0.14,0.751) (0.15,0.768) (0.16,0.783) (0.17,0.798) (0.18,0.812) (0.19,0.826) (0.2,0.838) (0.21,0.850) (0.22,0.862) (0.23,0.873) (0.24,0.883) (0.25,0.892) (0.26,0.902) (0.27,0.910) (0.28,0.918) (0.29,0.926) (0.3,0.933) (0.31,0.940) (0.32,0.947) (0.33,0.953) (0.34,0.958) (0.35,0.964) (0.36,0.968) (0.37,0.973) (0.38,0.977) (0.39,0.981) (0.4,0.984) (0.41,0.987) (0.42,0.990) (0.43,0.9922) (0.44,0.9943) (0.45,0.9960) (0.46,0.9975) (0.47,0.9986) (0.48,0.9995) (0.49,1) (0.5,1.)	(0,1) (0.,0.383)
};

		\addplot[draw=none,smooth,fill=red!10,forget plot] coordinates {
(0.,0)(0.01,0.0529982)(0.02,0.102378)(0.03,0.148514)(0.04,0.191781)(0.05,0.232724)(0.06,0.272772)(0.07,0.311078)(0.08,0.347749)(0.09,0.382883)(0.1,0.416564)(0.11,0.448867)(0.12,0.479862)(0.13,0.509608)(0.14,0.538161)(0.15,0.56557)(0.16,0.591881)(0.17,0.617135)(0.18,0.64137)(0.19,0.664621)(0.2,0.686919)(0.21,0.708295)(0.22,0.728776)(0.23,0.748387)(0.24,0.767152)(0.25,0.785093)(0.26,0.80223)(0.27,0.818582)(0.28,0.834165)(0.29,0.848996)(0.3,0.863089)(0.31,0.876456)(0.32,0.889109)(0.33,0.901058)(0.34,0.912312)(0.35,0.922877)(0.36,0.932759)(0.37,0.941964)(0.38,0.950494)(0.39,0.95835)(0.4,0.965534)(0.41,0.972046)(0.42,0.977884)(0.43,0.982154)(0.44,0.987529)(0.45,0.99133)(0.46,0.999999)(0.47,0.999999)(0.48,0.999999)(0.49,0.999652)(0.5,1.)			
(0.5,1)(0.49,0.999652)(0.48,0.998609)(0.47,0.996869)(0.46,0.994432)(0.45,0.991297)(0.44,0.987461)(0.43,0.982924)(0.42,0.977682) (0.41,0.971733)(0.4,0.965074)(0.39,0.9577) (0.38,0.949609)(0.37,0.940796)(0.36,0.931256)	(0.35,0.920984)(0.34,0.909973) (0.33,0.898218)(0.32,0.885712)(0.31,0.872448)(0.3,0.858417)(0.29,0.843611) (0.28,0.828021) (0.27,0.811637) (0.26,0.794449)(0.25,0.776444) (0.24,0.757612) (0.23,0.737938) (0.22,0.717409) (0.21,0.69601) (0.2,0.673725)(0.19,0.650535)	(0.18,0.626423) 	(0.17,0.601368)(0.16,0.575349) (0.15,0.548342)(0.14,0.520321)(0.13,0.491259)	(0.12,0.461126)(0.11,0.42989)(0.1,0.397514)(0.09,0.363959)(0.08,0.329184)(0.07,0.293139)	(0.06,0.255774)(0.05,0.217029)(0.04,0.176839)	(0.03,0.13513)	(0.02,0.0918204)(0.01,0.0468134)(0.,0)						
	};

\addplot [draw=none,smooth,fill=blue!05,forget plot]coordinates {
(0.,0.) (0.01,0.0396) (0.02,0.0784) (0.03,0.1164) (0.04,0.1536) (0.05,0.19) (0.06,0.2256) (0.07,0.2604) (0.08,0.2944) (0.09,0.3276) (0.1,0.36) (0.11,0.3916) (0.12,0.4224) (0.13,0.4524) (0.14,0.4816) (0.15,0.51) (0.16,0.5376) (0.17,0.5644) (0.18,0.5904) (0.19,0.6156) (0.2,0.64) (0.21,0.6636) (0.22,0.6864) (0.23,0.7084) (0.24,0.7296) (0.25,0.75) (0.26,0.7696) (0.27,0.7884) (0.28,0.8064) (0.29,0.8236) (0.3,0.84) (0.31,0.8556) (0.32,0.8704) (0.33,0.8844) (0.34,0.8976) (0.35,0.91) (0.36,0.9216) (0.37,0.9324) (0.38,0.9424) (0.39,0.9516) (0.4,0.96) (0.41,0.9676) (0.42,0.9744) (0.43,0.9804) (0.44,0.9856) (0.45,0.99) (0.46,0.9936) (0.47,0.9964) (0.48,0.9984) (0.49,0.9996) (0.5,1.) (0.5,0) (0,0)	
};


		\addplot[ForestGreen,thick,smooth] coordinates {
(0.,0.383) (0.01,0.420) (0.02,0.453) (0.03,0.483) (0.04,0.511) (0.05,0.541) (0.06,0.572) (0.07,0.600) (0.08,0.626) (0.09,0.650) (0.1,0.673) (0.11,0.694) (0.12,0.714) (0.13,0.733) (0.14,0.751) (0.15,0.768) (0.16,0.783) (0.17,0.798) (0.18,0.812) (0.19,0.826) (0.2,0.838) (0.21,0.850) (0.22,0.862) (0.23,0.873) (0.24,0.883) (0.25,0.892) (0.26,0.902) (0.27,0.910) (0.28,0.918) (0.29,0.926) (0.3,0.933) (0.31,0.940) (0.32,0.947) (0.33,0.953) (0.34,0.958) (0.35,0.964) (0.36,0.968) (0.37,0.973) (0.38,0.977) (0.39,0.981) (0.4,0.984) (0.41,0.987) (0.42,0.990) (0.43,0.9922) (0.44,0.9943) (0.45,0.9960) (0.46,0.9975) (0.47,0.9986) (0.48,0.9995) (0.49,1) (0.5,1.)									
	};
	\addlegendentry{$\mathcal{R}_{\varepsilon}^n(\bar \W_\gamma) \subseteq \mathcal{R}_{\varepsilon}^n(\bar \V_\gamma)$}



		\addplot[red,thick,smooth] coordinates {
(0.,0)(0.01,0.0529982)(0.02,0.102378)(0.03,0.148514)(0.04,0.191781)(0.05,0.232724)(0.06,0.272772)(0.07,0.311078)(0.08,0.347749)(0.09,0.382883)(0.1,0.416564)(0.11,0.448867)(0.12,0.479862)(0.13,0.509608)(0.14,0.538161)(0.15,0.56557)(0.16,0.591881)(0.17,0.617135)(0.18,0.64137)(0.19,0.664621)(0.2,0.686919)(0.21,0.708295)(0.22,0.728776)(0.23,0.748387)(0.24,0.767152)(0.25,0.785093)(0.26,0.80223)(0.27,0.818582)(0.28,0.834165)(0.29,0.848996)(0.3,0.863089)(0.31,0.876456)(0.32,0.889109)(0.33,0.901058)(0.34,0.912312)(0.35,0.922877)(0.36,0.932759)(0.37,0.941964)(0.38,0.950494)(0.39,0.95835)(0.4,0.965534)(0.41,0.972046)(0.42,0.977884)(0.43,0.982154)(0.44,0.987529)(0.45,0.99133)(0.46,0.999999)(0.47,0.999999)(0.48,0.999999)(0.49,0.999652)(0.5,1.)									
	};
	\addlegendentry{\!$|\mathcal{R}_{\varepsilon}^n(\bar \W_\gamma)\! \cap \! \mathcal{D}_{\varepsilon}^n(\bar \V_\gamma)|\! = \!\Omega(n)$\!\!\!}		
	


	\addplot[black,thick,smooth] coordinates {
(0.,0)(0.01,0.0468134)(0.02,0.0918204)(0.03,0.13513)(0.04,0.176839)(0.05,0.217029)(0.06,0.255774)(0.07,0.293139)(0.08,0.329184)(0.09,0.363959)(0.1,0.397514)(0.11,0.42989)(0.12,0.461126)(0.13,0.491259)(0.14,0.520321)(0.15,0.548342)(0.16,0.575349)(0.17,0.601368)(0.18,0.626423)(0.19,0.650535)(0.2,0.673725)(0.21,0.69601)(0.22,0.717409)(0.23,0.737938)(0.24,0.757612)(0.25,0.776444)(0.26,0.794449)(0.27,0.811637)(0.28,0.828021)(0.29,0.843611)(0.3,0.858417)(0.31,0.872448)(0.32,0.885712)(0.33,0.898218)(0.34,0.909973)(0.35,0.920984)(0.36,0.931256)(0.37,0.940796)(0.38,0.949609)(0.39,0.9577)(0.4,0.965074)(0.41,0.971733)(0.42,0.977682)(0.43,0.982924)(0.44,0.987461)(0.45,0.991297)(0.46,0.994432)(0.47,0.996869)(0.48,0.998609)(0.49,0.999652)(0.5,1)										
	};
	\addlegendentry{$I(\bar \W_{\gamma})=I(\bar \V_{\gamma})$}

		\addplot[blue,thick,smooth] coordinates {
(0.,0.) (0.01,0.0396) (0.02,0.0784) (0.03,0.1164) (0.04,0.1536) (0.05,0.19) (0.06,0.2256) (0.07,0.2604) (0.08,0.2944) (0.09,0.3276) (0.1,0.36) (0.11,0.3916) (0.12,0.4224) (0.13,0.4524) (0.14,0.4816) (0.15,0.51) (0.16,0.5376) (0.17,0.5644) (0.18,0.5904) (0.19,0.6156) (0.2,0.64) (0.21,0.6636) (0.22,0.6864) (0.23,0.7084) (0.24,0.7296) (0.25,0.75) (0.26,0.7696) (0.27,0.7884) (0.28,0.8064) (0.29,0.8236) (0.3,0.84) (0.31,0.8556) (0.32,0.8704) (0.33,0.8844) (0.34,0.8976) (0.35,0.91) (0.36,0.9216) (0.37,0.9324) (0.38,0.9424) (0.39,0.9516) (0.4,0.96) (0.41,0.9676) (0.42,0.9744) (0.43,0.9804) (0.44,0.9856) (0.45,0.99) (0.46,0.9936) (0.47,0.9964) (0.48,0.9984) (0.49,0.9996) (0.5,1.)															
	};
		\addlegendentry[align=left]{$\mathcal{R}_{\varepsilon}^n(\bar \W_\gamma)$ $ {\fontsize{3.8}{1}\selectfont \overset{\boldsymbol{\cdot}}{\supseteq}}$ $\mathcal{R}_{\varepsilon}^n(\bar \V_\gamma)$ \\
		 $\mathcal{D}_{\varepsilon}^n(\bar \W_\gamma)$ $ {\fontsize{3.8}{1}\selectfont \overset{\boldsymbol{\cdot}}{\subseteq}}$ $\mathcal{D}_{\varepsilon}^n(\bar \V_\gamma)$ }

	\end{axis} 	



\end{tikzpicture} }
    \subfloat[$\gamma=0.2$]{  \begin{tikzpicture}
	\begin{axis}[
		height=6cm,
		width=8cm,
		xlabel=$\alpha$,
		ylabel=$\beta$,
		xmin=0,
		xmax=0.5,
		ymax=1.0,
		ymin=0,
		legend style={at={(0.655,0.52)},anchor=north,legend cell align=left,font=\scriptsize} 
	]


\addplot [draw=none,smooth,fill=ForestGreen!05,forget plot]coordinates {
(0.,0.556) (0.01,0.579) (0.02,0.605) (0.03,0.629) (0.04,0.651) (0.05,0.672) (0.06,0.692) (0.07,0.711) (0.08,0.729) (0.09,0.745) (0.1,0.761) (0.11,0.776) (0.12,0.790) (0.13,0.804) (0.14,0.816) (0.15,0.828) (0.16,0.840) (0.17,0.851) (0.18,0.861) (0.19,0.871) (0.2,0.880) (0.21,0.889) (0.22,0.897) (0.23,0.905) (0.24,0.913) (0.25,0.920) (0.26,0.927) (0.27,0.933) (0.28,0.939) (0.29,0.945) (0.3,0.950) (0.31,0.955) (0.32,0.960) (0.33,0.965) (0.34,0.969) (0.35,0.973) (0.36,0.976) (0.37,0.980) (0.38,0.983) (0.39,0.986) (0.4,0.988) (0.41,0.991) (0.42,0.993) (0.43,0.9941) (0.44,0.9957) (0.45,0.9970) (0.46,0.9981) (0.47,0.9990) (0.48,0.9999) (0.49,1) (0.5,1.)	(0,1) (0.,0.556)
};

\addplot [draw=none,smooth,fill=blue!05,forget plot]coordinates {
(0.,0.) (0.01,0.0396) (0.02,0.0784) (0.03,0.1164) (0.04,0.1536) (0.05,0.19) (0.06,0.2256) (0.07,0.2604) (0.08,0.2944) (0.09,0.3276) (0.1,0.36) (0.11,0.3916) (0.12,0.4224) (0.13,0.4524) (0.14,0.4816) (0.15,0.51) (0.16,0.5376) (0.17,0.5644) (0.18,0.5904) (0.19,0.6156) (0.2,0.64) (0.21,0.6636) (0.22,0.6864) (0.23,0.7084) (0.24,0.7296) (0.25,0.75) (0.26,0.7696) (0.27,0.7884) (0.28,0.8064) (0.29,0.8236) (0.3,0.84) (0.31,0.8556) (0.32,0.8704) (0.33,0.8844) (0.34,0.8976) (0.35,0.91) (0.36,0.9216) (0.37,0.9324) (0.38,0.9424) (0.39,0.9516) (0.4,0.96) (0.41,0.9676) (0.42,0.9744) (0.43,0.9804) (0.44,0.9856) (0.45,0.99) (0.46,0.9936) (0.47,0.9964) (0.48,0.9984) (0.49,0.9996) (0.5,1.) (0.5,0) (0,0)	
};


		\addplot[ForestGreen,thick,smooth] coordinates {
(0.,0.556) (0.01,0.579) (0.02,0.605) (0.03,0.629) (0.04,0.651) (0.05,0.672) (0.06,0.692) (0.07,0.711) (0.08,0.729) (0.09,0.745) (0.1,0.761) (0.11,0.776) (0.12,0.790) (0.13,0.804) (0.14,0.816) (0.15,0.828) (0.16,0.840) (0.17,0.851) (0.18,0.861) (0.19,0.871) (0.2,0.880) (0.21,0.889) (0.22,0.897) (0.23,0.905) (0.24,0.913) (0.25,0.920) (0.26,0.927) (0.27,0.933) (0.28,0.939) (0.29,0.945) (0.3,0.950) (0.31,0.955) (0.32,0.960) (0.33,0.965) (0.34,0.969) (0.35,0.973) (0.36,0.976) (0.37,0.980) (0.38,0.983) (0.39,0.986) (0.4,0.988) (0.41,0.991) (0.42,0.993) (0.43,0.9941) (0.44,0.9957) (0.45,0.9970) (0.46,0.9981) (0.47,0.9990) (0.48,0.9999) (0.49,1) (0.5,1.)									
	};
	\addlegendentry{$\mathcal{R}_{\varepsilon}^n(\bar \W_\gamma) \subseteq \mathcal{R}_{\varepsilon}^n(\bar \V_\gamma)$}



		\addplot[red,smooth,thick] coordinates {
(0.,0)(0.01,0.0440092)(0.02,0.0866269)(0.03,0.127927)(0.04,0.167963)(0.05,0.206775)(0.06,0.2444)(0.07,0.280869)(0.08,0.316212)(0.09,0.350457)(0.1,0.383629)(0.11,0.415754)(0.12,0.446854)(0.13,0.476951)(0.14,0.506064)(0.15,0.534214)(0.16,0.561417)(0.17,0.587692)(0.18,0.613053)(0.19,0.637517)(0.2,0.661097)(0.21,0.683808)(0.22,0.70566)(0.23,0.726667)(0.24,0.746839)(0.25,0.766188)(0.26,0.784722)(0.27,0.802451)(0.28,0.819384)(0.29,0.835529)(0.3,0.850893)(0.31,0.865482)(0.32,0.879303)(0.33,0.892362)(0.34,0.904662)(0.35,0.916207)(0.36,0.927003)(0.37,0.93705)(0.38,0.946352)(0.39,0.954911)(0.4,0.962726)(0.41,0.969799)(0.42,0.9785257)(0.43,0.9825515)(0.44,0.9891258)(0.45,0.992999)(0.46,0.995999)(0.47,0.998999)(0.48,0.999999)(0.49,0.999999)(0.5,1.)									
	};
	\addlegendentry{\!$|\mathcal{R}_{\varepsilon}^n(\bar \W_\gamma)\! \cap \! \mathcal{D}_{\varepsilon}^n(\bar \V_\gamma)|\! = \!\Omega(n)$\!\!\!}

	\addplot[black,thick,smooth] coordinates {
(0.,0)(0.01,0.0425749)(0.02,0.0840078)(0.03,0.124323)(0.04,0.163542)(0.05,0.201686)(0.06,0.238776)(0.07,0.274829)(0.08,0.309865)(0.09,0.343899)(0.1,0.376947)(0.11,0.409024)(0.12,0.440144)(0.13,0.470322)(0.14,0.499568)(0.15,0.527896)(0.16,0.555317)(0.17,0.581842)(0.18,0.60748)(0.19,0.632242)(0.2,0.656137)(0.21,0.679174)(0.22,0.70136)(0.23,0.722705)(0.24,0.743214)(0.25,0.762896)(0.26,0.781757)(0.27,0.799803)(0.28,0.81704)(0.29,0.833475)(0.3,0.849111)(0.31,0.863955)(0.32,0.87801)(0.33,0.891282)(0.34,0.903775)(0.35,0.915491)(0.36,0.926436)(0.37,0.936611)(0.38,0.946021)(0.39,0.954668)(0.4,0.962555)(0.41,0.969683)(0.42,0.976056)(0.43,0.981674)(0.44,0.986541)(0.45,0.990656)(0.46,0.994021)(0.47,0.996637)(0.48,0.998506)(0.49,0.999626)(0.5,1.)									
	};
	\addlegendentry{$I(\bar \W_{\gamma})=I(\bar \V_{\gamma})$ }

		\addplot[blue,thick,smooth] coordinates {
(0.,0.) (0.01,0.0396) (0.02,0.0784) (0.03,0.1164) (0.04,0.1536) (0.05,0.19) (0.06,0.2256) (0.07,0.2604) (0.08,0.2944) (0.09,0.3276) (0.1,0.36) (0.11,0.3916) (0.12,0.4224) (0.13,0.4524) (0.14,0.4816) (0.15,0.51) (0.16,0.5376) (0.17,0.5644) (0.18,0.5904) (0.19,0.6156) (0.2,0.64) (0.21,0.6636) (0.22,0.6864) (0.23,0.7084) (0.24,0.7296) (0.25,0.75) (0.26,0.7696) (0.27,0.7884) (0.28,0.8064) (0.29,0.8236) (0.3,0.84) (0.31,0.8556) (0.32,0.8704) (0.33,0.8844) (0.34,0.8976) (0.35,0.91) (0.36,0.9216) (0.37,0.9324) (0.38,0.9424) (0.39,0.9516) (0.4,0.96) (0.41,0.9676) (0.42,0.9744) (0.43,0.9804) (0.44,0.9856) (0.45,0.99) (0.46,0.9936) (0.47,0.9964) (0.48,0.9984) (0.49,0.9996) (0.5,1.)																
	};
		\addlegendentry[align=left]{$\mathcal{R}_{\varepsilon}^n(\bar \W_\gamma)$ $ {\fontsize{3.8}{1}\selectfont \overset{\boldsymbol{\cdot}}{\supseteq}}$ $\mathcal{R}_{\varepsilon}^n(\bar \V_\gamma)$ \\
		 $\mathcal{D}_{\varepsilon}^n(\bar \W_\gamma)$ $ {\fontsize{3.8}{1}\selectfont \overset{\boldsymbol{\cdot}}{\subseteq}}$ $\mathcal{D}_{\varepsilon}^n(\bar \V_\gamma)$ }



	\end{axis} 

\end{tikzpicture}}\\
 \hspace{-10mm}   
     \subfloat[$\gamma=0.3$]{   \begin{tikzpicture}
	\begin{axis}[
		height=6cm,
		width=8cm,
		xlabel=$\alpha$,
		ylabel=$\beta$,
		xmin=0,
		xmax=0.5,
		ymax=1.0,
		ymin=0,
		legend style={at={(0.655,0.52)},anchor=north,legend cell align=left,font=\scriptsize} 
	]


\addplot [draw=none,smooth,fill=ForestGreen!05,forget plot]coordinates {
(0.,0.734) (0.01,0.748) (0.02,0.762) (0.03,0.775) (0.04,0.787) (0.05,0.799) (0.06,0.810) (0.07,0.820) (0.08,0.831) (0.09,0.840) (0.1,0.849) (0.11,0.858) (0.12,0.867) (0.13,0.875) (0.14,0.882) (0.15,0.890) (0.16,0.897) (0.17,0.904) (0.18,0.910) (0.19,0.916) (0.2,0.922) (0.21,0.928) (0.22,0.933) (0.23,0.938) (0.24,0.943) (0.25,0.947) (0.26,0.952) (0.27,0.956) (0.28,0.960) (0.29,0.964) (0.3,0.967) (0.31,0.971) (0.32,0.974) (0.33,0.977) (0.34,0.980) (0.35,0.982) (0.36,0.985) (0.37,0.987) (0.38,0.989) (0.39,0.991) (0.4,0.992) (0.41,0.994) (0.42,0.9949) (0.43,0.9961) (0.44,0.9972) (0.45,0.9980) (0.46,0.9996) (0.47,0.9999) (0.48,1) (0.49,1) (0.5,1.) (0,1) (0.,0.734)
};

\addplot [draw=none,smooth,fill=blue!05,forget plot]coordinates {
(0.,0.) (0.01,0.0396) (0.02,0.0784) (0.03,0.1164) (0.04,0.1536) (0.05,0.19) (0.06,0.2256) (0.07,0.2604) (0.08,0.2944) (0.09,0.3276) (0.1,0.36) (0.11,0.3916) (0.12,0.4224) (0.13,0.4524) (0.14,0.4816) (0.15,0.51) (0.16,0.5376) (0.17,0.5644) (0.18,0.5904) (0.19,0.6156) (0.2,0.64) (0.21,0.6636) (0.22,0.6864) (0.23,0.7084) (0.24,0.7296) (0.25,0.75) (0.26,0.7696) (0.27,0.7884) (0.28,0.8064) (0.29,0.8236) (0.3,0.84) (0.31,0.8556) (0.32,0.8704) (0.33,0.8844) (0.34,0.8976) (0.35,0.91) (0.36,0.9216) (0.37,0.9324) (0.38,0.9424) (0.39,0.9516) (0.4,0.96) (0.41,0.9676) (0.42,0.9744) (0.43,0.9804) (0.44,0.9856) (0.45,0.99) (0.46,0.9936) (0.47,0.9964) (0.48,0.9984) (0.49,0.9996) (0.5,1.) (0.5,0) (0,0)	
};


		\addplot[ForestGreen,thick,smooth] coordinates {
(0.,0.734) (0.01,0.748) (0.02,0.762) (0.03,0.775) (0.04,0.787) (0.05,0.799) (0.06,0.810) (0.07,0.820) (0.08,0.831) (0.09,0.840) (0.1,0.849) (0.11,0.858) (0.12,0.867) (0.13,0.875) (0.14,0.882) (0.15,0.890) (0.16,0.897) (0.17,0.904) (0.18,0.910) (0.19,0.916) (0.2,0.922) (0.21,0.928) (0.22,0.933) (0.23,0.938) (0.24,0.943) (0.25,0.947) (0.26,0.952) (0.27,0.956) (0.28,0.960) (0.29,0.964) (0.3,0.967) (0.31,0.971) (0.32,0.974) (0.33,0.977) (0.34,0.980) (0.35,0.982) (0.36,0.985) (0.37,0.987) (0.38,0.989) (0.39,0.991) (0.4,0.992) (0.41,0.994) (0.42,0.9949) (0.43,0.9961) (0.44,0.9972) (0.45,0.9980) (0.46,0.9996) (0.47,0.9999) (0.48,1) (0.49,1) (0.5,1.)									
	};
	\addlegendentry{$\mathcal{R}_{\varepsilon}^n(\bar \W_\gamma) \subseteq \mathcal{R}_{\varepsilon}^n(\bar \V_\gamma)$}



		\addplot[red,smooth,thick] coordinates {
(0.,0)(0.01,0.0411076)(0.02,0.0812429)(0.03,0.120417)(0.04,0.158641)(0.05,0.195924)(0.06,0.232276)(0.07,0.267706)(0.08,0.302222)(0.09,0.335833)(0.1,0.368547)(0.11,0.400372)(0.12,0.431315)(0.13,0.461383)(0.14,0.490583)(0.15,0.518921)(0.16,0.546404)(0.17,0.573038)(0.18,0.598827)(0.19,0.623779)(0.2,0.647897)(0.21,0.671187)(0.22,0.693654)(0.23,0.715301)(0.24,0.736134)(0.25,0.756155)(0.26,0.775369)(0.27,0.793779)(0.28,0.811389)(0.29,0.828201)(0.3,0.844217)(0.31,0.859442)(0.32,0.873876)(0.33,0.887522)(0.34,0.900382)(0.35,0.912458)(0.36,0.92375)(0.37,0.93426)(0.38,0.943989)(0.39,0.952938)(0.4,0.961107)(0.41,0.968498)(0.42,0.975093)(0.43,0.980934)(0.44,0.985994)(0.45,0.990275)(0.46,0.993777)(0.47,0.9965)(0.48,0.998444)(0.49,0.999611)(0.5,1.)		};
	\addlegendentry{\!$|\mathcal{R}_{\varepsilon}^n(\bar \W_\gamma)\! \cap \! \mathcal{D}_{\varepsilon}^n(\bar \V_\gamma)|\! = \!\Omega(n)$\!\!\!}


	\addplot[black,thick,smooth] coordinates {
(0.,0)(0.01,0.0407277)(0.02,0.0805363)(0.03,0.119433)(0.04,0.157423)(0.05,0.194514)(0.06,0.230711)(0.07,0.26602)(0.08,0.300447)(0.09,0.333998)(0.1,0.366676)(0.11,0.398488)(0.12,0.429439)(0.13,0.459532)(0.14,0.488773)(0.15,0.517166)(0.16,0.544715)(0.17,0.571423)(0.18,0.597296)(0.19,0.622336)(0.2,0.646548)(0.21,0.669934)(0.22,0.692498)(0.23,0.714243)(0.24,0.735173)(0.25,0.755289)(0.26,0.774595)(0.27,0.793094)(0.28,0.810787)(0.29,0.827678)(0.3,0.843768)(0.31,0.859061)(0.32,0.873557)(0.33,0.887258)(0.34,0.900167)(0.35,0.912286)(0.36,0.923615)(0.37,0.934157)(0.38,0.943912)(0.39,0.952882)(0.4,0.961068)(0.41,0.968472)(0.42,0.975093)(0.43,0.980934)(0.44,0.985994)(0.45,0.990275)(0.46,0.993777)(0.47,0.9965)(0.48,0.998444)(0.49,0.999611)(0.5,1.)								
	};
	\addlegendentry{$I(\bar \W_{\gamma})=I(\bar \V_{\gamma})$}

		\addplot[blue,thick,smooth] coordinates {
(0.,0.) (0.01,0.0396) (0.02,0.0784) (0.03,0.1164) (0.04,0.1536) (0.05,0.19) (0.06,0.2256) (0.07,0.2604) (0.08,0.2944) (0.09,0.3276) (0.1,0.36) (0.11,0.3916) (0.12,0.4224) (0.13,0.4524) (0.14,0.4816) (0.15,0.51) (0.16,0.5376) (0.17,0.5644) (0.18,0.5904) (0.19,0.6156) (0.2,0.64) (0.21,0.6636) (0.22,0.6864) (0.23,0.7084) (0.24,0.7296) (0.25,0.75) (0.26,0.7696) (0.27,0.7884) (0.28,0.8064) (0.29,0.8236) (0.3,0.84) (0.31,0.8556) (0.32,0.8704) (0.33,0.8844) (0.34,0.8976) (0.35,0.91) (0.36,0.9216) (0.37,0.9324) (0.38,0.9424) (0.39,0.9516) (0.4,0.96) (0.41,0.9676) (0.42,0.9744) (0.43,0.9804) (0.44,0.9856) (0.45,0.99) (0.46,0.9936) (0.47,0.9964) (0.48,0.9984) (0.49,0.9996) (0.5,1.)																
	};
		\addlegendentry[align=left]{$\mathcal{R}_{\varepsilon}^n(\bar \W_\gamma)$ $ {\fontsize{3.8}{1}\selectfont \overset{\boldsymbol{\cdot}}{\supseteq}}$ $\mathcal{R}_{\varepsilon}^n(\bar \V_\gamma)$ \\
		 $\mathcal{D}_{\varepsilon}^n(\bar \W_\gamma)$ $ {\fontsize{3.8}{1}\selectfont \overset{\boldsymbol{\cdot}}{\subseteq}}$ $\mathcal{D}_{\varepsilon}^n(\bar \V_\gamma)$ }



	\end{axis} 

\end{tikzpicture} }
    \subfloat[$\gamma=0.4$]{  \begin{tikzpicture}
	\begin{axis}[
		height=6cm,
		width=8cm,
		xlabel=$\alpha$,
		ylabel=$\beta$,
		xmin=0,
		xmax=0.5,
		ymax=1.0,
		ymin=0,
		legend style={at={(0.655,0.52)},anchor=north,legend cell align=left,font=\scriptsize} 
	]


\addplot [draw=none,smooth,fill=ForestGreen!05,forget plot]coordinates {
(0.,0.885) (0.01,0.890) (0.02,0.895) (0.03,0.900) (0.04,0.905) (0.05,0.909) (0.06,0.914) (0.07,0.918) (0.08,0.922) (0.09,0.927) (0.1,0.930) (0.11,0.934) (0.12,0.938) (0.13,0.941) (0.14,0.945) (0.15,0.948) (0.16,0.951) (0.17,0.954) (0.18,0.957) (0.19,0.960) (0.2,0.963) (0.21,0.965) (0.22,0.968) (0.23,0.970) (0.24,0.972) (0.25,0.975) (0.26,0.977) (0.27,0.979) (0.28,0.981) (0.29,0.982) (0.3,0.984) (0.31,0.986) (0.32,0.987) (0.33,0.989) (0.34,0.990) (0.35,0.991) (0.36,0.993) (0.37,0.994) (0.38,0.995) (0.39,0.996) (0.4,0.997) (0.41,0.9982) (0.42,0.9997) (0.43,0.9999) (0.44,1) (0.45,1) (0.46,1) (0.47,1) (0.48,1) (0.49,1) (0.5,1.)	 (0,1) (0.,0.885)
};

\addplot [draw=none,smooth,fill=blue!05,forget plot]coordinates {
(0.,0.) (0.01,0.0396) (0.02,0.0784) (0.03,0.1164) (0.04,0.1536) (0.05,0.19) (0.06,0.2256) (0.07,0.2604) (0.08,0.2944) (0.09,0.3276) (0.1,0.36) (0.11,0.3916) (0.12,0.4224) (0.13,0.4524) (0.14,0.4816) (0.15,0.51) (0.16,0.5376) (0.17,0.5644) (0.18,0.5904) (0.19,0.6156) (0.2,0.64) (0.21,0.6636) (0.22,0.6864) (0.23,0.7084) (0.24,0.7296) (0.25,0.75) (0.26,0.7696) (0.27,0.7884) (0.28,0.8064) (0.29,0.8236) (0.3,0.84) (0.31,0.8556) (0.32,0.8704) (0.33,0.8844) (0.34,0.8976) (0.35,0.91) (0.36,0.9216) (0.37,0.9324) (0.38,0.9424) (0.39,0.9516) (0.4,0.96) (0.41,0.9676) (0.42,0.9744) (0.43,0.9804) (0.44,0.9856) (0.45,0.99) (0.46,0.9936) (0.47,0.9964) (0.48,0.9984) (0.49,0.9996) (0.5,1.) (0.5,0) (0,0)	
};


		\addplot[ForestGreen,thick,smooth] coordinates {
(0.,0.885) (0.01,0.890) (0.02,0.895) (0.03,0.900) (0.04,0.905) (0.05,0.909) (0.06,0.914) (0.07,0.918) (0.08,0.922) (0.09,0.927) (0.1,0.930) (0.11,0.934) (0.12,0.938) (0.13,0.941) (0.14,0.945) (0.15,0.948) (0.16,0.951) (0.17,0.954) (0.18,0.957) (0.19,0.960) (0.2,0.963) (0.21,0.965) (0.22,0.968) (0.23,0.970) (0.24,0.972) (0.25,0.975) (0.26,0.977) (0.27,0.979) (0.28,0.981) (0.29,0.982) (0.3,0.984) (0.31,0.986) (0.32,0.987) (0.33,0.989) (0.34,0.990) (0.35,0.991) (0.36,0.993) (0.37,0.994) (0.38,0.995) (0.39,0.996) (0.4,0.997) (0.41,0.9982) (0.42,0.9997) (0.43,0.9999) (0.44,1) (0.45,1) (0.46,1) (0.47,1) (0.48,1) (0.49,1) (0.5,1.)									
	};
	\addlegendentry{$\mathcal{R}_{\varepsilon}^n(\bar \W_\gamma) \subseteq \mathcal{R}_{\varepsilon}^n(\bar \V_\gamma)$}		




		\addplot[red,smooth,thick] coordinates {
(0.,0)(0.01,0.0399028)(0.02,0.0789727)(0.03,0.117212)(0.04,0.154621)(0.05,0.191204)(0.06,0.22696)(0.07,0.261893)(0.08,0.296003)(0.09,0.329292)(0.1,0.361762)(0.11,0.393414)(0.12,0.424249)(0.13,0.454269)(0.14,0.483475)(0.15,0.511868)(0.16,0.53945)(0.17,0.566221)(0.18,0.592183)(0.19,0.617337)(0.2,0.641684)(0.21,0.665224)(0.22,0.68796)(0.23,0.70989)(0.24,0.731017)(0.25,0.751341)(0.26,0.770863)(0.27,0.789584)(0.28,0.807503)(0.29,0.824623)(0.3,0.840943)(0.31,0.856463)(0.32,0.871186)(0.33,0.88511)(0.34,0.898236)(0.35,0.910565)(0.36,0.922097)(0.37,0.932832)(0.38,0.942771)(0.39,0.952438)(0.4,0.960261)(0.41,0.967813)(0.42,0.974568)(0.43,0.980529)(0.44,0.985696)(0.45,0.990067)(0.46,0.993643)(0.47,0.996424)(0.48,0.998411)(0.49,0.999603)(0.5,1)	};
	\addlegendentry{\!$|\mathcal{R}_{\varepsilon}^n(\bar \W_\gamma)\! \cap \! \mathcal{D}_{\varepsilon}^n(\bar \V_\gamma)|\! = \!\Omega(n)$\!\!\!}	

	
	\addplot[black,thick,smooth] coordinates {
(0.,0)(0.01,0.03986)(0.02,0.0788936)(0.03,0.117102)(0.04,0.154487)(0.05,0.191049)(0.06,0.226791)(0.07,0.261712)(0.08,0.295814)(0.09,0.329099)(0.1,0.361567)(0.11,0.393219)(0.12,0.424057)(0.13,0.454082)(0.14,0.483294)(0.15,0.511695)(0.16,0.539286)(0.17,0.566066)(0.18,0.592038)(0.19,0.617203)(0.2,0.64156)(0.21,0.665111)(0.22,0.687856)(0.23,0.709797)(0.24,0.730934)(0.25,0.751267)(0.26,0.770798)(0.27,0.789527)(0.28,0.807454)(0.29,0.824581)(0.3,0.840907)(0.31,0.856433)(0.32,0.871161)(0.33,0.885089)(0.34,0.89822)(0.35,0.910552)(0.36,0.922087)(0.37,0.932825)(0.38,0.942766)(0.39,0.95191)(0.4,0.960259)(0.41,0.967811)(0.42,0.974568)(0.43,0.980529)(0.44,0.985696)(0.45,0.990067)(0.46,0.993643)(0.47,0.996424)(0.48,0.998411)(0.49,0.999603)(0.5,1)							
	};
	\addlegendentry{$I(\bar \W_{\gamma})=I(\bar \V_{\gamma})$}

		\addplot[blue,thick,smooth] coordinates {
(0.,0.) (0.01,0.0396) (0.02,0.0784) (0.03,0.1164) (0.04,0.1536) (0.05,0.19) (0.06,0.2256) (0.07,0.2604) (0.08,0.2944) (0.09,0.3276) (0.1,0.36) (0.11,0.3916) (0.12,0.4224) (0.13,0.4524) (0.14,0.4816) (0.15,0.51) (0.16,0.5376) (0.17,0.5644) (0.18,0.5904) (0.19,0.6156) (0.2,0.64) (0.21,0.6636) (0.22,0.6864) (0.23,0.7084) (0.24,0.7296) (0.25,0.75) (0.26,0.7696) (0.27,0.7884) (0.28,0.8064) (0.29,0.8236) (0.3,0.84) (0.31,0.8556) (0.32,0.8704) (0.33,0.8844) (0.34,0.8976) (0.35,0.91) (0.36,0.9216) (0.37,0.9324) (0.38,0.9424) (0.39,0.9516) (0.4,0.96) (0.41,0.9676) (0.42,0.9744) (0.43,0.9804) (0.44,0.9856) (0.45,0.99) (0.46,0.9936) (0.47,0.9964) (0.48,0.9984) (0.49,0.9996) (0.5,1.)																
	};
		\addlegendentry[align=left]{$\mathcal{R}_{\varepsilon}^n(\bar \W_\gamma)$ $ {\fontsize{3.8}{1}\selectfont \overset{\boldsymbol{\cdot}}{\supseteq}}$ $\mathcal{R}_{\varepsilon}^n(\bar \V_\gamma)$ \\
		 $\mathcal{D}_{\varepsilon}^n(\bar \W_\gamma)$ $ {\fontsize{3.8}{1}\selectfont \overset{\boldsymbol{\cdot}}{\subseteq}}$ $\mathcal{D}_{\varepsilon}^n(\bar \V_\gamma)$ }

	\end{axis} 	



\end{tikzpicture}}\\
    \caption{Alignment for $\BSC$/$\BEC$ broadcast channel. This figure is an overview about the alignment properties of the polarized sets for a superposition coding scheme that achieves the capacity region of a $\BSC(\alpha)/\BEC(\beta)$ broadcast channel. As discussed above the optimal preprocessing $P_{X|U}$ is a $\BSC(\gamma)$ which defines the channels $\bar \W_\gamma := \BSC(\alpha) \circ \BSC(\gamma)$ and $\bar \V_\gamma:=\BEC(\beta)\circ \BSC(\gamma)$. The graphic depicts alignment results for these two channels for different values of $\gamma$. Recall that the case $\gamma=0$ has been analyzed in Figure~\ref{fig:BSCBEC}. The blue are shows the region where $\bar \V_{\gamma}$ is less noisy than $\bar \W_{\gamma}$, i.e., $0\leq \beta \leq 4\alpha(1-\alpha)$, which implies that $\mathcal{R}_{\varepsilon}^n(\bar \W_{\gamma})$ $\scriptsize{\overset{\boldsymbol{\cdot}}{\supseteq}}$ $\mathcal{R}_{\varepsilon}^n(\bar \V_{\gamma})$ and $\mathcal{D}_{\varepsilon}^n(\bar \W_{\gamma})$ $\scriptsize{\overset{\boldsymbol{\cdot}}{\subseteq}}$ $\mathcal{D}_{\varepsilon}^n(\bar \V_{\gamma})$. The green area depicts scenarios where $\mathcal{R}_{\varepsilon}^n(\bar \W_{\gamma}) \subseteq \mathcal{R}_{\varepsilon}^n(\bar \V_{\gamma})$ that are determined by evaluating the conditions given in Theorem~\ref{thm:UB} at level $2$. The conditions in Theorem~\ref{thm:LB} evaluated for level $4$ determine a region (ploted in red) where there is no proper alignment, i.e., $|\mathcal{R}_{\varepsilon}^n(\bar \W_{\gamma}) \cap \mathcal{D}_{\varepsilon}^n(\bar \V_{\gamma})|= \Omega(n)$.
    }
    \label{fig:BCoverview}
\end{figure}
\section{Entanglement assistance for quantum polar codes} \label{sec:entanglement}
Suppose we are given a quantum channel $\Phi:\mathcal{S}(\mathcal{H}) \to \mathcal{S}(\mathcal{H})$ and would like to use it to transmit quantum information. Consider a fixed orthonormal basis for the qubits at the input, call it \emph{amplitude basis}, which induces a classical-quantum channel $\W^{(A)}:\mathcal{X} \to \mathcal{D}(\mathcal{H})$, i.e., a channel that maps a classical input to a quantum mechanical output. Fixing a complementary basis at the input, call it \emph{phase basis}, induces another classical-quantum channel $\W^{(P)}:\mathcal{X} \to \mathcal{D}(\mathcal{H})$.
The central insight of \cite{renes12} is that the polarization phenomenon occurs simultaneously for $\W^{(A)}$ and $\W^{(P)}$ which allows us to reliably transmit quantum information over a noisy quantum channel. Note that the polarization in the phase basis occurs in the reversed order as the polarization in the amplitude basis. A more detailed discussion can be found in \cite{renes12}.

When considering $n$ copies of the original channel $\Phi$, this induces as explained above two classical-quantum channels $\W^{(A),n}$ and $\W^{(P),n}$ which polarize simultaneously. For $\varepsilon \in(0,1)$, four polarized sets can be defined as
\begin{subequations}
\begin{align}
\mathcal{Q}_{\varepsilon}^n(\Phi)&:=\left \lbrace i \in [n]: F(\W_{b(i-1)}^{(A),n})\leq \varepsilon \, \wedge \, F(\W_{\bar b(i-1)}^{(P),n})\leq \varepsilon \right \rbrace\\
\mathcal{A}_{\varepsilon}^n(\Phi)&:=\left \lbrace i \in [n]: F(\W_{b(i-1)}^{(A),n})\geq 1-\varepsilon \, \wedge \, F(\W_{\bar b(i-1)}^{(P),n})\leq \varepsilon \right \rbrace\\
\mathcal{P}_{\varepsilon}^n(\Phi)&:=\left \lbrace i \in [n]: F(\W_{b(i-1)}^{(A),n})\leq \varepsilon \, \wedge \, F(\W_{\bar b(i-1)}^{(P),n})\geq 1-\varepsilon \right \rbrace\\
\mathcal{E}_{\varepsilon}^n(\Phi)&:=\left \lbrace i \in [n]: F(\W_{b(i-1)}^{(A),n})\geq 1-\varepsilon \, \wedge \, F(\W_{\bar b(i-1)}^{(P),n})\geq 1- \varepsilon \right \rbrace,
\end{align}
\end{subequations}
where $b(i)$ for $i \in [n]$ denotes the binary representation of the integer $i$ with $\log n$ bits. As mentioned before, for most applications it is convenient to choose $\varepsilon$ as small as possible which is $\varepsilon =O(2^{-n^{\nu}})$ for $\nu < \frac{1}{2}$. 
 $\mathcal{Q}_{\varepsilon}^n(\Phi)$ denotes the set of synthesized channels that are good in both bases. The set $\mathcal{A}_{\varepsilon}^n(\Phi)$ contains synthesized channels that are bad in the amplitude and good in the phase basis and $\mathcal{P}_{\varepsilon}^n(\Phi)$ characterizes the synthesized channels that are bad in the phase and good in the amplitude basis. Finally the set $\mathcal{E}_{\varepsilon}^n(\Phi)$ contains the indices corresponding to synthesized channels that are bad in both bases. As explained in \cite{renes12}, the inputs characterized by $\mathcal{Q}_{\varepsilon}^n(\Phi)$ are used to send the quantum data and the inputs corresponding to $\mathcal{A}_{\varepsilon}^n(\Phi)$ respectively $\mathcal{P}_{\varepsilon}^n(\Phi)$ are frozen in the amplitude respectively phase basis. The inputs given by $\mathcal{E}_{\varepsilon}^n(\Phi)$ must be entangled with the decoder to ensure proper decoding, which is the reason quantum polar codes are entanglement-assisted codes whenever $|\mathcal{E}_{\varepsilon}^n(\Phi)|=\Omega(n)$. In the following, we introduce two conditions that for an arbitrary quantum channel $\Phi$ can be used to determine if $|\mathcal{E}_{\varepsilon}^n(\Phi)|=0$ or $|\mathcal{E}_{\varepsilon}^n(\Phi)|=\Omega(n)$.
\subsection{Induced channels of qubit Pauli channels} \label{sec:pauliChannels}
A Pauli qubit channel applies a random Pauli operator to its input. In its most general form it can be written as the mapping $\Phi:\mathcal{S}(\mathcal{H}) \to \mathcal{S}(\mathcal{H})$ with $\dim \mathcal{H}=2$ such that $\rho \mapsto \sum_{u,v \in \{0,1\}} p_{u,v} \sigma_X^u \sigma_Z^v\, \rho\, \sigma_Z^v \sigma_X^u$. The corresponding induced amplitude channel $\W^{(A)}:\{0,1\} \to \mathcal{D}(\mathcal{H})$ is described by a $\BSC(p_u)$ and thus $\F{\BSC(p_u)}=2\sqrt{p_0\, p_1}$. The induced phase channel $\W^{(P)}:\{0,1\} \to \mathcal{D}(\mathcal{H})$ is described by the mapping $x\mapsto \Trp{E}{\braket{\psi_x}{\psi_x}^{ABE}}=:\rho_x^{AB}$ with
\begin{align}
\ket{\psi_x}^{ABE}= \frac{1}{\sqrt{2}} \sum_{z=0}^1 (-1)^{xz} \ket{z}_A \sum_{u,v \in \{0,1\}} \sqrt{p_{u,v}}X^u Z^v \ket{z}_B \ket{u,v}_E
\end{align}
being the output state of the circuit depicted in Figure~\ref{fig:phasePauli}. Thus the fidelity of channel $\W^{(P)}$ is given by $\F{\W^{(P)}}=\|\sqrt{\rho_0^{AB}}\sqrt{\rho_1^{AB}}\|_1$.
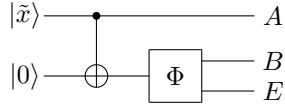
\begin{figure}[!htb]
\centering
\def \x{0.7}
\def \xb{0.7}
\def \y{0.8}
\def \yb{0.35}
\def \s{0.20}
\def \la{0.25}

\def \r{0.15} 

\begin{tikzpicture}[scale=1,auto, node distance=1cm,>=latex']
	

 \draw[] (0,0) -- (3*\x+\xb,0);   
\draw[] (0,-\y) -- (2*\x,-\y); 

\draw[] (2*\x+\xb,-\y-\s) -- (3*\x+\xb,-\y-\s); 
\draw[] (2*\x+\xb,-\y+\s) -- (3*\x+\xb,-\y+\s); 

 \draw[] (2*\x,-\y-\yb) -- (2*\x,-\y+\yb);   
 \draw[] (2*\x+\xb,-\y-\yb) -- (2*\x+\xb,-\y+\yb);  
 \draw[] (2*\x,-\y-\yb) -- (2*\x+\xb,-\y-\yb);  
 \draw[] (2*\x,-\y+\yb) -- (2*\x+\xb,-\y+\yb);
 \node at (2*\x+0.5*\xb,-\y) {$\Phi$};
   
\draw[] (\x,-\y) circle(\r); 
  \fill (\x,0) circle (1.5pt);
\draw[] (\x,0) -- (\x,-\y-\r); 

\node at (-\la,0) {$\ket{\tilde x}$};
\node at (-\la,-\y) {$\ket{0}$};
\node at (3*\x+\xb+\la,-\y+\s) {$B$};
\node at (3*\x+\xb+\la,-\y-\s) {$E$};
\node at (3*\x+\xb+\la,0) {$A$};

\end{tikzpicture}
\caption{The induced phase channel for an arbitrary Pauli channel $\Phi$, with $ \ket{\tilde x}=\tfrac{1}{\sqrt{2}} \sum_{z \in \{0,1 \}} (-1)^{xz} \ket{z}$ for $x \in \{0,1 \}$.}
\label{fig:phasePauli}
\end{figure} 

\subsection{Sufficient conditions for the need of entanglement-assistance}\label{sec:bounds}
Inspired by the techniques presented in Section~\ref{sec:alignmentBounds} we can define sufficient conditions for $|\mathcal{E}_{\varepsilon}^n(\Phi)|$ being small or large on every level of the polarization tree.
\begin{myprop}[Level $k$ condition for $|\mathcal{E}_{\varepsilon}^n(\Phi)|$ being large]\label{prop:Elb}
Let $k \in \mathbb{N}_0$ and $\varepsilon \in (0,1)$. If $I(\W^{(A)}_b) + I(\W^{(P)}_{\bar b}) \leq 1$ for some $b\in \{0,1\}^k$ and for all possible choices of an amplitude basis, then $|\mathcal{E}_{\varepsilon}^n(\Phi)| = \Omega(n)$.
\end{myprop}
\begin{proof}
Let $\mathcal{L}_{\varepsilon}^n (\Phi):=\left \lbrace i \in [n]: F(\W_{b(i)}^{(A),n})\leq \varepsilon \right \rbrace$ and $\mathcal{M}_{\varepsilon}^n (\Phi):=\left \lbrace i \in [n]: F(\W_{\bar b(i)}^{(P),n})\leq \varepsilon \right \rbrace$.
The polarization phenomenon \cite{arikan09,arikan10,wilde_polar_2011} ensures that $n I(\W^{(A)})=|\mathcal{L}_{\varepsilon}^n(\Phi)|-o(n)$ and $n I(\W^{(P)})=|\mathcal{M}_{\varepsilon}^n(\Phi)|-o(n)$. Therefore $I(\W^{(A)}) + I(\W^{(P)}) \leq 1$ directly implies $|\mathcal{E}_{\varepsilon}^n(\Phi)|=\Omega(n)$. Applying the same argument at every level of the polarization tree proves the assertion. Note that as we are free to choose the basis which we call amplitude basis, we have to verify the condition $I(\W^{(A)}_b) + I(\W^{(P)}_{\bar b}) \leq 1$ for all possible choices of this basis.
\end{proof}

\begin{myprop}[Level $k$ condition for $|\mathcal{E}_{\varepsilon}^n(\Phi)|$ being small] \label{prop:Eub}
Let $k \in \mathbb{N}_0$ and $\varepsilon \in (0,1)$. If $\F{\W^{(A)}_b} + \F{\W^{(P)}_{\bar b}} \leq 1$ for all $b\in \{0,1\}^k$, then $|\mathcal{E}_{\varepsilon}^n(\Phi)|=0$.
\end{myprop}
\begin{proof}
Let $n\geq k$ and $d\in \{0,1\}^n$ such that $F(\W_d^{(A)}) \geq 1-\varepsilon$. Corollary~\ref{cor:Hamed} together with the assumption of the proposition implies that $F(\W_{\bar d}^{(P)})\leq \varepsilon$ and thus $|\mathcal{E}_{\varepsilon}^n(\Phi)|=0$. 
\end{proof}

\begin{myremark}[Conditions get stronger for higher levels] 
By the same arguments that are given in Remark~\ref{rmk:LBworse} and Remark~\ref{rmk:UBworse} one can prove that the conditions to conclude if enganglement-assistance is needed or not given in Propositions~\ref{prop:Elb} and \ref{prop:Eub} get stronger by increasing the level, i.e., if their assumption is satisfied for level $k$ it is also satisfied for all levels $\ell \leq k$. 
\end{myremark}

\subsection{Examples}
In this section we show the performance of the conditions given in Propositions~\ref{prop:Elb} and \ref{prop:Eub} on three examples. In addition with these examples, we will prove that there exist both channels for which $|\mathcal{E}_{\varepsilon}^n(\Phi)|=0$ however also channels such that $|\mathcal{E}_{\varepsilon}^n(\Phi)|=\Omega(n)$, implying that quantum polar codes sometimes do and sometimes do not require preshared entanglement.
\begin{myex}[Depolarizing channel]
Consider a qubit depolarizing channel $\Phi:\mathcal{S}(\mathcal{H}) \to \mathcal{S}(\mathcal{H})$ with $\dim \mathcal{H}=2$ that maps $\rho \mapsto (1-p)\rho +\tfrac{p}{3} (\sigma_X \rho\, \sigma_X + \sigma_Z \rho \, \sigma_Z + \sigma_Y\rho \, \sigma_Y)$. The channel coherent information can be computed to be $Q^{(1)}(\Phi)=1+(1-p) \log(1-p) + p \log(\tfrac{p}{3})$ \cite{wilde_book}. As we are interested in a region where the depolarizing channel has a nonnegative channel coherent information,  we can restrict ourselves to $p \in [0,0.18929]$. The conditions given in Proposition~\ref{prop:Eub} for level 0, ensure that for $p$ such that $\F{\W^{(A)}}+\F{W^{(P)}} \leq 1$ we have $|\mathcal{E}_{\varepsilon}^n(\Phi) | = 0$. As explained in Section~\ref{sec:pauliChannels} we have $\F{\W^{(A)}}= 2 \sqrt{\tfrac{2p}{3}(1-\tfrac{2p}{3})}$ and $\F{\W^{(P)}}=\tfrac{2}{3}(p + \sqrt{3}\sqrt{p(1-p)})$ which gives $|\mathcal{E}_{\varepsilon}^n(\Phi)|=0$ if $p \in [0,0.120535]$. We can improve the condition by considering higher levels as discussed in Section~\ref{sec:bounds}, such that for level 2, we obtain that $|\mathcal{E}_{\varepsilon}^n(\Phi)|=0$ if $p \in [0,0.149062]$. The condition of Proposition~\ref{prop:Elb} evaluated for level 3 shows that for $p \in [0.187757,0.18929]$ we have $|\mathcal{E}_{\varepsilon}^n(\Phi)|=\Omega(n)$, i.e., entanglement assistance is needed. Note that due to the symmetry of the depolarizing channel it is sufficient to consider an amplitude basis being $\sigma_Z$. This example disproves the conjecture stated in \cite{renes12} saying that entanglement assistance is not needed for Pauli channels.
\end{myex}
\begin{myex}[BB84 channel]
Consider a qubit Pauli channel with independent bit flip and phase flip error probability where $q_X \in [0,\tfrac{1}{2}]$ denotes the bit flip and $q_Z\in [0,\tfrac{1}{2}]$ the phase flip probability. More formally this is a channel $\Phi:\mathcal{S}(\mathcal{H}) \to \mathcal{S}(\mathcal{H})$ with $\dim \mathcal{H}=2$ that maps $\rho \mapsto (1-q_X-q_Z+q_X q_Z)\rho +(q_X - q_X q_Z )\sigma_X \rho\, \sigma_X + (q_Z-q_Z q_X)\sigma_Z \rho\, \sigma_Z + q_X q_Z \sigma_Y \rho\, \sigma_Y$. The channel coherent information of the BB84 channel is given by $Q^{(1)}(\Phi)=1-\Hb(q_X)-\Hb(q_Z)$ \cite{wilde_book}. As shown in \cite{renes12}, the induced amplitude channel --- when considering $\sigma_Z$ being the amplitude basis --- $\W^{(A)}$ is a $\BSC(q_X)$ and the induced phase channel $\W^{(P)}$ is a $\BSC(q_Z)$. Applying the conditions given in Theorem~\ref{thm:UB} for these two channels allows us to determine a region, that is depicted in Figure~\ref{fig:examples}, where entanglement assistance is not needed. We note that the region where $|\mathcal{E}_{\varepsilon}^n(\Phi)|=0$ evaluated for level $4$ (cf.\ Figure~\ref{fig:examples}) is strictly larger than the level $0$ bound which was already mentioned in \cite{renes12}.
\end{myex}
\begin{myex}[Two-Pauli channel]
Consider a qubit two-Pauli channel which is described by the mapping $\Phi:\mathcal{S}(\mathcal{S}) \to \mathcal{S}(\mathcal{H})$, $\rho \mapsto (1-q_X-q_Z) \rho + q_X \sigma_X \rho\, \sigma_X + q_Z \sigma_Z \rho\, \sigma_Z$ for $q_X,q_Z \in [0,\tfrac{1}{2}]$. Using Proposition~\ref{prop:Eub} a region, that is shown in Figure~\ref{fig:examples}, where no entanglement assistance is needed can be determined. We note that the coherent information of $\Phi$ is given by $Q^{(1)}(\Phi)=1+(1-q_X-q_Z)\log(1-q_X-q_Z) + q_X \log q_X + q_Z \log q_Z$ \cite{wilde_book}.
\end{myex}

\begin{figure}[!htb]
\hspace{-5mm}
    \subfloat[BB84 channel]{   \begin{tikzpicture}
	\begin{axis}[
		height=8cm,
		width=8cm,
		xlabel=$q_X$,
		ylabel=$q_Z$,
		xmin=0,
		xmax=0.5,
		ymax=0.5,
		ymin=0,
		legend style={at={(0.59,0.95)},anchor=north,legend cell align=left} 
	]
		
\addplot [draw=none,smooth,fill=ForestGreen!05,forget plot]coordinates {
(0.,0.5) (0.001,0.364) (0.0025,0.3292) (0.005,0.2970) (0.01,0.2594) (0.015,0.2345) (0.02,0.2155) (0.025,0.1999) (0.03,0.1866) (0.035,0.1750) (0.04,0.1646) (0.045,0.1554) (0.05,0.1469) (0.055,0.1392) (0.06,0.1320) (0.07,0.1192) (0.08,0.1078) (0.09,0.0971) (0.1,0.0872) (0.11,0.0781) (0.12,0.0694) (0.13,0.0615) (0.14,0.0544) (0.15,0.0481) (0.16,0.0424) (0.17,0.0373) (0.18,0.0327) (0.19,0.0286) (0.2,0.0249) (0.21,0.0216) (0.22,0.0187) (0.23,0.0160) (0.24,0.0137) (0.25,0.0117) (0.26,0.0099) (0.27,0.0083) (0.28,0.0069) (0.29,0.0057) (0.3,0.0047) (0.31,0.0038) (0.32,0.003) (0.33,0.0024) (0.34,0.0019) (0.35,0.0014) (0.36,0.0011) (0.37,0.0008) (0.38,0.0006) (0.39,0.0004) (0.4,0.0002) (0.41,0.0001) (0.42,0) (0.43,0) (0.44,0) (0.45,0) (0.46,0) (0.47,0) (0.48,0) (0.49,0) (0.5,0) (0,0) (0,0.5)			
	};

\addplot [draw=none,smooth,fill=black!05,forget plot]coordinates {
(0.,0.5) (0.001,0.364) (0.0025,0.3292) (0.005,0.2970) (0.01,0.2594) (0.015,0.2345) (0.02,0.2155) (0.025,0.1999) (0.03,0.1866) (0.035,0.1750) (0.04,0.1646) (0.045,0.1554) (0.05,0.1469) (0.055,0.1392) (0.06,0.1320) (0.07,0.1192) (0.08,0.1078) (0.09,0.0971) (0.1,0.0872) (0.11,0.0781) (0.12,0.0694) (0.13,0.0615) (0.14,0.0544) (0.15,0.0481) (0.16,0.0424) (0.17,0.0373) (0.18,0.0327) (0.19,0.0286) (0.2,0.0249) (0.21,0.0216) (0.22,0.0187) (0.23,0.0160) (0.24,0.0137) (0.25,0.0117) (0.26,0.0099) (0.27,0.0083) (0.28,0.0069) (0.29,0.0057) (0.3,0.0047) (0.31,0.0038) (0.32,0.003) (0.33,0.0024) (0.34,0.0019) (0.35,0.0014) (0.36,0.0011) (0.37,0.0008) (0.38,0.0006) (0.39,0.0004) (0.4,0.0002) (0.41,0.0001) (0.42,0) (0.43,0) (0.44,0) (0.45,0) (0.46,0) (0.47,0) (0.48,0) (0.49,0) (0.5,0)
(0.49,0.0000166647)(0.48,0.0000763557)(0.47,0.000188027) (0.46,0.0003586)(0.45,0.000594261)(0.44,0.000900853)(0.43,0.00128408) (0.42,0.0017496)(0.41,0.00230315)(0.4,0.00295057) (0.39,0.00369787)(0.38,0.00455129)(0.37,0.00551734) (0.36,0.00660285)(0.35,0.00781501)(0.34,0.00916144)(0.33,0.0106502)(0.32,0.0122901) (0.31,0.0140901)(0.3,0.0160605)(0.29,0.0182118)(0.28,0.0205557)(0.27,0.0231049) (0.26,0.0258732)(0.25,0.0288757)(0.24,0.0321292)(0.23,0.035652) (0.22,0.0394647)(0.21,0.0435901) (0.2,0.0480539)(0.19,0.0528851)(0.18,0.0581167)  (0.17,0.0637863)(0.16,0.0699375)(0.15,0.076621)(0.14,0.0838962)(0.13,0.0918339)(0.12,0.100519) (0.11,0.110056)(0.1,0.120573)(0.09,0.132233)(0.08,0.145248)(0.07,0.159903)(0.06,0.176588)(0.05,0.195876) (0.04,0.218657)(0.03,0.246453) (0.02,0.282295) (0.01,0.334247)(0.009,0.341147)(0.008,0.348563) (0.007,0.356603) (0.006,0.365413)(0.005,0.375205)(0.004,0.386311) (0.003,0.399294)(0.002,0.415272) (0.001,0.437205)(0,0.5)
	};


	\addplot[black,thick,smooth] coordinates {
(0,0.5) (0.001,0.437205) (0.002,0.415272) (0.003,0.399294) (0.004,0.386311) (0.005,0.375205) (0.006,0.365413) (0.007,0.356603) (0.008,0.348563) (0.009,0.341147) (0.01,0.334247) (0.02,0.282295) (0.03,0.246453) (0.04,0.218657) (0.05,0.195876) (0.06,0.176588) (0.07,0.159903) (0.08,0.145248) (0.09,0.132233) (0.1,0.120573) (0.11,0.110056) (0.12,0.100519) (0.13,0.0918339) (0.14,0.0838962) (0.15,0.076621) (0.16,0.0699375) (0.17,0.0637863) (0.18,0.0581167) (0.19,0.0528851) (0.2,0.0480539) (0.21,0.0435901) (0.22,0.0394647) (0.23,0.035652) (0.24,0.0321292) (0.25,0.0288757) (0.26,0.0258732) (0.27,0.0231049) (0.28,0.0205557) (0.29,0.0182118) (0.3,0.0160605) (0.31,0.0140901) (0.32,0.0122901) (0.33,0.0106502) (0.34,0.00916144) (0.35,0.00781501) (0.36,0.00660285) (0.37,0.00551734) (0.38,0.00455129) (0.39,0.00369787) (0.4,0.00295057) (0.41,0.00230315) (0.42,0.0017496) (0.43,0.00128408) (0.44,0.000900853) (0.45,0.000594261) (0.46,0.0003586) (0.47,0.000188027) (0.48,0.0000763557) (0.49,0.0000166647) (0.5,2.1283*10^-10)											
	};
	\addlegendentry{$Q^{(1)}(\Phi) \geq 0$}

		\addplot[ForestGreen,thick,smooth] coordinates {
(0.,0.5) (0.001,0.364) (0.0025,0.3292) (0.005,0.2970) (0.01,0.2594) (0.015,0.2345) (0.02,0.2155) (0.025,0.1999) (0.03,0.1866) (0.035,0.1750) (0.04,0.1646) (0.045,0.1554) (0.05,0.1469) (0.055,0.1392) (0.06,0.1320) (0.07,0.1192) (0.08,0.1078) (0.09,0.0971) (0.1,0.0872) (0.11,0.0781) (0.12,0.0694) (0.13,0.0615) (0.14,0.0544) (0.15,0.0481) (0.16,0.0424) (0.17,0.0373) (0.18,0.0327) (0.19,0.0286) (0.2,0.0249) (0.21,0.0216) (0.22,0.0187) (0.23,0.0160) (0.24,0.0137) (0.25,0.0117) (0.26,0.0099) (0.27,0.0083) (0.28,0.0069) (0.29,0.0057) (0.3,0.0047) (0.31,0.0038) (0.32,0.003) (0.33,0.0024) (0.34,0.0019) (0.35,0.0014) (0.36,0.0011) (0.37,0.0008) (0.38,0.0006) (0.39,0.0004) (0.4,0.0002) (0.41,0.0001) (0.42,0) (0.43,0) (0.44,0) (0.45,0) (0.46,0) (0.47,0) (0.48,0) (0.49,0) (0.5,0)					
	};
		\addlegendentry{$|\mathcal{E}_{\varepsilon}^n(\Phi)| =0$ (level 4)}	
		
	\end{axis}  
\node[ForestGreen] at (1.1,0.30) {\footnotesize $|\mathcal{E}_{\varepsilon}^n(\Phi)|=0$};

\end{tikzpicture} }
    \subfloat[Two-Pauli channel]{\input{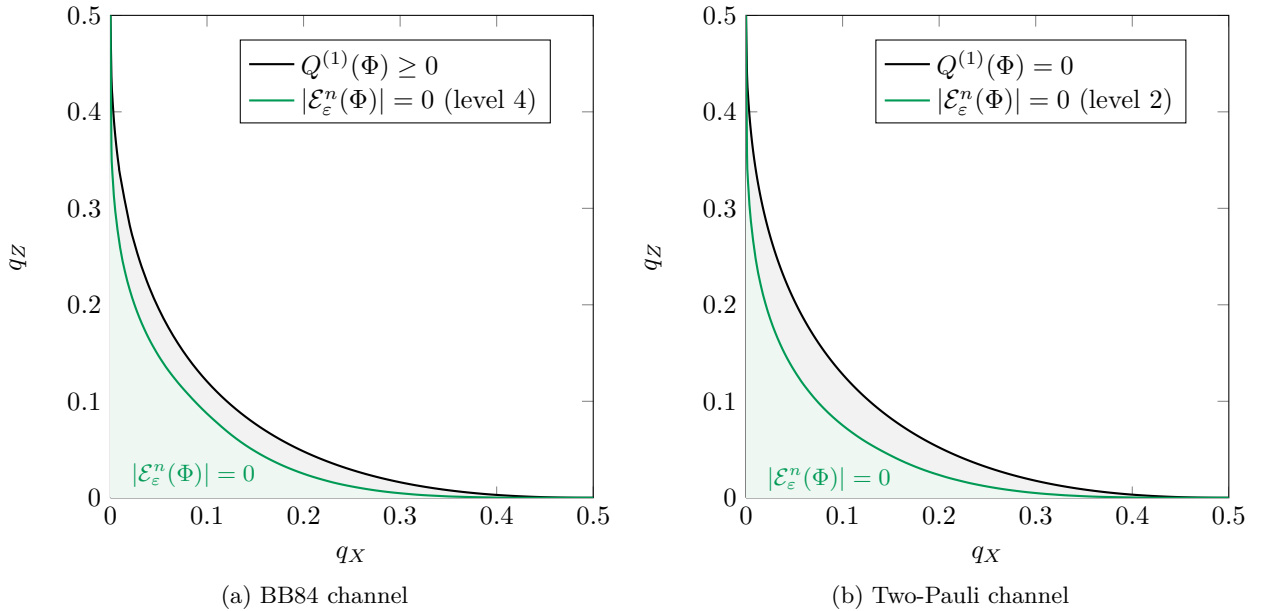}}
    \caption{Regions for a BB84 and a two-Pauli channel where $|\mathcal{E}_{\varepsilon}^n(\Phi)|=0$, i.e., no preshared entanglement is required. The black area shows the region of positive channel coherent information.}
    \label{fig:examples}
\end{figure}
\section{Conclusion} \label{sec:conclusion}
We derived two analytical conditions that can be used to determine the alignment of polarized sets between different DMCs. The condition of Theorem~\ref{thm:LB} that recognizes  situations where there is no alignment (not even essentially) uses a simple counting argument. The condition of Theorem~\ref{thm:UB}, which identifies scenarios where there is an alignment of the polarized sets, is based on the uncertainty principle of quantum mechanics. As the authors are not aware of a purely classical proof for this statement, this seems to display one of the rare incidences where a quantum argument is useful to prove a classical result which is hard to obtain with a purely classical proof technique. We demonstrated on the example of a BSC-BEC pair that the two conditions can be close in the sense that essentially every possible setup can be classified into a proper or improper alignment of the polarized sets.

As we discussed in the main text, it is important to understand the alignment of polarized sets as it is directly related to the universality question of (standard) polar codes. This is oftentimes needed for certain coding tasks, such as network coding problems. Whenever there is the need of a universal polar code, the conditions proposed in this paper will be useful to determine if there is a proper alignment of the polarized sets or not. If there is, the standard coding techniques can be used, if not one could use further universal polarization techniques (as discussed in \cite{hassani13,sasogluLele13}), at the expense of a worse scaling behaviour in the blocklength. Understanding the alignment of polarized sets is not however limited to universality considerations. For example, it could be also used for various (network) coding scenarios where alignment of some polarized sets is needed (see e.g. \cite{vardy11, goela13, mond14}).

Another interesting application of the proposed alignment conditions is to  offer a way to determine if the quantum polar codes introduced in \cite{renes12} do or do not need entanglement assistance. This is particularly relevant when using quantum polar codes in practice as distributing (noiseless) entangled states is difficult. 
For future work, it is of interest to analyze how the conditions of Theorem~\ref{thm:UB} change if one uses different versions of the channel counterpart that are for example more pure. 

\section*{Acknowledgments}
The authors thank Dominik Waldburger for helpful discussions.
JMR and DS were supported by the Swiss National Science Foundation (through the National Centre of Competence in Research `Quantum Science and Technology') and by the European Research Council (grant No.~258932).

\bibliography{./bibtex/header,./bibtex/bibliofile}


\clearpage



\end{document}